\documentclass[11pt,a4paper]{article}
\usepackage{ulem}
\usepackage{amsmath}
\usepackage{amsthm}
\usepackage{amssymb}
\usepackage{answers}
\usepackage{setspace}
\usepackage{graphicx}
\usepackage{enumitem}
\usepackage{multicol}
\usepackage{blkarray}
\usepackage{tikz}
\usetikzlibrary{automata, positioning, arrows}
\usepackage[colorlinks = true,
            linkcolor = blue,
            urlcolor  = blue,
            citecolor = blue,
            anchorcolor = blue,pagebackref=true]{hyperref}
\renewcommand*\backref[1]{\ifx#1\relax \else (cited on #1) \fi}
\usepackage{mathrsfs}
\usepackage{lmodern}
\usepackage{amsmath,amsthm,amssymb}
\usepackage[T1]{fontenc}
\usepackage[noend]{algpseudocode}
\usepackage{comment}
\usepackage{textcomp}
\usepackage[ruled,vlined,linesnumbered]{algorithm2e}
\let\oldnl\nl
\newcommand{\nonl}{\renewcommand{\nl}{\let\nl\oldnl}}
\makeatother

\newcommand{\C}{\mathbb{C}}
\newcommand{\R}{\mathbb{R}}

\newtheorem{theorem}{Theorem}[section]
\newtheorem{corollary}{Corollary}[theorem]
\newtheorem{lemma}[theorem]{Lemma}
\newtheorem{claim}[theorem]{Claim}
\newtheorem{definition}[theorem]{Definition}

\newtheorem{remark}[theorem]{Remark}

\usepackage[utf8]{inputenc}
\usepackage[english]{babel}
\usepackage{graphicx}
\graphicspath{{images/}{../images/}}
\usepackage{blindtext}

\usepackage{subfiles}
\title{Subfile Example}
\author{Team Learn Overleaf}
\date{\vspace{-5ex}} 

\begin{document}
\normalem
\title{
{ Complete Decomposition of Symmetric Tensors} in Linear Time and Polylogarithmic Precision}

\author{Pascal Koiran and Subhayan Saha\thanks{P.K. is with Univ Lyon, EnsL, UCBL, CNRS, LIP.
  Email: firstname.lastname@ens-lyon.fr. S.S. is with Department of Mathematics and Operational Research, 
University of Mons, 
Mons, Belgium, Email: firstname.lastname@umons.ac.be.}}
\maketitle
\begin{abstract}

{\footnotesize We  study
symmetric tensor decompositions, i.e., decompositions of the form 
$T = \sum_{i=1}^r u_i^{\otimes 3}$ where $T$ is a symmetric tensor of order 3 and $u_i \in \mathbb{C}^n$.  In order to obtain efficient decomposition algorithms, it is necessary to require additional properties from the $u_i$. In this paper we assume that the $u_i$ are linearly independent. This implies $r \leq n$, i.e., the decomposition of $T$ is {\em undercomplete}. We will moreover assume that $r=n$ (we plan to extend this work to the case $r<n$ in a forthcoming paper).

We give a randomized algorithm for the following problem: given $T$, an accuracy parameter $\varepsilon$, and an upper bound $B$ on the \textit{condition number} of the tensor, output vectors $u'_i$ such that $||u_i - u'_i|| \leq \varepsilon$ (up to permutation and multiplication by phases) with high probability. The main novel features of our algorithm are:
    \begin{itemize}
        \item We provide the first algorithm for this problem that works in the  computation model of finite arithmetic and requires only poly-logarithmic (in $n, B$ and $\frac{1}{\varepsilon}$) many bits of precision.
        \item { Moreover, this is also} the first algorithm that runs in linear time in the size of the input tensor. It requires $O(n^3)$ arithmetic operations for all accuracy parameters $\varepsilon = \frac{1}{\text{poly}(n)}$.
        \end{itemize}

In order to obtain these results, we rely on a mix of techniques from algorithm design and algorithm analysis.
The algorithm is a modified version of simultaneous diagonalisation algorithm for symmetric tensors. In terms of algorithm design, our main contribution lies in replacing the usual appeal to resolution of a linear system of equations \cite{Kayal11,BCMV13} by a matrix trace-based method. The analysis of the algorithm depends on the following components:
    \begin{enumerate}
        \item We use the fast and numerically stable diagonalisation algorithm from \cite{9317903}. We provide better guarantees for the approximate solution returned by the diagonalisation algorithm when the input matrix is diagonalisable.
        \item We show strong anti-concentration bounds for certain families of polynomials when the randomness is sampled uniformly from a discrete grid.
    \end{enumerate}
\textbf{Keywords:} Tensor Decomposition; Finite precision arithmetic ; simultaneous diagonalisation algorithm ; Computational Linear Algebra
}
\end{abstract}
\newpage
\section{Introduction}
Tensor decompositions have generated significant interest  in  recent years due to their applications in different fields such as signal processing, computer vision, chemometrics, neuroscience and others (see \cite{BK09} for a comprehensive survey on the applications and available software for this problem). In fact, a number of learning algorithms for certain models have been developed through the fundamental machinery of tensor decompositions. Pure topic models (\cite{AHK'12}), blind source separation and independent component analysis (\cite{DELATHAUWER2010155}), Hidden Markov Models (\cite{MR'05},\cite{HKZ'08}), mixture of spherical gaussians (\cite{HK12},\cite{GHK15}), Latent Dirichlet Allocation (\cite{AFH+12}). Numerous algorithms have been devised for solving the tensor decomposition problem with different assumptions on the input tensor and different efficiency and accuracy bounds \cite{Har70,LRA93,BCMT09,BGI'11,GVX13,BCMV13,AGH+15,GM15,HSS15,MSS16,KP20,BHKX22,DdOLST22,AAM-38-385}.
\par
In this paper, we study the algorithmic problem of \textit{approximately} decomposing an arbitrary symmetric order-$3$ tensor $T \in \C^n \otimes \C^n \otimes \C^n$ uniquely (up to permutation and scaling) into a sum of rank-one tensors. To do this efficiently, we need to impose certain restrictions on the "independence" of the rank-one components. More formally, we assume that the rank-one components are of the form $u_i \otimes u_i \otimes u_i$ where the $u_i$'s are linearly independent. We will explore these restrictions in more detail in Section \ref{sec:symmtensordecomp}.
\par
While this problem is well-studied when the underlying model of computation is exact real arithmetic (even when the input tensor has some noise), not much work has been done in the setting where the underlying model of computation is finite precision arithmetic  (see  Section \ref{sec:fparithmetic} for a presentation of this model). The key difficulties lie in the fact that every arithmetic operation in this model is done approximately and the stored numbers can also have some adversarial error (even the input).
\par
An iterative algorithm is called numerically stable if it can be
implemented using polylogarithmic many bits in finite precision arithmetic \cite{smale_1997,9317903}.
{ The central contribution of this paper is a rigorous analysis of a numerically stable algorithm 
that runs in linear time in the input size. This algorithm is inspired by simultaneous diagonalisation algorithm, and appears as 
Algorithm \ref{algo:Jennrich} in Section~\ref{sec:complete}.}

\subsection{Symmetric tensor decomposition}\label{sec:symmtensordecomp}
{ Let $T \in \C^n \otimes \C^n \otimes \C^n$ be a symmetric tensor of order $3$. We recall that such an object can be viewed
as a 3-dimensional array $(T_{ijk})_{1 \leq i,j,k \leq n}$ that 
is invariant under all 6 permutations of the indices $i,j,k$.
This is therefore a 3-dimensional generalization of the notion 
of  symmetric matrix.}
In this paper, we { study} 
symmetric tensor decompositions, i.e., decompositions of the form 
\begin{equation}\label{eq:decompdef}
    T = \sum_{i=1}^r u_i \otimes u_i \otimes u_i
\end{equation}
where $u_i \in \C^n$. The smallest possible value of $r$ is the symmetric tensor rank of $T$ and it is NP-hard to compute already for $d=3$. 
This was shown by Shitov~\cite{Shi16}, and a similar NP-hardness result for ordinary tensors was obtained much earlier by H{\aa}stad~\cite{Hstad'89}.
In this paper, we impose an additional linear independence condition on the $u_i$. Note that such a decomposition is unique if it exists, up to a permutation of the $u_i$'s and scaling by cube roots of unity~\cite{KRUSKAL77,Har70}.
{ There is a traditional distinction
between {\em undercomplete} decompositions, 
where $r\leq n$ in~(\ref{eq:decompdef}), and {\em overcomplete} decompositions, where $r>n$.
In this paper  we consider only undercomplete decompositions because of the linear independence condition on the $u_i$. Moreover, we will impose 
the addditional condition that $r$ is exactly equal to $n$, i.e., we focus on {\em complete decompositions.}
We say that a tensor is {\em diagonalisable} if it satisfies these two conditions. 
The results of the present paper will be extended to general undercomplete decomposition in a forthcoming work by reduction to the complete case.}
\par
{ One can also study the \textit{decision} version of the problem: Given an arbitrary symmetric tensor $T$, is $T$ diagonalisable? A randomized polynomial-time algorithm is known for this problem \cite{KoiranSaha21,KS21}.}


\subsection{ 
Approximate tensor decomposition}\label{sec:fwddef}

{ As explained above,} an order-$3$ symmetric tensor $T \in \C^n \otimes \C^n \otimes \C^n$ is called diagonalisable\footnote{ In a different terminology, these are the "concise symmetric tensors of minimal rank." Concise tensors of minimal {\em border} rank have been studied recently from a geometric point of view in~\cite{JLP22}.} if there exist linearly independent vectors $u_i \in \C^n$ such that $T = \sum_{i=1}^n u_i^{\otimes 3}$. The objective of the $\varepsilon$-approximation problem for tensor decomposition is to find linearly independent vectors $u'_1,...,u'_n$ such that there exists a permutation $\pi \in S_n$ where
\begin{align*}
    ||\omega_iu_{\pi(i)} - u'_i|| \leq \varepsilon
\end{align*}
with $\omega_i$  a cube root of unity. Here $\varepsilon$ is the desired accuracy parameter given as input. 
{ Hence the problem is essentially that of approximating the vectors $u_i$ appearing in the
decomposition of $T$. Note that this is a {\em forward approximation} in the sense of numerical analysis (compare with definitions~\ref{def:forward} and~\ref{def:backward}).}

\subsection{Model of Computation}\label{sec:fparithmetic}

{ We are chiefly interested in the finite precision model 
of arithmetic. Some algorithms are also presented in
exact real arithmetic as an intermediate step toward 
their derivation in the finite precision model.
For the latter model, we use like \cite{9317903} (refer to Section 1.1.2) the standard ﬂoating point axioms from \cite{Hig02}. 
We now elaborate on this model for completeness of the exposition. 

 It is assumed that} numbers are stored and manipulated up to some machine precision $u$ which is a function of $n$, the size of the input and $\delta$ which is the desired accuracy parameter. This means that every number $x \in \mathbb{C}$ is stored as $\text{fl}(x) = (1 + \Delta)x$ for some adversarially chosen $\Delta \in \mathbb{C}$, satisfying $|\Delta| \leq u$ and each arithmetic operation $ * \in \{+ , - , \times, \div\}$ is guaranteed to yield an output satisfying 
\begin{equation}\label{eq:floatingpointarithmetic}
    \text{fl}(x * y) = (x * y)(1 + \Delta) \text{ where } |\Delta| \leq u
\end{equation}
It is also standard and convenient to assume that we can evaluate $\sqrt{x}$ and $x^{\frac{1}{3}}$ for any $x \in \C$, where again $\text{fl}(\sqrt{x}) = \sqrt{x}(1 + \Delta)$ and $\text{fl}(x^{\frac{1}{3}}) = y(1 + \Delta)$ for $|\Delta| \leq u$ where $y$ is a cube root of $x$.
\par
Thus, the outcomes of all operations are adversarially noisy due to roundoff. The bit lengths
of numbers stored in this form remain ﬁxed at $\log(\frac{1}{u})$.
An iterative algorithm that can be implemented in ﬁnite precision (typically, polylogarithmic in the input size and desired accuracy) is called numerically stable. Note that in this model it is not even assumed that the input is stored exactly.

\subsection{Results and Techniques}
\label{sec:results}

Recall that an order-$3$ tensor $T \in (\C^n)^{\otimes 3}$ is called diagonalisable if there exist linearly independent vectors $u_1,...,u_n \in \C^n$ such that $T$ can be decomposed as in (\ref{eq:decompdef}).
\begin{definition}[Condition number of a diagonalisable symmetric tensor]
\label{def:conditionnumber}
Let $T$ be a diagonalisable symmetric tensor over $\mathbb{C}$ such that $T = \sum_{i=1}^n u_i^{\otimes 3}$. Let $U \in M_n(\mathbb{C})$ be the matrix with rows $u_1,\ldots,u_n$. We define the tensor decomposition condition number of $T$ as: $\kappa(T) = ||U||^2_F + ||U^{-1}||^2_F$.
\end{definition}
{ 
 We will show in Section \ref{sec:complete} that $\kappa(T)$ is well defined: for a diagonalisable tensor the condition number is independent of the choice of $U$.
Note that when $U$ is close to a singular matrix, the
corresponding tensor is poorly conditioned, i.e., has a large condition number. This is not surprising since 
our goal is to find a decomposition where the vectors $u_i$ are linearly independent.} 
\par
Our main result is a randomized polynomial time algorithm in the finite precision 
model which on input a diagonalisable tensor, an estimate $B$ for the condition number of the tensor and an accuracy parameter $\varepsilon$, returns a forward approximate solution to the tensor decomposition problem (following the definition in Section \ref{sec:fwddef}).
\par
In the following, we denote by $T_{MM}(n)$ the number of arithmetic operations required to multiply two $n \times n$ matrices in a numerically stable manner. { If $\omega$ denotes the exponent of matrix multiplication, it is known that $T_{MM}(n)=O(n^{\omega+\eta})$ for all $\eta>0$ (see Section~\ref{sec:fastlinalg} for details).}
\begin{theorem}[Main Theorem] \label{th:main}
{ There is an algorithm} which, given a diagonalisable tensor $T$, a desired accuracy parameter $\varepsilon$ and some estimate $B \geq \kappa(T)$, 
outputs an $\varepsilon$-approximate solution to the tensor decomposition problem for $T$ in 
$$O(n^3 + T_{MM}(n)\log^2 \frac{nB}{\varepsilon})$$ arithmetic operations on a floating point machine with 
$$O(\log^4(\frac{nB}{\varepsilon})\log n)$$ bits of precision, with probability at least $\Big(1- \frac{1}{n} - \frac{12}{n^2}\Big)\Big(1 - \frac{1}{\sqrt{2n}} - \frac{1}{n}\Big)$.
\end{theorem}
{ The corresponding algorithm appears as Algorithm~\ref{algo:Jennrich} in Section~\ref{sec:complete}.}
A simplified version of this algorithm is presented in Section~\ref{sec:algorithmvague}.
The following are the important conclusions from the above theorem:
\begin{itemize}
    \item The number of bits of precision required for this algorithm is polylogarithmic in $n$, $B$ and $\frac{1}{\varepsilon}$.
    \item { The running time as measured by the number of arithmetic operations is $O(n^3)$ for all $\varepsilon = \frac{1}{\text{poly(n)}}$, i.e., it is linear in the size of the input tensor. This requires the use of fast matrix multiplication. With standard matrix multiplication, the running time is quasilinear instead of linear  (i.e., it is multiplied by a polylogarithmic factor). The bit complexity of the algorithm is also quasilinear.}
{ \item The algorithm can provide inverse exponential accuracy, 
i.e., it still runs in polynomial time even when the desired accuracy parameter is $\varepsilon = \frac{1}{\exp{(n)}}$.}
\end{itemize}
{ In order to obtain this result we combine techniques from algorithm design and algorithm analysis;
the main ideas are outlined in sections \ref{sec:algorithmvague}
and \ref{sec:overview}.
To the best of our knowledge, this is the first tensor
{ decomposition algorithm shown to work in polylogarithmic precision.  Moreover, this algorithm is also the first to run in a 
linear number of arithmetic operations (i.e., prior to this work 
no linear time algorithm was known, {\em even in the exact arithmetic model}).} }

{ Even if our algorithm is run 
in exact arithmetic, we do not know how to improve the
$O(n^3 + T_{MM}(n)\log^2 \frac{nB}{\varepsilon})$ bound on the number
of arithmetic operations from Theorem~\ref{th:main} 
In particular, some dependency on $\epsilon$
is unavoidable.  Indeed, the exact decomposition of a symmetric tensor 
sometimes requires the use of irrational numbers even if the input tensor has rational entries (see for instance the examples in Section 3 of~\cite{KS21}).
}

\subsubsection{Related work and discussion}
\label{sec:related}

Our algorithm can be viewed as an optimized version of the simultaneous diagonalisation algorithm \cite{Har70,LRA93,Moitra2018AlgorithmicAO}. This algorithm, also referred to in the  literature as  "Jennrich's algorithm," was one of the first  to give provable guarantees for tensor decomposition. In fact, if the input tensor satisfies certain genericity conditions 
this algorithm returns the unique decomposition (up to permutation and scaling) almost surely. It was shown in \cite{BCMV13} that this algorithm runs in polynomial time in the exact arithmetic computational model, i.e., when the model has the underlying assumption that all the steps of the algorithm can be performed exactly. 
Moreover, it is shown in the same paper that
the algorithm is robust to noise in the input. Namely, it was shown that for an input tensor $\Tilde{T} = \sum_{i=1}^n v_i^{\otimes 3} + E$ where $E$ is some arbitrary inverse-polynomial noise, the "simultaneous diagonalisation" algorithm can also be used to output a decomposition $\Tilde{v}_i$ such that  $||v_i - \Tilde{v}_i|| \leq \varepsilon$.  At the heart of the robustness analysis of the simultaneous diagonalisation algorithm in \cite{BCMV13} (refer to their Appendix A) is the following statement about diagonalisability of perturbed matrices: Let $M$ be a diagonalisable matrix that can be written as $M = UDU^{-1}$ where the condition number of $U$ is bounded and let $\Tilde{M}$ be another matrix such that $||M - \Tilde{M}||$ is small. Then $\Tilde{M}$ is also diagonalisable, has distinct eigenvalues and the eigenvectors of $M$ and $\Tilde{M}$ are close. This is similar in spirit to Proposition 1.1 in \cite{9317903} (Theorem \ref{thm:backwardtofwd} in this paper.)
\par

{ One may also drop the genericity condition and attempt to decompose an arbitrary low-rank tensor given as input.
For symmetric tensors with constant rank, such an algorithm can be found in \cite{BSV21}. This algorithm was recently extended to slightly superconstant rank in~\cite{PSV22}.
Still other algorithms for symmetric tensor decomposition can be found in the algebraic literature, see e.g.  \cite{BCMT09,BGI'11}. These two papers do not provide any complexity analysis for their algorithms.

The results presented in this paper are in stark contrast with those of Beltr\'an et al.~\cite{BBV19}. That paper analyzes a  class of tensor decomposition algorithms related to the simultaneous diagonalisation algorithm.
Their conclusion is that all these  "pencil-based algorithms" are numerically unstable. 
In Section \ref{sec:condition} we will argue that we can escape this negative result because (like the version 
of the simultaneous diagonalisation algorithm analyzed in~\cite{BCMV13}) 
our algorithm is randomized. Indeed, the pencil-based
algorithms of~\cite{BBV19} are all deterministic.
Beltr\'an et al. conclude their paper with the following 
sentence:
"We hope that these observations may (re)invigorate the search for numerically stable algorithms for computing CPDs."\footnote{CPD stands for \emph{Canonical Polyadic Decomposition, i.e., decomposition as a sum of rank-1 tensors.}"} The algorithm presented in this paper answers their call, at least for the case of complete decomposition of symmetric tensors. 
 We believe that our techniques can also 
 be applied to decomposition of ordinary tensors. In this paper we have chosen to focus on symmetric tensors because this setting is somewhat simpler technically.\footnote{As is already apparent from the length of the paper, 
 even in the symmetric setting there are plenty of technical details  that need to be taken  care of.}
{ In \cite{koiran2024undercomplete}, we extend this work to 
 the case of undercomplete decompositions. Note that this requires a change in the definition of the condition number $\kappa(T)$; the role of $U^{-1}$ will now be played by the Moore-Penrose pseudoinverse. In  \cite{koiran2024undercomplete} we also perform a smoothed analysis of the condition number.  It confirms that our algorithm indeed runs in linear time, except for some rare badly conditioned inputs.}}
 
\par

\subsubsection{The Algorithm}\label{sec:algorithmvague}

{ Before giving a high-level presentation of our algorithm, we introduce a few  notations. A symmetric tensor $T \in \C^n \otimes \C^n \otimes \C^n$ can be cut into $n$ slices $T_1,\ldots,T_n$ where $T_k = (T_{ijk})_{1 \leq i,j \leq n}$. Each slice is a symmetric matrix of size $n$.
In the algorithm below we also make use of a "change of basis" operation, which applies a linear map of the form $A \otimes A \otimes A$ to a tensor. Here, $A \in M_n(\C)$ and we apply
$A$ to the 3 components of the input tensor. In particular,
for rank-1 symmetric tensors we have 
$$(A \otimes A \otimes A).(u \otimes u \otimes u)=(A^Tu)^{\otimes 3}.$$ 
We give more details on this operation at the beginning of Section~\ref{sec:cobfinitear}.}
The algorithm proceeds as follows.
\begin{enumerate}
    \item[(i)] Pick vectors $a = (a_1,...,a_n)$ and $b = (b_1,...,b_n)$ at random from a finite set and compute two random linear combinations $T^{(a)} = \sum_{i=1}^n a_i T_i$ and $T^{(b)} = \sum_{i=1}^n b_i T_i$ of the slices of $T$.
    \item[(ii)] Diagonalise $(T^{(a)})^{-1}T^{(b)} = VDV^{-1}$. Let $v_1,...,v_n$ be the columns of~$V$.
    \item[(iii)]  Let $u_1,...,u_n$ be the rows of $V^{-1}$.
    \item[(iv)] Let $T' = (V \otimes V \otimes V). T$. Let $T'_1,...,T'_n$ be the slices of $T'$. Define $\alpha_i = \text{Tr}(T'_i)$. We will refer to the computation of $Tr(T'_i)$ as the trace of slices after a change of basis  (TSCB). 
    \item[(v)] Output $(\alpha_1)^{\frac{1}{3}}u_1,...,(\alpha_n)^{\frac{1}{3}} u_n$.
\end{enumerate}
{The above algorithm is a modified version of 
the simultaneous diagonalisation algorithm for symmetric tensors. In terms of algorithm design, our main contribution lies in step (iv). 
Previous versions of the simultaneous diagonalisation algorithm have appealed instead
to the resolution of a linear system of equations: 
see e.g. \cite{BCMV13,Moitra2018AlgorithmicAO} for the case
of ordinary tensors. In the symmetric case, the algebraic algorithm in~\cite{Kayal11} for decomposition of a polynomial as a sum of powers of linear forms
also appeals to the resolution of a linear system for essentially
the same purpose. Our trace-based version of step (iv) is more efficient, and this is crucial for the derivation of the complexity bounds in Theorem~\ref{th:main}. 
Step (iv) is indeed the most expensive: 
it is responsible for the $O(n^3)$ term in the arithmetic complexity of the algorithm. We explain informally at the beginnining of Section~\ref{sec:overview} why our trace-based
approach works.}

\par
In Section \ref{sec:completeexactarithmetic}, we state the above algorithm in more detail and show that if this algorithm is given a complete diagonalisable tensor exactly as input, it indeed returns the (unique) decomposition. In the underlying computational model assumed for the analysis in that section, all arithmetic operations can be done exactly and matrices can be diagonalised exactly.
This is the algorithm that we will adapt to the finite arithmetic model in Section \ref{sec:complete}.
\par
{
\subsubsection{Proof overview and organization of the paper}
\label{sec:overview}

In this section we outline the main steps of the proof,
in the order in which they appear in the paper.

{\bf Trace of the slices after a change of basis (TSCB).} 
After step (iii) of the algorithm, we have 
determined vectors $u_1,\ldots,u_n$ such that
$T=\sum_{i=1}^n \alpha_i u_i^{\otimes 3}$. Here the
$\alpha_i$ are unknown coefficients. As explained in Section~\ref{sec:algorithmvague}, the traditional approach
is to find them by solving the corresponding linear system. One difficulty here is that this system is highly overdetermined: we have one equation for each entry of $T$, but only $n$ unknowns. In this paper we show that the system {\em can} be solved quickly in a numerical stable way by exploiting some of its structural properties. Our approach relies on a change of basis. More precisely, let $T' = (V \otimes V \otimes V). T$
be the tensor defined at the beginning of step (iv).
Since $u_1,\ldots,u_n$ are the rows of $V^{-1}$,
$T'=\sum_{i=1}^n \alpha_i e_i^{\otimes 3}$ where
$e_i$ is the $i$-th standard basis vector. Therefore we can read off the $\alpha_i$ from the entries of $T'$.
This observation is not sufficient to obtain the desired
running time since it is not clear how to perform
a change of basis in $O(n^3)$ arithmetic operations. Indeed, since a symmetric tensor 
of size $n$ has $\Omega(n^3)$ coefficients, one would have to perform a constant number of  operations per coefficient. A further observation is that we do not
need to compute every entry of $T'$: assuming that $T$ is diagonalisable, we know in advance that all entries
of $T'$ except the diagonal ones will be equal to 0 (up to rounding errors). As a result, $\alpha_i$ is approximately equal to the trace of $T'_i$, the $i$-the slice of $T'$. In Section~\ref{sec:cobfinitear} we give
a fairly simple algorithm for the computation of
these $n$ traces in $O(n^3)$ arithmetic operations.
For this we do not even need to assume that the
input tensor is diagonalisable.  
We also
analyse this TSCB algorithm in finite arithmetic in the same 
section. The correctness of our main algorithm 
in exact arithmetic (as presented in Section~\ref{sec:algorithmvague}) is established in Section~\ref{sec:completeexactarithmetic} based on 
the results of Section~\ref{sec:cobfinitear}. 
{ In Section~\ref{sec:completeexactarithmetic} we also present an alternative to the TSCB 
algorithm suggested by an anonymous reviewer. This alternative algorithm
uses the trace (and more precisely the {\em partial trace}) in a different way to solve the linear system.}
The remainder of the paper consists mostly of algorithm analysis, in particular finite precision analysis and probabilistic analysis.

\textbf{Matrix diagonalisation:} For step (ii) of our algorithm we require a fast and numerically stable diagonalisation algorithm. Amazingly for such a fundamental mathematical task, no satisfactory solution was available before the recent breakthrough in~\cite{9317903}. Given $A \in M_n(\C)$ and $\delta>0$, their algorithm computes an invertible matrix $V$ and a diagonal matrix $D$ such that $||A-VDV^{-1}|| \leq \delta.$ Moreover, $V$ is guaranteed to be reasonably well-conditioned in the sense that $||V||.||V^{-1}|| = O(n^{2.5}/\delta)$. Note however that~$V$ might become 
arbitrarily poorly conditioned as $\delta$ goes to 0.
This is a problem for our algorithm because we keep working with $V$ in subsequent steps (for inversion in step (iii), and change of basis in step (iv)). 
The main question that we address in Section~\ref{sec:diag} is: can we have a better guarantee on $V$ assuming that the input matrix $A$ is diagonalisable? We show that this is indeed the case, 
and the bound that we obtain is based on the {\em Frobenius eigenvector condition number} defined 
in that section. Note that the choice of $\kappa(T)$ as our tensor decomposition number arises from that analysis.

Another issue that we address  in Section~\ref{sec:diag}
is the assumption $||A|| \leq 1$ on the input matrix.
Relaxing this assumption in infinite precision arithmetic is very straightforward: given a bound $B \geq 1$ on $||A||$, one can simply divide $A$ by $B$ and this does 
not change the eigenvectors. In finite arithmetic, however, this simple scaling leads to round-off errors.
The corresponding analysis is performed in Section~\ref{sec:diag} as well.
Finally, we note that the diagonalization algorithm of~\cite{9317903} is responsible for the 
number of bits of precision needed in our main result (Theorem~\ref{th:main}).

\textbf{Finite precision analysis of tensor decomposition:} The correctness of the infinite-precision version of our main algorithm is established in Section~\ref{sec:completeexactarithmetic}, and we proceed with its analysis in finite arithmetic in Section~\ref{sec:complete}. The principle behind this analysis is relatively straightforward: we need to 
show that the output of each of the 7 steps does not 
deviate too much from the ideal, infinite-precision output.
For each step, we have two sources of error:
\begin{itemize}
\item[(i)] The input to that step might not be exact because of errors accumulated in previous steps.
\item[(ii)] The computation performed in that step (on an inexact input) is inexact as well.
\end{itemize}
Summing up these two contributions, we can upper bound the error for that step. Moreover, for each step we already have estimates for the error (ii) due to the inexact computation. In particular, for basic operations such as 
matrix multiplication and inversion there are well-known guarantees recalled in Section~\ref{sec:fastlinalg};  for the change of basis algorithm we have the guarantees from Section~\ref{sec:cobfinitear}; and for diagonalisation we have the guarantees from Section~\ref{sec:diag} based on~\cite{9317903}. Nevertheless, obtaining reasonably precise error bounds from this analysis requires rather long and technical developments. For this reason, a part of the analysis is relegated to the appendix.
\par
A more streamlined error analysis of this algorithm can be found in \cite{Saha23}. In that exposition, some conditions have been derived related to \textit{numerical stability} of a composition of \textit{numerically stable} algorithms which have then been applied to this algorithm {\color{blue} (the composition theorem in \cite{Saha23} is similar in spirit to the earlier composition theorem in~\cite{BNV23}).} But that analysis loses out on some bits of precision - it gives a guarantee that the algorithm requires $\log^{12}(\frac{nB}{\varepsilon})\log(n)$ bits of precision as compared to the $\log^{4}(\frac{nB}{\varepsilon})\log(n)$ bits of precision shown in this paper.

\textbf{Probabilistic analysis:} There are two sources of randomization in our algorithm: the diagonalisation algorithm from~\cite{9317903} is randomized, and moreover 
our algorithm begins with the computation of two random linear combinations $T^{(a)}, T^{(b)}$ of slices of the input tensor (Section~\ref{sec:algorithmvague}, step (i)).
As it turns out, the error bounds from Section~\ref{sec:complete} are established under the hypothesis that the {\em Frobenius condition
number} of $T^{(a)}$ is "small" and the eigenvalue gap of $(T^{(a)})^{-1}T^{(b)}$ is "large". We therefore need to show that this hypothesis is satisfied for most choices of the random vectors $a,b$. For this we assume that $a,b$ are chosen uniformly at random from a discrete grid.
Our analysis  in Section~\ref{subsection:correctnessjennrich} follows a two-stage process: 
\begin{itemize}
\item[(i)] First we assume that 
$a$ and $b$ are drawn from the uniform distribution on the hypercube $[-1,1)^{n}$. This is analyzed with the Carbery-Wright inequality, a well-known anticoncentration inequality. 

\item[(ii)] In a second stage, we round the (real valued) coordinates of $a$ and $b$ in order to obtain a point of 
the discrete grid. This is analysed with the multivariate Markov inequality.\footnote{This discretization stage could also be analyzed with~\cite[Theorem~3]{Koi95}, but we would not obtain a sharper bound in this case.}
\end{itemize}
This two-stage process is inspired by the construction of "robust hitting sets" in~\cite{forbes2017pspace}. However, the general bounds from~\cite[Theorem~3.6]{forbes2017pspace} are not sharp enough for our purpose: they would lead to an algorithm using polynomially many bits of precision, but we are aiming for  polylogarithmic precision.
As a result, we need to perform an ad hoc analysis for certain linear and quadratic polynomials connected to the
 Frobenius condition
number of $T^{(a)}$ and to the eigenvalue gap of $(T^{(a)})^{-1}T^{(b)}$. These are essentially the polynomials occuring in~\cite{BCMV13} in their analysis 
of the stability of the simultaneous diagonalisation algorithm with respect to input noise; but in that paper they choose $a,b$ to be
(normalized) Gaussian vectors rather than points from a discrete grid.
}

{ 
\subsection{Condition numbers, numerical (in)stability and the negative result of~\cite{BBV19}}
\label{sec:condition}

In this section we provide more background on condition numbers in numerical computation.
A book-length treatment of this subject can be found in~\cite{BC13}. We also discuss in more detail the numerical instability result of~\cite{BBV19}.

 There is no universally accepted definition of a "condition number" in numerical analysis, but a common one, used in~\cite{BBV19}, is as follows. Suppose we wish to compute a map $f:X \rightarrow Y$. The condition number of $f$ at an input $x$ is a measure of the variation of the image $f(x)$ when $x$ is perturbed by a small amount. This  requires the choice of appropriate distances on the spaces $X$ and $Y$. The condition number is therefore a quantitative measure of the continuity of $f$ at $x$. In particular, it is independent of the choice of an algorithm for computing $f$. In finite arithmetic, we cannot hope to approximate $f(x)$ with a low precision algorithm 
at an input $x$ with a high condition number since 
we do not even assume that the input is stored exactly. Moreover, designing algorithms that work in low precision at well-conditioned inputs is often a challenging task. 
Suppose for instance that we want to approximate the
eigenvectors of a matrix. In order to estimate 
the condition number in the above sense, we need to
understand how the eigenvectors evolve under 
a perturbation of the input matrix. This is a relatively standard task in perturbation theory, see for instance 
Appendix~A of \cite{BCMV13}\footnote{As already mentioned in Section~\ref{sec:related}, this property is at the heart of their analysis of the robustness of the simultaneous diagonalisation algorithm.} or the proof of Proposition~1.1 in~\cite{9317903}. However, until the 
recent breakthrough~\cite{9317903} we did not have any 
efficient, low-precision algorithm for this task 
(see Theorem \ref{thm:eig} in Section \ref{sec:diag} for
a precise statement of their result).

Sometimes, the above continuity-based definition of condition numbers is not suitable. This is for instance the case for decision problems, where the map $f$ is boolean-valued. A popular alternative is to use the inverse of the (normalized) distance to the set of ill-posed instances~\cite[chapter 6]{BC13}. One can sometimes show that these two notions coincide~\cite[Section 1.3]{BC13}.}

{ For the purpose of this paper we work with the somewhat ad-hoc choice of $\kappa(T)$ as our condition number 
because this parameter controls the numerical precision
needed for our main algorithm, as shown by Theorem~\ref{th:main}\footnote{ $\kappa(T)$  also appears in the sublinear term for the arithmetic complexity of the algorithm.}. In particular, we have found it more convenient to work with $\kappa(T)$ than with a quantity such as $||U||.||U^{-1}||$, commonly used as a condition number in numerical linear algebra.

A precise comparison of our results with the numerical instability 
 result of~\cite{BBV19} is delicate because we do not work 
 in the same setting. In particular, they work with ordinary instead of symmetric tensors; they do not work with the same condition number; and their result is obtained for undercomplete rather than complete decompositions. As already mentioned in Section~\ref{sec:related}, we believe that the main reason why we obtain a positive result is due to yet another
 difference, namely, the use of randomization in step (i) of our algorithm. In the setting of~\cite{BBV19} one would have to 
 take two {\em fixed} linear combinations $T^{(a)}, T^{(b)}$ of the slices. Essentially, they show that
 for every fixed choice of a pair of linear combinations, there are input tensors for which this choice is bad; whereas we show that for every (well conditioned) input $T$, most choices of $a$ and $b$ are good. }
\par

{ \subsection{Error Reduction by Repetition}

An anonymous referee has suggested to decrease the probability of error
of our main algorithm by running it a constant number of
times on the same input.
In general, if an algorithm has probability of error $p$ and we repeat it, say, 100 times, the probability of error is obviously reduced to $p^{100}$ if we can check whether a proposed output is correct (we have $p=O(1/n)$ in Theorem~\ref{th:main}). A precise analysis of this generic
method for our main algorithm raises some interesting questions
because ideally, one would like to check in linear time that a proposed 
decomposition is correct. We need to face the two following issues:
\begin{itemize}
\item[(i)] Given a tensor $T$ and vectors $u_1,\ldots,u_n$, it is not completely obvious how to test that 
\begin{equation} \label{eq:proposed}
T=\sum_{i=1}^n u_i^{\otimes 3} 
\end{equation}in $O(n^3)$ operations, even in exact arithmetic.
Indeed, expanding one tensor power $u_i^{\otimes 3}$ by brute force gives
rise to $n^3$ terms, and adding up these $n$ tensors will
therefore require $O(n^4)$ operations.

\item[(ii)] In reality we do not have to check that
the proposed decomposition~(\ref{eq:proposed}) is exact, but that it is
approximately correct.
\end{itemize}
One solution to issue (i) is to view symmetric tensor decomposition as a problem about the decomposition of multivariate polynomials as sums of cubes.
Namely, one can view $T$ as the array of coefficients of the polynomial
$f(x_1,\ldots,n)=\sum_{i,j,k}T_{ijk}x_ix_jx_k$. 
Then~(\ref{eq:proposed})  becomes equivalent to the polynomial identity $f(x)=\sum_{i=1}^n (u_i^Tx)^3$. 
One can test whether this identity holds at any given point in $O(n^3)$
arithmetic operations. If $x$ is chosen uniformly at random from a finite set $S^n$, the probability of error is at most $3/|S|$ by the Schwartz-Zippel lemma. For any polynomial $p$,  one can therefore make 
the probability of error smaller than $1/p(n)$ with logarithmically
many bits of precision for each coordinate of $x$. As a side remark, we note that an
alternative solution to this Schwartz-Zippel based algorithm can 
be derived from the linear time LCSCB algorithm for computing a linear combination of slices after a change of basis. This algorithm is described in our follow-up work~\cite{koiran2024undercomplete} (we omit the details of this alternative solution). 

Next, we briefly explain how the second issue (ii) could be addressed. First, recall that there are two main notions of approximation error in numerical analysis: backward error and forward error. 
As explained in Section~\ref{sec:fwddef} this paper is focused on 
forward error, i.e., we want to compute approximations $u'_i$ to the
"true vectors" $u_i$ occurring in the decomposition~(\ref{eq:proposed}).
If we define $T'$ as the tensor $\sum_{i=1}^n {u'_i}^{\otimes 3}$, 
backward error would be measured by $||T-T'||$ for some appropriate norm, 
for instance, the Frobenius norm (see Definition~\ref{def:tensornorm}).
A forward error for tensor decomposition immediately yields a backward error. Obtaining a forward error from a backward error takes more work,
but can be done from the stability analysis in~\cite{BCMV13} or directly from a perturbation analysis for the eigendecomposition problem (\cite[Proposition 1.1]{9317903} or ~\cite{BCMV13}).

After this reminder, we return to issue (ii) and briefly explain how
to test that a decomposition is approximately correct in the backward
sense (as we have just explained, this backward error can then be translated 
into a forward error bound). We have addressed issue (i) with the Schwartz-Zippel lemma. What we need for (ii) is a robust version of 
that lemma: if a polynomial is "far away" (e.g., in $L^2$ norm) from the identically 0 polynomial then its evaluations at most points $x \in S^n$ are also far away from zero. As shown by Forbes and Shpilka~\cite[Theorem 3.6]{forbes2017pspace}, such results can be obtained  with a combination of the
Carbery-Wright inequality and the multivariate Markov Theorem (see also~\cite[Lemma 12]{VX11}). As pointed out at the end of Section~\ref{sec:overview}, we use these tools in Section~\ref{subsection:correctnessjennrich} for the probabilistic analysis of our main algorithm.}

\section{Preliminaries: Fast and Stable Linear Algebra}\label{sec:fastlinalg}

In this section, we explore the computational model of finite precision arithmetic that has already been introduced in Section \ref{sec:fparithmetic} in greater detail. 
We present the different estimates for various linear algebraic operations such as inner product of vectors, matrix multiplication and matrix inversion. This is the main content of Sections \ref{sec:morefparithmetic} and \ref{sec:mminvqr} and they have been taken from \cite{Hig02} and \cite{9317903}. We include these for completeness of the exposition.
\subsection{Finite precision arithmetic}\label{sec:morefparithmetic}

We'll need to compute the inner product of two vectors $x , y \in \mathbb{C}^n$. For this purpose,  we will assume that
\begin{equation}\label{eq:ipbound}
    |x^Ty - \text{fl}(x^Ty)| \leq \gamma_n||x||||y||
\end{equation}
where $\textbf{u}$ is the machine precision and $\gamma_n = \frac{n\textbf{u}}{1-n\textbf{u}}$.
For a proof, refer to the discussion at the discussion in \cite{Hig02}, Section 3.1.
\par
We will also assume similar guarantees for matrix-matrix addition and matrix-scalar multiplication. More specifically, if $A \in \mathbb{C}^{n \times n}$ is the exact output of such an operation, then its floating point representation $\text{fl}(A)$ will satisfy $$\text{fl}(A) = A + A\circ\Delta \text{ where } |\Delta_{ij}| < \textbf{u}.$$ Here $A\circ\Delta$ denotes the entry-wise product $A_{ij}\Delta_{ij}$. This multiplicative error can be converted into an additive form i.e.
\begin{equation}\label{eq:multfl}
    ||A\circ\Delta|| \leq \textbf{u}\sqrt{n}||A||.
\end{equation}
For more complicated linear algebraic operations like matrix multiplication and matrix inversion, we require more sophisticated error guarantees which we now explain. 

\subsection{Matrix Multiplication and Inversion}
\label{sec:mminvqr}
The definitions we state here are taken from \cite{9317903} (Definitions 2.6 and 2.7)
\begin{definition}\label{def:MULTalg}
A $\mu_{\text{MM}}(n)$-stable multiplication algorithm $\text{MM}(.,.)$ takes as input $A,B \in \mathbb{C}^{n \times n}$ and a precision $\textbf{u} > 0$ and outputs $C = \text{MM}(A,B)$ satisfying
\begin{align*}
    ||C - AB|| \leq \mu_{\text{MM}}(n) \cdot \textbf{u} ||A|| ||B||
\end{align*}
on a floating point machine with precision $\textbf{u}$, in $T_{\text{MM}}(n)$ arithmetic operations.
\end{definition}
\begin{definition}\label{def:INValg}
A $(\mu_{\text{INV}}(n), c_{\text{INV}})$-stable inversion algorithm $\text{INV}(.)$ takes as input $A \in \mathbb{C}^{n \times n}$ and a precision $\textbf{u}$ and outputs $C = \text{INV}(A)$ satisfying 
\begin{align*}
    ||C - A^{-1}|| \leq \mu_{\text{INV}}(n).\textbf{u}. (\kappa(A))^{c_{\text{INV}}\log n} ||A^{-1}||.
\end{align*}
on a floating point machine with precision $\textbf{u}$, in $T_{\text{INV}}(n)$ arithmetic operations.
\end{definition}
The following theorem by \cite{DDHK'07} gives a numerically stable matrix multiplication algorithm which is used by \cite{DDH07} to gives numerically stable algorithm for matrix inversion and a numerically stable algorithm for QR factorization of a given matrix. We use the presentation of these theorems from \cite{9317903} (Theorem 2.10 (1) and (2)).
\begin{theorem}\label{thm:fastlinearalgebra}
\begin{enumerate}
    \item If $\omega$ is the exponent of matrix multiplication, then for every $\eta > 0$, there is a $\mu_{\text{MM}}(n)$-stable matrix multiplication algorithm with $\mu_{\text{MM}}(n) = n^{c_{\eta}}$ and $T_{\text{MM}}(n) = O(n^{\omega + \eta})$, where $c_{\eta}$ does not depend on $n$.
    \item Given an algorithm for matrix multiplication satisfying part (1), there is a $(\mu_{\text{INV}}(n), c_{\text{INV}})$-stable inversion algorithm with 
    \begin{align*}
        \mu_{\text{INV}}(n) \leq O(\mu_{\text{MM}}(n)n^{\log 10}) \text{ and } c_{\text{INV}} \leq 8,
    \end{align*}
    and $T_{\text{INV}}(n) = O(T_{\text{MM}})(n)$.
\end{enumerate}
In particular, all of the running times above are bounded by $T_{\text{MM}}(n)$ for a $n \times n$ matrix.
\end{theorem}
{ We used the constant $c_{\eta}$ keeping in line with the bounds from \cite{9317903}. One can refer to \cite{Bini1980/81,DDHK'07} for exact bounds on $c_{\eta}$.}

Instead of the fast matrix multiplication algorithm, one can also consider the errors from the conventional computation. Let $A, B$ be two matrices and let $C = AB$ computed on a floating point machine with machine precision~\textbf{u}. From (3.13) in \cite{Hig02}, we have that
\begin{equation}\label{eq:slowmm}
    ||C - AB|| \leq 2n \textbf{u} ||A||_F ||B||_F.
\end{equation}
{ We will use this bound in the next section, where (in contrast to Section~\ref{sec:diag}) fast matrix multiplication is not needed.}

\section{Slices after a change of basis}\label{sec:cobfinitear}

Given tensors $T,T' \in \mathbb{C}^{n \times n \times n}$, we say that there is a change of basis $A \in \text{M}_n(\mathbb{C})$ that takes $T$ to $T'$ if $T' = (A \otimes A \otimes A).T$. This notation was already introduced
in Section~\ref{sec:algorithmvague} to give an outline of our main algorithm.

In  the present section we give a fast and numerically stable algorithm for computing the trace of the slices after a change of basis. More formally, given a tensor $T$ and a matrix $V$, it computes $Tr(S_1),...,Tr(S_n)$ where $S_1,...,S_n$ are the slices of the tensor $S = (V \otimes V \otimes V).T$ with small error in $O(n^3)$ many arithmetic operations. 
\par
 Written in standard basis notation, the equality $T' = (A \otimes A \otimes A).T$
 corresponds to the fact that for all $i_1,i_2,i_3 \in [n]$,
\begin{equation}\label{eq:changeofbasisdef}
    T'_{i_1i_2i_3} = \sum_{j_1,j_2,j_3 \in [n]} A_{j_1i_1}A_{j_2i_2}A_{j_3i_3} T_{j_1j_2j_3}.
\end{equation}
Note that if $T = u^{\otimes 3}$ for some vector $u \in \C^n$, then $(A \otimes A \otimes A).T = (A^Tu)^{\otimes 3}$.
The choice of making $A$ act by multiplication by $A^T$ rather than 
by multiplication by $A$ is somewhat arbitrary, but it is natural from the point of view of the polynomial-tensor equivalence in Definition~\ref{def:polytensoreq} below.
Indeed, from the polynomial point of view a change of basis corresponds to a linear change of variables. More precisely, 
if $f(x_1,\ldots,x_n)$ is the polynomial associated to $T$ and 
$f'(x_1,\ldots,x_n)$ is the polynomial associated to $T'=(A \otimes A \otimes A).T$, we have $f'(x)=f(Ax)$.
\par
\begin{definition}[Polynomial-Tensor Equivalence]\label{def:polytensoreq}
Let $f \in \C[x_1,...,x_n]_3$ be a homogeneous degree-$3$ polynomial in $n$ variables. We can form with the coefficients of $f$ a symmetric tensor of order three $T_f=(T_{ijk})_{1 \leq i,j,k \leq n}$ so that
$$f(x_1,\ldots,x_n)=\sum_{i,j,k=1}^n T_{ijk} x_i x_j x_k.$$
\end{definition}
The following theorem was derived in~\cite{KS21} in the polynomial language of Definition~\ref{def:polytensoreq}.
\begin{theorem} \label{thm:P3structural}
Let $T \in C^{n \times n \times n}$ be a tensor with slices $T_1,...,T_n$ and let $S = (A \otimes A \otimes A).T$ where $A \in M_{n}(\C)$. Then the slices $S_1,...,S_n$ of $S$ are given by the formula:
\begin{align*}
    S_k = A^T D_k A
\end{align*}
 where $D_k = \sum_{i=1} a_{i,k} T_i$ and $a_{i,k}$ are the entries of A.
 \par
In particular, if $T=\sum_{i=1}^n e_i^{\otimes 3}$, we have $D_k = \text{diag }(a_{1,k},...,a_{n,k} )$. 
\end{theorem}
\begin{corollary}\label{corr:P3structural}
Let $S =\sum_{i=1}^r a_i^{\otimes 3}$. Let $A$ be the $r \times n$ matrix with rows $a_1,...,a_r$. Then the slices $S_k$ of $S$ are given by the formula 
$$S_k = A^T D_k A \text{ where } D_k = \text{diag}(a_{1,k},...,a_{r,k}).$$
\end{corollary}

\textbf{Norms: }We denote by $||x||$ the $\ell^2$ (Hermitian) norm of a vector $x \in \C^n$. For $A \in M_n(\C)$, we denote by $||A||$ its operator norm and by $||A||_F$ its Frobenius norm:
\begin{equation} \label{eq:frobeniusnorm}
    ||A||^2_F=\sum_{i,j=1}^n |A_{ij}|^2.
\end{equation}
We always have $||A|| \leq ||A||_F$. 
For a given matrix $V$, we define $\kappa_F(V) = ||V||_F^2 + ||V^{-1}||_F^2$.
\newline
\begin{definition}[Tensor Norm]\label{def:tensornorm}
Given a tensor $T \in (\C^n)^{\otimes 3}$, we define the Frobenius norm $||T||_F$ of $T$ as
\begin{align*}
    ||T||_F = \sqrt{\sum_{i,j,k=1}^n |T_{i,j,k}|^2}
\end{align*}
Then if $T_1,...,T_n$ are the slices of $T$, we also have that
\begin{equation}\label{eq:tensornormslices}
\begin{split}
    \sum_{i=1}^n||T_i||^2_F &= \sum_{j,k \in [n]} |(T_i)_{j,k}|^2 \\
    &= \sum_{i,j,k \in [n]} |(T_i)_{j,k}|^2 = \sum_{i,j,k \in [n]} |T_{ijk}|^2 = ||T||_F^2.    
\end{split}
\end{equation}
\end{definition}

{ Let $T$ be an order-$3$ tensor and let $S$ be another order-$3$ tensor obtained by a change of basis of $T$ by a matrix $V$. Following (\ref{eq:changeofbasisdef}), this is denoted by $S = (V \otimes V \otimes V).T$. Then given $T$ and $V$, we give the following linear time (in the input size) algorithm to compute the trace of the slices of $S$. As we show in \cite{koiran2024undercomplete}, the tensor $S$ can be computed explicitly in $O(n^{3.25164})$ arithmetic operations. But as discussed in Section \ref{sec:overview}, for the purposes of being used as a subroutine in Algorithm \ref{algo:Jennrich}, just computing the trace of the slices of the tensor $S$ is sufficient. Moreover, we would need this computation to be \textit{numerically stable} and take $O(n^3)$ many arithmetic operations. We propose the following algorithm for this problem which is based on the computation strategy in (\ref{eq:traceofslicesmaineq}).}

\begin{algorithm}[H] \label{algo:fastcob}
\SetAlgoLined
\nonl \textbf{Input:} An order-$3$ symmetric tensor $T \in \mathbb{C}^{n \times n \times n}$, a matrix $V = (v_{ij}) \in \mathbb{C}^{n \times n}$.\\
\nonl Let $T_1,...,T_n$ be the slices of $T$. \\
Compute $W = V^TV$ on a floating point machine. \\
Compute $x_{m,k} = (WT_m)_{k,k}$ on a floating point machine for all $m,k \in [n]$. \\
Compute $x_m = \sum_{k=1}^n x_{m,k}$ on a floating point machine for all $m \in [n]$ . \\
Compute $\Tilde{s}_i = \sum_{m=1}^n v_{m,i}x_m$ on a floating point machine for all $i \in [n]$. \\
\nonl Output $\Tilde{s}_1,...,\Tilde{s}_n$
\caption{
Trace of the slices after a change of basis (TSCB)}
\end{algorithm}

The following is the main theorem of this section.
\begin{theorem}\label{thm:fastcob}
Let us assume that a tensor $T \in (\C^n)^{\otimes 3}$ and a matrix $V \in M_n(\C)$ are given as input to Algorithm \ref{algo:fastcob}. Set $S = (V \otimes V \otimes V).T$ following the definition in (\ref{eq:changeofbasisdef}) and let $S_1,...,S_n$ be the slices of $S$. Then the algorithm returns $\Tilde{s}_1,...,\Tilde{s}_n$ such that
\begin{equation}
    |\Tilde{s}_i - \text{Tr}(S_i)| \leq \mu_{CB}(n) \cdot \textbf{u} \cdot ||V||_F^3 ||T||_F
\end{equation}
where $\mu_{CB}(n) \leq 14n^{\frac{3}{2}} $. It performs $T_{CB}(n) = O(n^3)$ operations on a machine with precision $\textbf{u} < \frac{1}{10n}$.
\end{theorem}
\begin{proof}
Let $S' \in \C^{n \times n \times n}$ be such that $S' = (V \otimes V \otimes V).T$. Let $S'_1,...,S'_n$ be the slices of $S'$. We first claim that
\begin{equation}\label{eq:traceofslicesmaineq}
\sum_{m=1}^n v_{mi} \Big(\sum_{k=1}^n (V^T V T_m)_{k,k}\Big) = Tr(S'_i)    
\end{equation}
Using Theorem \ref{thm:P3structural}, we know that $S'_i = V^T D_i V$ where $D_i = \sum_{m=1}^n v_{m,i}T_m$. Now using the cyclic property and the linearity of the trace operator, we get that
\begin{equation}\label{eq:trace}
\begin{split}
    Tr(S'_i) = Tr(V^T D_i V) &= Tr(V^TV D_i) = Tr(V^T V (\sum_{m=1}^n v_{m,i}T_m))  \\
    &= \sum_{m=1}^n v_{mi} Tr(V^TVT_m) = \sum_{m=1}^n v_{mi}\Big( \sum_{k=1}^n  (V^T V T_m)_{k,k}\Big).
\end{split}
\end{equation}
From this, we conclude that if Algorithm \ref{algo:fastcob} is run in exact arithmetic, it computes exactly the trace of the slices $S'_i$ of $S'$.
\par
\textbf{Running Time:} We analyse the steps of the algorithm and deduce the number of arithmetic operations required to perform the algorithm. Note that only the numbered steps contribute to the complexity analysis. 
\begin{enumerate}
    \item Since $V \in M_n(\C)$, Step 1 can be done { in $O(n^3)$ operations with ordinary matrix multiplication.}
    \item In Step 2, for each $m,k \in [n]$, we compute the inner product of the $k$-th row of $W$ with the $k$-th column of $T_m$. Computation of each inner product takes $n$ arithmetic operations. There are $n^2$ such inner product computations. So this step requires $n^3$ arithmetic operations. 
    \item In Step 3, we compute each $x_m$ by adding $x_{m,k}$ for all $k \in [n]$. Thus each $x_m$ requires $n$ arithmetic operations and hence, this step requires $n^2$ arithmetic operations. 
    \item In Step 4, we compute each $\Tilde{s}_i$ by taking the inner product of the $i$-th column of $V$ and $X = (X_1,...,x_m)$. 
    Each inner product requires~$n$ arithmetic operations and hence, this step requires $n^2$ arithmetic operations.  
\end{enumerate}
So, the total number of arithmetic operations required is $T_{CB}(n) = O(n^3)$.
\textbf{Error Analysis:}
We denote by $A_k$ the $k$-th row of any matrix $A$ and by $A_{\_,k}$ we denote the $k$-th column of $A$.
\par
We proceed step by step and analyse the errors at every step of the algorithm. At every step, we explain what the ideal output would be if the algorithm was run in exact arithmetic. And we show that the output in finite arithmetic at every stage is quite close to the ideal output.  
\par
\textbf{Step 1: } Let $V$ be the matrix given as input. In this step, we want to compute a product of the matrices $V^T$ and  $V$.  We use the standard matrix multiplication algorithm and the bounds from (\ref{eq:slowmm}). Let $W  = MM(V^T, V)$ be the output of Step 1 of this algorithm. 
\par
Using (\ref{eq:slowmm}) and the fact that for any matrix $V$, $||V^T|| = ||V||$, we have: 
\begin{equation}\label{eq:V^TV}
    ||W - V^T V|| \leq 2n\cdot \textbf{u} \cdot ||V||_F^2.
\end{equation}
From (\ref{eq:V^TV}) and the triangle inequality, we also have that
\begin{equation}\label{eq:normW}
    ||W|| \leq ||V||_F^2 + 2n\textbf{u}||V||_F^2 < 2||V||_F^2.
\end{equation}
In the last inequality, we use the hypothesis that $2n\textbf{u} < 1$.
\par
\textbf{Step 2:} 
In this step, we take as input a matrix $W$ and compute all the diagonal elements of the matrix $WT_m$. Let $x_{m,k} = (WT_m)_{k,k}$ be computed on a floating point machine. If the algorithm is run in exact arithmetic, the output at the end of Step 2 is $(V^T V T_m)_{k,k}$ for all $m,k \in [n]$.
\par
Computationally, the $k$-th diagonal element can be computed as an inner product between the $k$-th row of $W$ and the $k$-th column of $T_m$. Then using the error bounds of inner product computation in (\ref{eq:ipbound}), we have that
\begin{equation}\label{eq:xmk}
\begin{split}
|x_{m,k} -   (WT_m)_{k,k}| \leq 2n \textbf{u} ||W_k|| ||(T_m)_{\_, k}|| \leq 2n \textbf{u} ||W|| ||(T_m)_{\_, k}||.
\end{split}
\end{equation}
Also, from (\ref{eq:V^TV}), we have that
\begin{equation}\label{eq:xmk2}
\begin{split}
    |(WT_m)_{k,k} - (V^T V T_m)_{k,k} | &\leq |\langle (W-V^TV)_k, (T_m)_{\_,k} \rangle| \\
    &\leq ||(W-V^TV)_k|| ||(T_m)_{\_,k}|| \\
    &\leq 2n \textbf{u} ||V||_F^2 ||(T_m)_{\_,k}||.
\end{split}
\end{equation}
Combining (\ref{eq:xmk}) and (\ref{eq:xmk2}), 
the triangle inequality and the bound from (\ref{eq:normW}), we deduce that
\begin{equation}\label{eq:step2fastcobslices}
    |x_{m,k} - (V^T V T_m)_{k,k}| \leq 6n \textbf{u} ||V||_F^2 ||(T_m)_{\_,k}||. 
\end{equation}
We also want to give an upper bound for $|x_{m,k}|$. 
By~(\ref{eq:step2fastcobslices}) and the triangle inequality, we have
\begin{equation}\label{eq:normxmk1}
    |x_{m,k}| \leq 6n \textbf{u} ||V||_F^2 ||(T_m)_{\_, k}|| + |(V^TVT_m)_{k,k}| 
\end{equation}
So we need to give an upper bound for $|(V^TVT_m)_{k,k}| $. Expanding along the definition and using the Cauchy-Schwarz inequality, we obtain 
\begin{equation}\label{eq:V^TVTmkk}
    |(V^T V T_m)_{k,k}| = |\langle (V^TV)_k , (T_m)_{\_,k} \rangle| \leq || (V^TV)_k|| ||(T_m)_{\_,k}|| . 
\end{equation}
Putting this back in (\ref{eq:normxmk1}), we have that
\begin{equation}\label{eq:normxmk}
   |x_{m,k}| \leq \Big(6n\textbf{u} + 1\Big) ||V||_F^2 ||(T_m)_{\_,k}|| \leq 2 ||V||_F^2 ||(T_m)_{\_,k}||. 
\end{equation}
The final inequality uses the hypothesis that $\textbf{u} < \frac{1}{6n}$. 

\textbf{Step 3:} In this step, we take as input $x_{m,k}$ for all $m,k \in [n]$. We then compute $x_m = \sum_{k=1}^n x_{m,k}$ on a floating point machine for all $m \in [n]$. If the algorithm was run in exact arithmetic, the output at the end of this step would be $\sum_{k=1}^n (V^T V T_m)_{k,k}$ for all $m \in [n]$.
\par
Computation of $x_m$ can also be thought of as inner product between the all $1$'s vector $\Bar{1}$ and the vector $(x_{m,1},...,x_{m,n})$. So, we can again use the bounds from (\ref{eq:ipbound}). This gives us that
\begin{equation}\label{eq:X_m}
\begin{split}
|x_m -  \sum_{k=1}^n x_{m,k} | &\leq 2n^{\frac{3}{2}} \textbf{u} \sqrt{\sum_{k=1}^n |x_{m,k}|^2} \\
&\leq 4n^{\frac{3}{2}} ||V||_F^2\textbf{u}\sqrt{\sum_{k=1}^n ||(T_m)_{\_, k}||^2} .
\end{split}
\end{equation}
The last equation uses (\ref{eq:normxmk}) to bound the norm of $|x_{m,k}|$. Also, summing up (\ref{eq:step2fastcobslices}) for all $k \in [n]$ and using the triangle inequality, we have that
\begin{equation}\label{eq:X_m2}
\begin{split}
      |\sum_{k=1}^n x_{m,k} - \sum_{k=1}^n (V^T V T_m)_{k,k}| &\leq \textbf{u} ||V||_F^2 (4n + \mu_{MM}(n))\Big(\sum_{k=1}^n ||(T_m)_{\_,k}||\Big) \\
     &\leq \textbf{u} ||V||_F^2 (6n^{\frac{3}{2}})\Big(\sum_{k=1}^n ||(T_m)_{\_,k}||^2\Big).
\end{split}
\end{equation}
The last inequality follows from the Cauchy-Schwarz inequality. Combining (\ref{eq:X_m}) and (\ref{eq:X_m2}), we finally have the error at the end of  that
\begin{equation}\label{eq:step3fastcobslices}
\begin{split}
    |x_m - \sum_{k=1}^n (V^T V T_m)_{k,k}| &\leq \textbf{u} ||V||_F^2 10n^{\frac{3}{2}}\sqrt{\sum_{k=1}^n ||(T_m)_{\_,k}||^2} \\
    &=  \textbf{u} ||V||_F^2 10n^{\frac{3}{2}}||T_m||_F.    
\end{split}
\end{equation}
In the last equality, we use the definition of the Frobenius norm of matrices from (\ref{eq:frobeniusnorm}).
\par
We also want to derive bounds for $|x_m|$. From the previous equation, by the  triangle inequality we already get that
\begin{equation}\label{eq:normXm1}
    |x_m| \leq \textbf{u} ||V||_F^2 10n^{\frac{3}{2}}||T_m||_F + |\sum_{k=1}^n (V^T V T_m)_{k,k}|.
\end{equation}
So it is enough to derive bounds for $|\sum_{k=1}^n (V^T V T_m)_{k,k}|$. Summing up (\ref{eq:V^TVTmkk}) for all $m \in [n]$ and using the Cauchy-Schwarz inequality,
we obtain:
\begin{equation}\label{eq:normsumvTvTmkk}
\begin{split}
    |\sum_{k=1}^n (V^T V T_m)_{k,k}| &\leq \sqrt{n} \sqrt{\sum_{k=1}^n |(V^T V T_m)_{k,k}|^2}   \\
    &\leq \sqrt{n}||V||_F^2 ||T_m||_F.
\end{split}
\end{equation}
Putting this back in (\ref{eq:normXm1}), we have that
\begin{equation}\label{eq:normXm}
    |x_m| \leq 2\sqrt{n}||V||_F^2 ||T_m||_F.
\end{equation}
Here in the last inequality, we use the hypothesis that $\textbf{u} \leq \frac{1}{10n}$.
\textbf{Step 4:} In this step, we take as input $x_{m}$ for all $m \in [n]$. We then compute $\Tilde{s}_i = \sum_{m=1}^n v_{mi} x_m$ in floating point arithmetic. Recall that $S = (V \otimes V \otimes V).T$ and $S_1,...,S_n$ are the slices of $S$. Ideally if the algorithm is run in exact arithmetic, the output at this stage is $Tr(S_i) = \sum_{m=1}^n v_{mi} \Big(\sum_{k=1}^n (V^T V T_m)_{k,k}\Big)$.
\par
Using error bounds for the inner product operation (\ref{eq:ipbound}) and using (\ref{eq:normXm}) to bound $|x_m|$ we have that
\begin{equation}\label{eq:step4fastcobslices1}
\begin{split}
    |\Tilde{s}_i - \sum_{m=1}^n v_{mi} x_m| &\leq 2n \textbf{u} ||V_{\_,i}|| \sqrt{\sum_{m=1}^n |x_m|^2}\\
    &\leq 4n^{\frac{3}{2}} \textbf{u} ||V||_F^3 \sqrt{\sum_{m=1}^n ||T_m||_F^2}.
\end{split}
\end{equation}
Also, summing up (\ref{eq:step3fastcobslices}) for all $m \in [n]$ and using the triangle inequality, we have: 
\begin{equation}\label{eq:step4fastcobslices2}
\begin{split}
    |\sum_{m=1}^n v_{mi} \Big(x_m - \sum_{k=1}^n (V^T V T_m)_{k,k}\Big) | &\leq ||V_{\_,i}||\sqrt{\sum_{m=1}^n |x_m - \sum_{k=1}^n (V^T V T_m)_{k,k}|^2}     \\
    &\leq \textbf{u}||V||_F^3 (8n^{\frac{3}{2}} + \sqrt{n}\mu_{MM}(n))\sqrt{\sum_{m=1}^n ||T_m||_F^2}
\end{split}
\end{equation}
Moreover, it follows from (\ref{eq:tensornormslices}), that $\sum_{m=1}^n ||T_m||_F^2 = ||T||^2_F$. Using this and combining (\ref{eq:step4fastcobslices1}) and (\ref{eq:step4fastcobslices2}) using triangle inequality, we have: 
\begin{equation}
\begin{split}
    &||\Tilde{s}_i - \sum_{m=1}^n v_{mi} \Big(\sum_{k=1}^n (V^T V T_m)_{k,k}\Big)| \\
    &\leq \textbf{u}||V||_F^3 14n^{\frac{3}{2}}\sqrt{\sum_{m=1}^n ||T_m||_F^2} \\
    &= \mu_{CB}(n) \cdot \textbf{u} \cdot ||V||_F^3||T||_F,
\end{split}
\end{equation}
where $\mu_{CB}(n) \leq 14n^{\frac{3}{2}}$.
\end{proof}

\section{Diagonalisation algorithm for diagonalisable matrices}\label{sec:diag}

In their recent breakthrough result,
\cite{9317903}  gave a numerically stable algorithm for matrix diagonalisation that also runs in matrix multiplication time in the finite precision arithmetic model. { As explained in Section~\ref{sec:overview}, 
we address two related issues in this section:
\begin{itemize}
    \item[(i)] Strengthening the conditioning guarantee from  \cite{9317903}  on
    the similarity $V$ that approximately diagonalises 
    the input matrix $A$.
    \item[(ii)] Relaxing the assumption $||A|| \leq 1$ on the input.
\end{itemize}
Our contribution regarding (i) appears in Theorem~\ref{thm:eig} and comes at the expense 
of additional assumptions on $A$: this matrix must be diagonalisable with distinct eigenvalues. The bounds in that theorem are expressed as a function of the condition number of the eigenproblem~(\ref{kappaeig}), already defined in~\cite{9317903}, and of the Frobenius eigenvector 
condition number~(\ref{kappaf}).
Regarding (ii), 
we need to  slightly modify the algorithm from~\cite{9317903}  in order to scale the input matrix. The main result of this section is Theorem~\ref{thm:eigfwd}, where we combine (i) and (ii). In particular, the error analysis due to the scaling of $A$ is worked out in the proof of Theorem~\ref{thm:eigfwd}.
}
\par
\begin{definition}[\textbf{Eigenpair and eigenproblem}] (Section 1.1 of \cite{9317903})
An eigenpair of a matrix $A \in \mathbb{C}^{n \times n}$ is a tuple $(\lambda,v) \in \mathbb{C} \times \mathbb{C}^n $ such that $Av= \lambda v$ and $||v||_2 = 1$. The eigenproblem is the problem of finding a maximal set of linearly independent eigenpairs $(\lambda_i,v_i)$ of a given matrix $A$. Note that an eigenvalue may appear more than once if it has geometric multiplicity greater than one. In the case when $A$ is diagonalizable, the solution consists of exactly $n$ eigenpairs, and if $A$ has distinct eigenvalues, then the solution is unique, up to the phases of $v_i$.
\end{definition}
\begin{definition}[$\delta$-\textbf{forward approximation for the eigenproblem}]
\label{def:forward}
Let $(\lambda_i,v_i)$ be true eigenpairs for a diagonalizable matrix $A$. Given an accuracy parameter $\delta$, the problem is to find pairs $(\lambda_i',v'_i)$ such that $||v_i - v'_i|| \leq \delta$ and $|\lambda_i - \lambda_i'| \leq \delta$ i.e., to ﬁnd a solution close to the exact solution.
\end{definition}
\begin{definition}[$\delta$-\textbf{backward approximation for the eigenproblem}]
\label{def:backward}
Given a diagonalizable matrix $A$ and an accuracy parameter $\delta$, find exact eigenpairs $(\lambda_i',v'_i)$ for a matrix $A'$ such that $||A - A'|| \leq \delta$ i.e., ﬁnd an exact solution to a nearby problem. Since diagonalizable matrices are dense in $\mathbb{C}^{n \times n}$, one can always find a complete set of eigenpairs for some nearby $A'$.
\end{definition}

{\bf Condition numbers.} If $A$ is diagonalizable, we define following Equation 2 in ~\cite{9317903}, its {\em eigenvector condition number}:
\begin{equation} \label{eq:eveccndtnnum}
    \kappa_V(A)=\inf_V ||V|| \cdot ||V^{-1}||,
\end{equation}
where the infimum is over all invertible $V$ such that $V^{-1}AV$ is diagonal. Its minimum eigenvalue gap is defined as
\begin{equation}\label{eq:gap}
    \text{gap}(A) := \text{min}_{i \neq j} |\lambda_i(A) - \lambda_j(A)|,
\end{equation}
where $\lambda_i$ are the eigenvalues of $A$ (with multiplicity).
Instead of the eigenvector condition number, it is sometimes more convenient to work instead
with the \textit{Frobenius eigenvector condition number}
\begin{equation} \label{kappaf}
\kappa^F_V(A)=\inf_V (||V||^2_F+ ||V^{-1}||^2_F) = \inf_V \kappa_F(V),
\end{equation}
where the infimum is taken over the same set of invertible matrices. 
We always have $\kappa^F_V(A) \geq 2 \kappa_V(A)$.
Following Equation 3 of \cite{9317903}, we define the condition number of the eigenproblem to be:
\begin{equation} \label{kappaeig}
    \kappa_{\text{eig}}(A) := \frac{\kappa_V(A)}{\text{gap}(A)} \in [0,\infty].
\end{equation}
\begin{lemma}\label{lem:scaling}
Suppose that $A$ has $n$ distinct eigenvalues $\lambda_1,\ldots,\lambda_n$, with $v_1,\ldots,v_n$ the corresponding eigenvectors. Let $W$ be the matrix with columns $v_1,\ldots,v_n$; let $u_1,\ldots,u_n$ be the left eigenvectors of $A$, i.e., the rows of $W^{-1}$. Then $\kappa^F_V(A)=2\sum_{i=1}^n ||u_i|| \cdot ||v_i||,$ and the infimum in~(\ref{kappaf}) is reached for
the matrix $V$ obtained from $W$ by multiplication of each column by $\sqrt{||u_i|| / ||v_i||}.$
\end{lemma}
\begin{proof}
Since $A$ has distinct eigenvalues, any matrix $V$ that diagonalizes $A$ is obtained from $W$ by multiplication of each column by some nonzero scalar~$x_i$. In matrix notation, we have $V=WD$ where $D=\text{diag}(x_1,...,x_n)$. 
We also have $V^{-1}=D^{-1}W^{-1}$, and the $i$-th row of $V^{-1}$ is therefore equal to $u_i/x_i$.
As a result,
$$||V||^2_F+ ||V^{-1}||^2_F = \sum_{i=1}^n \left(x_i^2 ||v_i||^2+\frac{||u_i||^2}{x_i^2}\right).$$
An elementary computation shows that the infimum is reached for $x_i=\sqrt{||u_i|| / ||v_i||}.$ Here we have assumed that $x_i \in \R$ for all $i$. This is without loss of generality since multiplying each entry of $V$ or $V^{-1}$ by a complex number of modulus 1 does not change their Frobenius norms.
\end{proof}

If $M$ is diagonalisable as $VDV^{-1}$ over $\C$, let $v_i$ be the columns of $V$ and $u_j^T$ be the rows of $V^{-1}$. Then $M$ admits a spectral expansion of the form
\begin{equation} \label{eq:spectralexpansion}
    M = \sum_{i=1}^n \lambda_iv_iu_i^*.  
\end{equation}
\begin{definition}\label{def:eigconditionnumber}
For $M \in M_n(\C)$, if $M$ has distinct eigenvalues $\lambda_1,...,\lambda_n$ and a spectral expansion as in (\ref{eq:spectralexpansion}), then we define the eigenvalue condition number of $\lambda_i$ as 
\begin{align*}
    \kappa(\lambda_i) := ||v_iu_i^*|| = ||v_i||||u_i||
\end{align*}
\end{definition}
\begin{remark}\label{rem:kappaind}
 Note that $\kappa(\lambda_i)$ is independent of the choice of the $v_i$'s. This follows from the fact that  $M$ has distinct eigenvalues: if $v_i',u_i'$ is another pair of vectors corresponding to $\lambda_i$ in the spectral expansion, then $v_i' = k_iv_i$ and $u_i' = l_iv_i$ for some non-zero constants $k_i,l_i$. Using the fact that $\langle v_i', u_i'\rangle = 1$, we have $k_il_i = 1$. Hence $||u_i'||||v_i'|| = ||u_i||||v_i||$ which proves that $\kappa(\lambda_i)$ is indeed independence of the choice of the $v_i$'s.
\end{remark} 
\begin{lemma}\label{lem:bgvkslemma}
Let $M,M'$ be $n \times n$ matrices such that $||M|| , ||M'|| \leq 1$ and $||M - M'|| \leq \delta$ where $\delta < \frac{1}{8\kappa_{ \text{eig}}(M)}$. Let $\lambda_1,...,\lambda_n$ be the distinct eigenvalues of $M$. Then
\begin{enumerate}
    \item $M'$ has distinct eigenvalues.
    \item $|\kappa(\lambda_i) - \kappa(\lambda'_i)| \leq 2 \kappa_V(A) $.
    \item $\sqrt{n\sum_{i \in [n]} \kappa(\lambda_i)^2} \leq n\kappa_V(A)$.
\end{enumerate}
\end{lemma}
\begin{proof}
    Refer to the proof of Proposition 1.1 in \cite{9317903} for a proof of the first two properties. 
     Towards the third one, we first show that $\kappa_V(A) \geq \kappa(\lambda_i)$ for all $i \in [n]$. Using the definition of $\kappa_V(A)$, we get that 
    \begin{align*}
        \kappa_V(A) &= \text{inf}_{V \in D(A)} ||V||||V^{-1}|| \\
        &= \text{inf}_{V \in D(A)} (\max_{x \in \C^n \setminus \{0\}} \frac{||Vx||}{||x||})(\max_{x \in \C^n \setminus \{0\}} \frac{||V^{-1}x||}{||x||}) \\
        &\geq \text{inf}_{V \in D(A)} ||v_i|| ||u_i||
    \end{align*}
    where $v_i$ are the columns of $V$ and $u_i$ are the rows of $V^{-1}$. The inequality follows from the fact that $\max_{x \in \C^n} \frac{||Vx||}{||x||} \geq ||Ve_i|| = ||v_i||$. Since $||V|| = ||V^T||$ for all matrices $V$, $\max_{x \in \C^n \setminus \{0\}} \frac{||V^{-1}x||}{||x||} = \max_{x \in \C^n \setminus \{0\}} \frac{||V^{-T}x||}{||x||} \geq ||V^{-T}e_i|| = ||u_i||$. By Remark \ref{rem:kappaind}, $||v_i|| ||u_i||$ is equal for all $V \in D(A)$. 
    As a result, $\kappa_V(A) \geq ||v_i|| ||u_i|| = \kappa(\lambda_i)$ for all $i \in [n]$. 
    This gives us 
    \begin{align*}
        \sqrt{n\sum_{i \in [n]} \kappa(\lambda_i)^2} \leq \sqrt{n\sum_{i \in [n]} \kappa_V(A)^2} = n\kappa_V(A).
    \end{align*}
\end{proof}
\begin{lemma}\label{lem:gapAA'}
Let $A,A' \in M_n(\C)$  be such that $A$ has $n$ distinct eigenvalues and $||A-A'|| \leq \delta$ where $\delta < \frac{1}{8\kappa_{\text{eig}}(A)}$. Then 
\begin{align*}
    \text{gap}(A') \geq \frac{3\text{gap}(A)}{4}.
\end{align*}
\end{lemma}
\begin{proof}
Refer to the proof of Proposition 1.1 in \cite{9317903}.
\end{proof}
\begin{lemma}\label{lem:AA'relation}
Let $A,A' \in M_n(\C)$  be such that $A$ has $n$ distinct eigenvalues and $||A-A'|| \leq \delta$ where $\delta < \frac{1}{8\kappa_{\text{eig}}(A)}$. Then
\begin{align*}
    \kappa^F_V(A') \leq 6n\kappa_V(A) \leq  3n\kappa^F_V(A).
\end{align*}
\end{lemma}
\begin{proof}
We first explain the proof for $||A||,||A'|| \leq 1$ and then modify it to deal with the general case.
Suppose that $\lambda_1,\ldots,\lambda_n$ are the eigenvalues of $A$ with corresponding eigenvectors $v_1,\ldots,v_n$. By Lemma \ref{lem:bgvkslemma}, $A'$ has distinct eigenvalues as well. Let $\lambda'_1,\ldots,\lambda'_n$ be the eigenvalues of $A'$ with corresponding eigenvectors $v'_1,\ldots,v'_n$. Let $W'$ be the matrix with columns $v'_1,\ldots,v'_n$; let $u'_1,\ldots,u'_n$ be the left eigenvectors of $A'$, i.e., the rows of $W^{-1}$. Applying Lemma \ref{lem:scaling} to $A'$, we know that
\begin{equation}\label{eq:scaling}
    \kappa^F_V(A') = 2\sum_{i=1}^n ||u_i'|| \cdot ||v_i'||.
\end{equation}
From Lemma \ref{lem:bgvkslemma}, we know that
\begin{align*}
    |\kappa(\lambda'_i) - \kappa(\lambda_i)| \leq 2\kappa_V(A)
\end{align*}
and hence
\begin{align*}
    \kappa(\lambda'_i) \leq \kappa(\lambda_i) + 2\kappa_V(A).
\end{align*}
Adding this up for all $i = 1$ to $n$,
\begin{align*}
    \sum_{i=1}^n \kappa(\lambda'_i) \leq \sum_{i=1}^n\kappa(\lambda_i) + 2n\kappa_V(A).
\end{align*}
Now using the definition of $\kappa(\lambda_i)$ as in Definition \ref{def:eigconditionnumber}, we get that
\begin{align*}
    \sum_{i=1}^n||v_i'||||u_i'|| \leq \sum_{i=1}^n\kappa(\lambda_i) + 2n\kappa_V(A)
\end{align*}
By the  Cauchy-Schwarz inequality, we have 
\begin{align*}
    \sum_{i=1}^n||v_i'||||u_i'|| \leq \sqrt{n\sum_{i \in [n]} \kappa(\lambda_i)^2} + 2n\kappa_V(A).
\end{align*}
By (\ref{eq:scaling}) and Lemma \ref{lem:bgvkslemma}, we get that
\begin{align*}
    \kappa^F_V(A') \leq 6n\kappa_V(A) \leq 3n\kappa^F_V(A).
\end{align*}
\par
Let us now take $A,A'$ to be arbitrary $n \times n$ matrices over $\mathbb{C}$ such that $||A-A'|| \leq \delta$ for some $\delta \in (0,\frac{1}{8\kappa_{\text{eig}}}(A))$. Let $N_0 := \max\{||A||,||A'||\}$. 
Then we can define $C = \frac{C}{N_0}$ and $C' = \frac{A'}{N_0}$ where 
$||C||,||C'|| \leq 1$. Also notice that $||C - C'|| \leq \delta' = \frac{\delta}{N_0}$ where $\delta' < \frac{1}{8N_0\kappa_{\text{eig}}(A)}$. Now, $\kappa_{\text{eig}}(C) = \frac{\kappa_V(C)}{\text{gap}(C)}$. Since $C = \frac{A}{N_0}$, we get that $\text{gap}(C) = \frac{\text{gap}(A)}{N_0}$ and $\kappa_V(C) = \kappa_V(A)$. Hence, $\kappa_{\text{eig}}(C) = N_0\kappa_{\text{eig}}(A)$ and this gives us $\delta' <  \frac{1}{8\kappa_{\text{eig}}(C)}$. Using the previous argument, we have $\kappa^F_V(C') \leq 6n\kappa_V(C) \leq  3n\kappa^F_V(C)$. Since scaling of matrices preserves $\kappa^F_V$ and $\kappa_V$, this gives us that $\kappa^F_V(A') \leq 6n\kappa_V(A) \leq 3n\kappa^F_V(A)$.
\end{proof}
\begin{lemma}\label{lem:kappaVkappaA}
Let $A \in M_n(\C)$ be a diagonalisable matrix with distinct eigenvalues and let $A = VDV^{-1}$ such that for all $i \in [n]$, for each column $v_i$ of $V$, $\Big|||v_i|| - 1\Big| \leq \delta$. Then $\kappa_F(V) \leq n(1 + \delta)^2 + \frac{(\kappa^F_V(A))^2}{4(1-\delta)^2}$.
\end{lemma}
\begin{proof}
By Lemma \ref{lem:scaling}, if $U = V^{-1}$ and $u_1,...,u_n$ are the rows of $U$, then $\kappa^F_V(A) = \sum_{i \in [n]} 2||u_i||||v_i||$. Since $|||v_i|| - 1| \leq \delta$ for all $i \in [n]$, we have that $(1 - \delta)\sum_{i \in [n]} 2||u_i|| \leq \kappa^F_V(A) \leq  (1+ \delta)\sum_{i \in [n]} 2||u_i||$. From the definition of~$\kappa_F$,
\begin{align*}
    \kappa_F(V) = ||V||_F^2 + ||V^{-1}||_F^2  &= n(1+\delta)^2 + \sum_{i \in [n]} ||u_i||^2 \\
    &\leq n(1+\delta)^2 + (\sum_{i \in [n]} ||u_i||)^2 \leq  n(1+\delta)^2 + \frac{(\kappa^F_V(A))^2}{4(1-\delta)^2}.
\end{align*}
\end{proof}
{ In the next theorem we give some properties
of the diagonalisation algorithm $EIG$ analyzed in~\cite{9317903}. The first two are from their paper (Theorem 1.6),
and the third one provides an additional conditioning guarantee for $V$. It is especially useful for small values of $\delta$.}
\begin{theorem}\label{thm:eig}
There is a randomized algorithm $EIG$ which on input a diagonalisable matrix $A \in \C^{n \times n}$ with $||A|| \leq 1$ and a desired accuracy parameter $\delta \in (0, \frac{1}{8\kappa_{\text{eig}}(A)})$ outputs a diagonal $D$ and an invertible $V$. The following properties are satisfied by the output matrices:
    \begin{enumerate}
    \item $||A - VDV^{-1}|| \leq \delta \text{ and } \kappa(V) \leq \frac{32n^{2.5}}{\delta}$.
    \item $||v_i|| = 1 \pm n\textbf{u}$ where $v_i$ are the columns of $V$.
    \item $\kappa(V) \leq \frac{\kappa_F(V)}{2}  \leq \frac{1}{2}(\frac{9n}{4} + 9n^2(\kappa^F_V(A))^2)$.
\end{enumerate}
The algorithm runs in $$O(T_{\text{MM}}(n)\log^2\frac{n}{\delta})$$ arithmetic operations on a floating point machine with $$\log(\frac{1}{u}) = O(\log^4(\frac{n}{\delta})\log n)$$ bits of precision with probability at least $1 - \frac{1}{n} - \frac{12}{n^2}$.
\end{theorem}
\begin{proof}
The first two properties are from \cite{9317903} (see in particular Theorem 1.6 for the first one). 
    From the second property and
   from Lemma \ref{lem:kappaVkappaA} applied to $A' = VDV^{-1}$ for $\delta = n\textbf{u}$, we have $$\kappa(V) \leq \frac{\kappa_F(V)}{2} \leq \frac{1}{2}( n(1+n\textbf{u})^2 + \frac{(\kappa_F(A'))^2}{4(1- n\textbf{u})^2}).$$ Since $\delta < \frac{1}{8\kappa_{\text{eig}}(A)}$, it follows from Lemma \ref{lem:AA'relation} that $\kappa^F_V(A') \leq 3n\kappa^F_V(A)$ and this gives us that $\kappa(V) \leq \frac{\kappa_F(V)}{2} \leq \frac{1}{2}( n(1+n\textbf{u})^2 + \frac{(9n^2\kappa^F_V(A))^2}{4(1- n\textbf{u})^2})$. Since $n\textbf{u} < \frac{1}{2}$, this implies that $$\kappa(V) \leq \frac{\kappa_F(V)}{2}  \leq \frac{1}{2}(\frac{9n}{4} + 9n^2(\kappa^F_V(A))^2).$$
\end{proof}
{ In the remainder of this section, we relax the hypothesis $||A|| \leq 1$ on the input matrix.}

\begin{theorem}[Proposition 1.1 from \cite{9317903}]\label{thm:backwardtofwd}
If $||A||, ||A'|| \leq 1$, $||A - A'|| \leq \delta$, $\delta < \frac{gap(A)}{8\kappa_V(A)}$ and $\{(v_i,\lambda_i)\}_{i \in [n]}$, $\{(v'_i,\lambda'_i)\}_{i \in [n]}$ are eigenpairs of $A,A'$ , then
\begin{align*}
    ||v_i - v_i'|| \leq 6n\kappa_{\text{eig}}(A)\delta \text{ and } | \lambda_i' - \lambda_i| \leq \kappa_V(A)\delta \leq 2\kappa_{\text{eig}}(A)\delta \text{ for all } i \in [n],
\end{align*}
where $v_i$'s are given up to multiplication by phases.
\end{theorem}
Note here that by "phases" we mean complex numbers of norm $1$.
\par
In the following corollary, we show that Theorem \ref{thm:backwardtofwd} can be extended to matrices with norm greater that $1$.
\begin{corollary}\label{corr:prop1.1new}
For matrices $A,A' \in M_n(\C)$, if $||A-A'|| \leq \delta$, $\{(v_i,\lambda_i)\}_{i \in [n]}$, $\{(v'_i,\lambda'_i)\}_{i \in [n]}$ are eigenpairs of $A,A'$ respectively and $\delta < \frac{gap(A)}{8\kappa_V(A)}$, then $$||v_i - v_i'|| \leq 6n\kappa_{\text{eig}}(A)\delta \text{ and } |\lambda_i - \lambda_i'| \leq \kappa_V(A)\delta \text{ for all } i \in [n]$$ where the $v_i$'s are given up to multiplication by phases.
\end{corollary}
\begin{proof}
Let $N_0 = \max\{||A||,||A'||\}$. Let $C = \frac{A}{N_0}$ and $C' = \frac{A'}{N_0}$. Then $||C||, ||C'|| < 1$ and taking $\delta' = \frac{\delta}{N_0}$, we get that $||C-C'|| \leq \delta'$ where $\delta' < \frac{1}{8N_0\kappa_{\text{eig}}(A)} = \frac{1}{8\kappa_{\text{eig}}(C)}$. Applying Theorem \ref{thm:backwardtofwd}, we have  $||u_i - u_i'|| \leq 6n\kappa_{\text{eig}}(C)\delta'$ where the $u_i$ are the eigenvectors of $C$ after possibly multiplying $u_i$ by phases. Using $\kappa_{\text{eig}}(C) = N_0\kappa_{\text{eig}}(A)$ gives us that $||u_i - u_i'|| \leq 6n\kappa_{\text{eig}}(A)\delta$. Since the eigenvectors remain unchanged after scaling the matrix by a constant, this implies that $||v_i - v_i'|| \leq 6n\kappa_{\text{eig}}(A)\delta$. 
\par
If $\mu_1,...,\mu_n$ are the corresponding eigenvalues of $C$ and $\mu'_1,...,\mu'_n$ are the corresponding eigenvalues of $C'$, then using Theorem \ref{thm:backwardtofwd}, we already have that $|\mu_i - \mu'_i| < \kappa_V(C)\delta'$. Since $C= \frac{A}{N_0}$ and $C' = \frac{A'}{N_0}$, this implies that the eigenvalues are also similarly scaled. More formally, $\mu_i = \frac{\lambda_i}{N_0}$ and $\mu'_i = \frac{\lambda'_i}{N_0}$ for all $i \in [n]$. Hence, multiplying both sides by $N_0$ gives us that  $N_0|\mu_i - \mu'_i| < \kappa_V(C)N_0\delta'$, hence $|\lambda_i - \lambda_i'| < \kappa_V(C)\delta$. Since $\kappa_V(C) = \kappa_V(A)$, we finally conclude that $|\lambda_i - \lambda_i'| < \kappa_V(A)\delta $.
\end{proof}
We will use the above bound on the eigenvalues in Section \ref{sec:step4}.
We now present the algorithm for computing a forward approximation to the eigenvectors of a diagonalisable matrix.

\begin{algorithm}[H] \label{algo:eigfwd}
\SetAlgoLined
\nonl \textbf{Input:} A diagonalisable matrix $A$ { with distinct eigenvalues}, estimates $K_{\text{norm}} > \max\{||A||_F,1\}$ and $K_{\text{eig}} > \kappa_{\text{eig}}(A)$, desired accuracy parameter $\delta$.
\begin{enumerate}
\item Compute $B' = \frac{A}{2K_{\text{norm}}}$ on a floating point machine. 
\item Compute $\delta' = \frac{\delta}{64nK_{\text{norm}}K_{\text{eig}}}$. 
\item 
Let $(W,D_0)$ be the output of $\text{EIG}(B',\delta')$. 
\item Output the columns $w_1,...,w_n$  of $W$.
\end{enumerate}
\caption{Forward approximation of the eigenvectors of a matrix (EIG-FWD)}
\end{algorithm}
Here we assume at step 2  that the parameter $\delta'$ is computed without any roundoff error. As in \cite{9317903}, this will be done for simplicity throughout the paper in computations whose size does not grow with $n$. Wherever we mention that the computation is done on a floating point machine, we assume that there is an adversarial error associated with that computation.
The next theorem is the main result of this section.
\begin{theorem}\label{thm:eigfwd}
Given a diagonalisable matrix $A \in M_n(\C)$, a desired accuracy parameter $ \delta \in~(0,\frac{1}{2})$ and estimates $K_{\text{norm}} > \max\{||A||_F,1\}$ and $K_{\text{eig}} > \kappa_{\text{eig}}(A)$ as input, Algorithm \ref{algo:eigfwd} outputs vectors $w_1,...,w_n \in \C^n$  such that the following properties are satisfied with probability at least $1 - \frac{1}{n} - \frac{12}{n^2}$:
\begin{itemize}
    \item If $v^{(0)}_1,...,v^{(0)}_n$ are the true normalized eigenvectors of $A$, then we have  $||v^{(0)}_i - w_i|| < \delta$ up to multiplication by phases.
    \item Let $W$ be the matrix with columns $w_1,...,w_n$. Then $$\kappa(W) \leq \frac{\kappa_F(W)}{2}  \leq \frac{1}{2}(\frac{9n}{4} + 81n^4(\kappa^F_V(A))^2).$$
\end{itemize}
The algorithm requires $$O(T_{\text{MM}}(n)\log^2\frac{nK_{\text{eig}}K_{\text{norm}}}{\delta})$$ arithmetic operations on a floating point machine with $$O(\log^4(\frac{nK_{\text{eig}}K_{\text{norm}}}{\delta})\log n)$$ bits of precision. 
\end{theorem}
\begin{proof}
Let $v^{(0)}_1,...,v^{(0)}_n$ be the true normalized eigenvectors of $A$.
By Theorem \ref{thm:eig},  we need $O(\log^4 (\frac{n}{\delta'})\log n)$ bits of precision to run $\text{EIG}$ in step~$3$. So, we assume that the number of bits of precision available for this algorithm is $\log(\frac{1}{\textbf{u}}) := c\log^4 (\frac{n}{\delta'})\log n$ for some constant $c > 1$. We define $B = \frac{A}{2K_{\text{norm}}}$. Then $||B|| \leq \frac{1}{2} < 1$. Following the notation of Section~\ref{sec:fparithmetic}, let $B' = \text{fl}(\frac{A}{2K_{\text{norm}}})$. Using (\ref{eq:multfl}), we know that 
\begin{equation}\label{eq:bb'}
    ||B - B'|| \leq \frac{\textbf{u}\cdot \sqrt{n}}{2}.
\end{equation}
Since $\textbf{u}\cdot \sqrt{n} \leq 1$, then $||B'|| \leq \frac{1}{2} + \frac{\textbf{u}\cdot \sqrt{n}}{2} < 1$. 
We first show that the conditions of Theorem \ref{thm:eig} are satisfied when we run $\text{EIG}$ on $(B',\delta')$. For this, we have to show that $\delta' \leq \frac{1}{8\kappa_{\text{eig}}(B')}$.
\par
Applying Lemma \ref{lem:boundsforu1} where $\log(\frac{1}{\textbf{u}}) = c\log^4(\frac{n}{\delta'})\log n$, we get that for large enough $n$, $\textbf{u}\sqrt{n} < n^2 \textbf{u} < \delta' < \frac{\delta}{4K_{\text{eig}}K_{\text{norm}}}$. This gives us:
$$\frac{\textbf{u}\cdot\sqrt{n}}{2} \leq  \frac{1}{8K_{\text{eig}}K_{\text{norm}}} < \frac{1}{8\kappa_{\text{eig}}(B)}.$$ Putting it back in (\ref{eq:bb'}), we also have that $||B - B'|| < \frac{1}{8\kappa_{\text{eig}}(B)}$. Now we can apply Lemma \ref{lem:AA'relation} to $B,B'$ and we obtain the inequality
\begin{equation}\label{eq:kappafbb'}
    \kappa_V(B') \leq \frac{\kappa^F_V(B')}{2} \leq 3n\kappa_V(B) \leq \frac{3n\kappa^F_V(B)}{2}.
\end{equation}
It from Lemma \ref{lem:gapAA'} that $\text{gap}(B') \geq \frac{3\text{gap}(B)}{4}$. This gives us that $$\frac{1}{K_{\text{eig}}} < \frac{1}{\kappa_{\text{eig}}(A)} = \frac{\text{gap}(A)}{\kappa_V(A)} = \frac{2K_{\text{norm}}\text{gap}(B)}{\kappa_V(B)}  \leq \frac{8nK_{\text{norm}}\text{gap}(B')}{\kappa_V(B')}.$$
Hence, $$\delta' = \frac{\delta}{64nK_{\text{eig}}K_{\text{norm}}} < \frac{\delta}{8\kappa_{\text{eig}}(B')} < \frac{1}{8\kappa_{\text{eig}}(B')},$$ and we can now run $\text{EIG}$ on $(B', \delta')$.
By Theorem \ref{thm:eig},  the algorithm therefore outputs $(W,D_0)$ such that 
\begin{equation}\label{eq:b'wd0w^-1}
    ||B' - WD_0W^{-1}|| \leq \delta'
\end{equation}
with probability at least $1 - \frac{1}{n} - \frac{12}{n^2}$. Applying the triangle inequality to (\ref{eq:bb'}) and (\ref{eq:b'wd0w^-1}) shows that $||B - WD_0W^{-1}|| \leq \frac{\textbf{u}\cdot \sqrt{n}}{2} + \frac{\delta}{64nK_{\text{eig}}K_{\text{norm}}}$. Since  $\textbf{u}\sqrt{n} \leq n^2\textbf{u}$, by Lemma \ref{lem:boundsforu1} and for large enough $n$ we have $\frac{\textbf{u}\cdot \sqrt{n}}{2} \leq \frac{\delta}{64nK_{\text{eig}}K_{\text{norm}}}$. This gives us that $||B - WD_0W^{-1}|| \leq \frac{\delta}{32nK_{\text{eig}}K_{\text{norm}}}$. Multiplying both sides by $2K_{\text{norm}}$,
we obtain $$||A - W(2K_{\text{norm}}D_0)W^{-1}|| \leq \frac{\delta}{16nK_{\text{eig}}}.$$ Let $A' = W(2K_{\text{norm}}D_0)W^{-1}$. We can now use Corollary \ref{corr:prop1.1new} for $A$ and $A'$ since $\frac{\delta}{16nK_{\text{eig}}} < \frac{1}{8nK_{\text{eig}}} < \frac{1}{8\kappa_{\text{eig}}(A)}$. Let $v_1,...,v_n$ be the eigenvectors of $A'$. Then there exists a phase $\rho_i$ such that
\begin{equation}\label{eq:evec1}
    ||v^{(0)}_i - \rho_i v_i|| \leq 6n\kappa_{\text{eig}}(A) \frac{\delta}{16nK_{\text{eig}}} < \frac{\delta}{2}.
\end{equation} 
 By Theorem \ref{thm:eig}, $||w_i|| = 1 \pm n\textbf{u}$ since $w_1,...,w_n$ are the columns of $W$. Note that the $w_i$'s are the eigenvectors of $A'$ as well. Since $A'$ has distinct eigenvalues, the eigenvectors are unique up to scaling by complex numbers. This along with the fact that $||v_i|| = 1$ gives us that there exists phase $\rho'_i$ such that
\begin{align*}
     ||v_i - (\rho'_i)^{-1}w_i|| = |n\textbf{u}| ||v_i|| < \frac{\delta}{2}.
\end{align*}
The final inequality comes from the fact that for $n >2$, $n\textbf{u} < \frac{n^2\textbf{u}}{2}$ and  we can therefore use Lemma \ref{lem:boundsforu1}. Now, multiplying by $|\rho_i|$ on both sides, we have 
\begin{equation}\label{eq:evec2}
     ||\rho_iv_i - \rho_i(\rho'_i)^{-1}w_i|| < \frac{\delta}{2}.
\end{equation}
Now, using the triangle inequality on (\ref{eq:evec1}) and (\ref{eq:evec2}), 
\begin{equation}\label{eq:maindiag1}
    ||v^{(0)}_i - \rho_i(\rho'_i)^{-1}w_i|| < \delta.
\end{equation}
\par
From Theorem \ref{thm:eig}, we get that $\kappa(W) \leq \frac{\kappa_F(W)}{2}  \leq \frac{1}{2}(\frac{9n}{4} + 9n^2(\kappa^F_V(B'))^2).$
By (\ref{eq:kappafbb'}) and the fact that $\kappa^F_V(A) = \kappa^F_V(B)$, we have
\begin{align*}
    \frac{1}{2}(\frac{9n}{4} + 9n^2(\kappa^F_V(B'))^2) \leq \frac{1}{2}(\frac{9n}{4} + 81n^4(\kappa^F_V(B))^2) = \frac{1}{2}(\frac{9n}{4} + 81n^4(\kappa^F_V(A))^2).
\end{align*}

\end{proof}
\section{Tensor decomposition for complete symmetric tensors in exact arithmetic}\label{sec:completeexactarithmetic}

{ Recall from the introduction the notion
of a diagonalisable tensor:}
\begin{definition}
We call a symmetric tensor $T \in \mathbb{C}^{n \times n \times n}$ diagonalisable if  $T = \sum_{i=1}^n u_i^{\otimes 3}$ where the $u_i$ are linearly independent.
\end{definition}
In other words, a symmetric tensor $T$ is diagonalisable if there exists an { invertible} change of basis $U$ that takes the diagonal tensor $\sum_{i \in [n]} e_i^{\otimes 3}$ to $T$, i.e., $T = (U \otimes U \otimes U). (\sum_{i \in [n]} e_i^{\otimes 3})$ for some $U \in \text{GL}_n(\mathbb{C})$.   
\par
Let $T \in (\C^n)^{\otimes 3}$ be a { diagonalisable tensor given as input.}
In this section we give an algorithm which returns the linearly independent  $u_i$'s in the decomposition of $T$, { up to multiplication by cube roots of unity.} 
This algorithm works in a computational model where all arithmetic operations can be done exactly and additionally, we can diagonalise a matrix exactly.

\begin{algorithm}[H] \label{algo:completeexact}
\SetAlgoLined
\nonl \textbf{Input:} An order-3 diagonalisable symmetric tensor $T \in \C^{n \times n \times n}$.  \\
\nonl \textbf{Output:} linearly independent vectors $l_1,...,l_r \in  \C^n$ such that $T = \sum_{i=1}^n l_i^{\otimes 3}$.\\
\nonl Pick $a_1,...,a_n$ and $b_1,...,b_n$ uniformly and independently at random from a finite set $S \subset \C$ such that $|S| \geq n^3$. \\
\nonl Let $T_1,...,T_n$ be the slices of $T$. \\
Set $T^{(a)} = \sum_{i=1}^n a_iT_i$ and $T^{(b)} = \sum_{i=1}^n b_iT_i$.  \\
Compute $T^{(a)'} = (T^{(a)})^{-1}$. \\
Compute $D = T^{(a)'}T^{(b)}$. \\
Compute the normalized eigenvectors $p_1,...,p_n$ of $D$. \\
Let $P$ be the matrix with $(p_1,...,p_n)$ as columns and compute $P^{-1}$. 
Let $v_i$ be the $i$-th row of $P^{-1}$.\\
Define $S = (P \otimes P \otimes P).T$ and let $S_1,...,S_n$ be the slices of $S$. Compute $\alpha_i = \text{Tr}(S_i)$. \\
Output $(\alpha_1)^{\frac{1}{3}}v_1,...,(\alpha_n)^{\frac{1}{3}}v_n$.
\caption{Tensor decomposition algorithm for complete symmetric tensors}
\end{algorithm}
{ Algorithm~\ref{algo:completeexact} is essentially the algorithm that was already presented in Section~\ref{sec:overview}.
As explained in that section, this is a symmetric version of the simultaneous diagonalisation algorithm where the coefficients
$\alpha_i$ are determined in a novel way (as the traces
of slices of a certain tensor). The algorithm will be analyzed in finite precision arithmetic in the two final sections of this paper.}

The remainder of this section is devoted to a correctness proof for Algorithm \ref{algo:completeexact} including an analysis of the probability of error. 
{ The main results of this section are 
Theorem \ref{thm:complexteexactmain} and Corollary~\ref{thm:Jennrichexact} 
(where we will see that one should 
take $S$ of size bigger than $n^3$ 
in Algorithm~\ref{algo:completeexact}).}
In that direction, we prove an intermediate lemma  showing that if $a_1,...,a_n$ and $b_1,...,b_n$ are picked at random from a finite set, then $T^{(a)}$ is invertible and the eigenvalues of $(T^{(a)})^{-1}T^{(b)}$ are distinct with high probability. 
\begin{theorem}\label{thm:eigenvalues}
Let $T= \sum_{i=1}^n u_i^{\otimes 3}$ where $u_i \in \C^n$ are linearly independent vectors. Let $T_1,...,T_n$ be the slices of $T$. Set $T^{(a)} = \sum_{i=1}^n a_iT_i$ and $T^{(b)} = \sum_{i=1}^n b_iT_i$ where $a_1,...,a_n,b_1,...,b_n$ are picked uniformly and independently at random from a finite set $S \subset \mathbb{K}$. If $T^{(a)}$ is invertible, let $T^{(a)'} = (T^{(a)})^{-1}$. Let $\lambda_1,...,\lambda_n$ be the eigenvalues of $T^{(a)'}T^{(b)}$. Then
\begin{align*}
    \text{Pr}_{a_1,...,a_n,b_1,...,b_n \in_r S}[T^{(a)} \text{ is invertible} \text{ and } \lambda_i \neq \lambda_j \text{ for all } i \neq j ] \geq 1- (\frac{2\binom{n}{2}}{|S|} + \frac{n}{|S|}).
\end{align*}
\end{theorem}
\begin{proof}
Let $U$ be the matrix with columns $u_1,...,u_n$. Since $T= \sum_{i=1}^n u_i^{\otimes 3}$, by Corollary \ref{corr:P3structural},  the slices $T_1,...,T_n$ of $T$ can be written as 
\begin{align*}
    T_i = U^TD_iU \text{ where } D_i = \text{diag}(u_{1,i},...,u_{n,i}).
\end{align*}
Taking $a = (a_1,...,a_n) \in \C^n$, this gives us that 
\begin{align*}
     T^{(a)} = U^TD^{(a)}U \text{ where } D^{(a)} = \text{diag}(\langle a,u_1 \rangle,...,\langle a,u_n \rangle).
\end{align*}
Similarly, $ T^{(b)} = U^TD^{(b)}U$ where $D^{(b)} = \text{diag}(\langle b,u_1 \rangle,...,\langle b,u_n \rangle)$. Now if $T^{(a)}$ is invertible, we can write
\begin{equation}\label{eq:formulata'tb}
    T^{(a)'}T^{(b)} = U^{-1} \Big(\text{diag}(\frac{\langle b,u_1 \rangle}{\langle a,u_1 \rangle},...,\frac{\langle b,u_n \rangle}{\langle a,u_n \rangle})\Big)U
\end{equation}
Hence, the eigenvalues of $T^{(a)'}T^{(b)}$ are $\lambda_i = \frac{\langle b,u_i \rangle}{\langle a,u_i \rangle}$. For all $i \neq j \in [n]$, we define the polynomial $$P_{ij}(x_1,...,x_n,y_1,...,y_n) = \langle y,u_i \rangle\langle x,u_j \rangle - \langle y,u_j \rangle\langle x,u_i \rangle = \sum_{k,l = 1}^n y_kx_l(u_{ik}u_{jl} - u_{il}u_{jk})$$ where $y = (y_1,...,y_n)$ and $x = (x_1,...,x_n)$.
Now $T^{(a)}$ is invertible iff $\langle a,u_i \rangle \neq 0$ for all $i \in [n]$. This gives us that
\begin{equation}
    \text{Pr}_{a_1,...,a_n \in_r S}[T^{(a)} \text{ is invertible}] = \text{Pr}_{a_1,...,a_n \in_r S}[\langle a,u_i \rangle \neq 0 \text{ for all } i \in [n]].
\end{equation}
For all $i \in [n]$, there exists $k \in [n]$ such that $u_{ik} \neq 0$. Hence $$\text{Pr}_{a_1,...,a_n \in_r S}[\langle a,u_i \rangle = 0] \leq \frac{1}{|S|}$$ by the Schwartz-Zippel lemma, and then 
\begin{equation}\label{eq:tainvertible}
    \text{Pr}_{a_1,...,a_n \in_r S}[T^{(a)} \text{ is invertible}] \geq 1 - \frac{n}{|S|}
\end{equation}
by the union bound.
Also, if $T^{(a)}$ is invertible, then $\lambda_i = \lambda_j$ if and only if $P_{ij}(a,b) = 0$. Written as a probabilistic statement, this gives us that
\begin{equation}\label{eq:equivalentchar}
\begin{split}
    &\text{Pr}_{a,b \in_r S}[T^{(a)} \text{ is invertible}\text{ and } \text{for all } i \neq j, \lambda_i \neq \lambda_j ]  \\
&= \text{Pr}_{a,b \in_r S}[T^{(a)} \text{ is invertible}\text{ and } \text{for all } i \neq j, P_{ij}(a,b) \neq 0 ].
\end{split}
\end{equation}
Since $U$ is an invertible matrix, its rows are linearly independent and there must exist some $k_0,l_0$ such that $(u_{ik_0}u_{jl_0} - u_{il_0}u_{jk_0}) \neq 0$. Hence, taking $a = e_{k_0}$ and $b = e_{l_0}$ (where $e_i$ denotes the vector with $1$ at the $i$-th position and $0$ otherwise), we get that $P_{ij}(e_{k_0},e_{l_0}) \neq 0$. Hence, $P_{ij} \not\equiv 0$ and 
\begin{align*}
    \text{Pr}_{a,b \in_r S}[P_{ij}(a,b) \neq 0] \geq 1- \frac{2}{|S|}
\end{align*} by Schwartz-Zippel since  $\text{deg}(P_{ij}) \leq 2$. 

Applying the union bound, we then have 
\begin{align*}
    \text{Pr}_{a,b \in_r S}[\text{For all }i\neq j \in [n], P_{ij}(a,b) \neq 0] \geq 1- \frac{2\binom{n}{2}}{|S|}.
\end{align*}
Finally, using (\ref{eq:tainvertible}) and (\ref{eq:equivalentchar}), we have that 
\begin{align*}
    &\text{Pr}_{a,b \in_r S}[T^{(a)} \text{ is invertible and } \lambda_i \neq \lambda_j \text{ for all } i \neq j ]  \\
    &= \text{Pr}_{a,b \in_r S}[T^{(a)} \text{ is invertible and }\text{for all } i \neq j, P_{ij}(a,b) \neq 0] \\
    &\geq 1- (\frac{2\binom{n}{2}}{|S|} + \frac{n}{|S|}).
\end{align*}
\end{proof}

\begin{theorem}\label{thm:complexteexactmain}
If an input tensor $T \in (\C^n)^{\otimes 3}$ can be written as $T = \sum_{i=1}^n u_i^{\otimes 3}$ where the $u_i \in \C^n$ are linearly independent vectors, then Algorithm~\ref{algo:completeexact} succeeds with high probability. More formally, if $a_1,...,a_n,b_1,...,b_n$ are chosen uniformly and independently at random from a finite subset $S \subset \C$, then the algorithm returns linearly independent $l_1,...,l_n \in \C^n$ such that $T = \sum_{i=1}^n l_i^{\otimes 3}$ with probability at least $1- (\frac{2\binom{n}{2}}{|S|} + \frac{n}{|S|})$. 
\end{theorem}
\begin{proof}
First, using Theorem \ref{thm:eigenvalues} we get that if $a_1,...,a_n,b_1,...,b_n$ are picked uniformly and independently at random from a finite subset $S \subset \mathbb{K}$, then $T^{(a)}$ is invertible and the eigenvalues of $T^{(a)'}T^{(b)}$ are distinct with probability at least $1- (\frac{2\binom{n}{2}}{|S|} + \frac{n}{|S|})$. 
\par
Now we show that if $T^{(a)}$ is invertible and the eigenvalues of $T^{(a)'}T^{(b)}$ are distinct, then Algorithm \ref{algo:completeexact} returns linearly independent vectors $l_1,...,l_n \in \C^n$ such that $T = \sum_{i=1}^n (l_i)^{\otimes 3}$.
\par
If the eigenvalues of a matrix are distinct, then the rank of the eigenspaces corresponding to each eigenvalue is $1$. Hence, the eigenvectors of $T^{(a)'}T^{(b)}$ are unique (up to a scaling factor). From (\ref{eq:formulata'tb}), we get that the columns of $U^{-1}$ are eigenvectors of $T^{(a)'}T^{(b)}$. Since the columns of $P$ are also eigenvectors of $T^{(a)'}T^{(b)}$, this gives us the relation 
\begin{equation}\label{eq:eigenvectormaineq}
    P = U^{-1}D \text{ where } D = \text{diag}(k_1,...,k_n) \text{ and } k_i \neq 0.
\end{equation}
Now we claim that the $\alpha_i$'s computed in Step 6 of Algorithm \ref{algo:completeexact} are exactly equal to $k_i^3$.
Let $S = (P \otimes P \otimes P).T$ and let $S_1,...,S_n$ be the slices of $S$. Since $T = \sum_{i=1}^n (U^Te_i)^{\otimes 3}$, 
\begin{equation}\label{eq:alphaiki3}
\begin{split}
    \alpha_i = Tr(S_i) &= \sum_{j=1}^n S_{i,j,j}  \\
    &= \sum_{j,t=1}^n ((D^TU^{-T}U^Te_t)^{\otimes 3})_{i,j,j}\\
    &= \sum_{j,t = 1}^n (D^Te_t)_i(D^Te_t)^2_j = (D^Te_i)^3_i = k_i^3.
\end{split}
\end{equation}
Since $U$ and $D$ are both invertible, $P$ is invertible as well and $DP^{-1} = U$. Putting this in vector notation, if $v_i$ are the rows of $P^{-1}$, then $u_i = k_iv_i$. As a result, for any cube root of unity $\omega_i$
\begin{align*}
    T = \sum_{i=1}^n (u_i)^{\otimes 3} = \sum_{i=1}^n (\omega_ik_iv_i)^{\otimes 3} = \sum_{i=1}^n ((\alpha_i)^{\frac{1}{3}}v_i)^{\otimes 3}.
\end{align*}
We say that Algorithm \ref{algo:completeexact} "succeeds" if the algorithm returns linearly independent $l_1,...,l_n \in \C^n$ such that $T = \sum_{i=1}^n l_i^{\otimes 3}$.
This finally gives us that 
\begin{align*}
    &\text{Pr}_{a,b \in_r S}[\text{Algorithm \ref{algo:completeexact} "succeeds"}] \\ 
    &\geq \text{Pr}_{a,b \in_r S}[\text{Algorithm \ref{algo:completeexact} "succeeds"},T^{(a)} \text{ is invertible and the eigenvalues of } T^{(a)'}T^{(b)} \text{ are distinct}] \\
    &= \text{Pr}_{a,b \in_r S}[T^{(a)} \text{ is invertible and the eigenvalues of } T^{(a)'}T^{(b)} \text{ are distinct}] \\
    &\geq 1- (\frac{2\binom{n}{2}}{|S|} + \frac{n}{|S|}).
\end{align*}
\end{proof}
Taking $S$ to a finite set such that $|S| > n^3$ and applying the previous analysis, we get the following corollary:
\begin{corollary}\label{thm:Jennrichexact}
Given a diagonalisable tensor $T \in C^{n \times n \times n}$, there is an algorithm which returns $u_1,...,u_n \in \C^n$ such that $T = \sum_{i=1}^n u_i^{\otimes 3}$ with probability $1- \frac{1}{n}$.  
\end{corollary}

{
\subsection*{An alternative to step 6: the partial trace}

After step 5 of Algorithm \ref{algo:completeexact} we have determined
vectors $v_i$ such that $T=\sum_{i=1}^n \alpha_i v_i^{\otimes 3}$ for some coefficients $\alpha_1,\ldots,\alpha_n$.
The goal of step 6 is to determine these coefficients,
and the proposed solution runs in $O(n^3)$ arithmetic operations thanks to the TSCB algorithm from Section~\ref{sec:cobfinitear}. An anonymous referee
has suggested to use the trace in a different way, and 
more precisely to use the partial trace. We sketch this approach below.

First, we recall that the partial trace is a linear map
which associates to the 3-tensor $T$ an $n$-dimensional
vector $Tr'(T)$ with entries 
$$Tr'(T)_i=\sum_{j=1}^n T_{ijj}.$$
In particular, the partial trace of a rank-one tensor satisfies 
$$Tr'(v^{\otimes 3})=Tr(v \otimes v)v=
(\sum_{j=1}^n v_j^2)v.$$
Suppose first that $T$ has real entries. From the relation 
$T=\sum_{i=1}^n \alpha_i v_i^{\otimes 3}$ we obtain
$Tr'(T)=\sum_{i=1}^n \alpha_i ||v_i||^2 v_i,$ 
and now we have a system of $n$ linear equations with $n$ unknowns. Note that 
$||v_i||^2 \neq 0$ since the $v_i$ are the rows of~$P^{-1}$.
We can therefore determine the unique solution in $O(n^3)$ operations.

This algorithm does not quite work when $T$ has complex
entries since the {\em Euclidean} norm of $v_i$
might vanish. In the complex case, the referee has suggested
to introduce random weights when computing the partial trace.
Namely, given complex coefficients 
$\alpha_1,\ldots,\alpha_n$, we define a new partial trace operation~$Tr'_{\alpha}$. This linear map associates
to $T$ the $n$-dimensional vector with entries
$$Tr'_{\alpha}(T)_i=\sum_{j=1}^n \alpha_j T_{ijj}.$$
Now we have
$Tr'_{\alpha}(v^{\otimes 3}) = (\sum_{j=1}^n \alpha_j v_j^2)v$.
For any nonzero vector $v$, 
the coefficient $\sum_{j=1}^n \alpha_j v_j^2$ will be nonzero with high probability
if the weights $\alpha_1,\ldots,\alpha_n$ are chosen at random. We can therefore find the unique  solution 
to the linear system in $O(n^3)$ operations like in the
real case.

}

\section{Complete decomposition of symmetric tensors in finite arithmetic}\label{sec:complete}

\par

{ We claimed at the beginning of Section~\ref{sec:results} that the condition number for symmetric tensor decomposition is well defined. 
In Section~\ref{sec:complete} we first justify that claim, then present our finite precision decomposition algorithm (Algorithm~\ref{algo:Jennrich}), and analyze its properties from Section~\ref{sec:algo} onward.}

First, we state a well-known result  showing that the tensor decomposition is unique up to permutation if it satisfies certain conditions. Here we will state the result following the notation of \cite{Moitra2018AlgorithmicAO}.
\begin{definition}
We say that two sets of factors 
\begin{align*}
    \{(u^{(i)},v^{(i)},w^{(i)})\}_{i=1}^r \text{ and } \{(\overline{u}^{(i)},\overline{v}^{(i)},\overline{w}^{(i)})\}_{i=1}^r
\end{align*}
are equivalent if there is a permutation $\pi: [r] \xrightarrow[]{} [r]$ such that for all $i$,
\begin{align*}
    u^{(i)} \otimes v^{(i)} \otimes w^{(i)} = \overline{u}^{(\pi(i))} \otimes \overline{v}^{(\pi(i))} \otimes \overline{w}^{(\pi(i))}.
\end{align*}
\end{definition}
\begin{theorem}\cite{Har70}
Suppose we are given a tensor of the form 
\begin{align*}
    T = \sum_{i \in [r]}  u^{(i)} \otimes v^{(i)} \otimes w^{(i)}
\end{align*}
where the following conditions are met:
\begin{itemize}
    \item the vectors $\{u^{(i)}\}_i$ are linearly independent.
    \item the vectors $\{v^{(i)}\}_i$ are linearly independent.
    \item every pair of vectors in $\{w^{(i)}\}_i$ is linearly independent.
\end{itemize}
Then $\{(u^{(i)},v^{(i)},w^{(i)})\}_{i=1}^r \text{ and } \{(\overline{u}^{(i)},\overline{v}^{(i)},\overline{w}^{(i)})\}_{i=1}^r$ are equivalent factors.
\end{theorem}
Applying this to the case of symmetric tensors, we get the following corollary.
\begin{corollary}\label{lem:uniqueness}
Let $T = \sum_{i \in [n]} u_i^{\otimes 3}$ be a symmetric tensor where the vectors  $u_i \in \mathbb{C}^n$ are linearly independent. For any other decomposition $T = \sum_{i \in [n]} (u'_i)^{\otimes 3}$, the vectors $u'_i$ must satisfy $u'_i = \omega_iu_{\pi(i)}$ where $\omega_i$ is a cube root of unity and $\pi \in S_n$ a permutation. 
\end{corollary}
{ The above result was also derived in~\cite{Kayal11} by a different method (uniqueness of polynomial factorization).}
For the next lemma, recall the definition of the condition number of a diagonalisable symmetric tensor from Definition \ref{def:conditionnumber}.
\begin{lemma}\label{lem:equalcondtncomplete}
Let $T$ be a diagonalisable tensor. Then for all $U \in M_n(\mathbb{C})$ such that $U$ diagonalises $T$, the condition numbers $\kappa_F(U)$ are equal.
\end{lemma}
\begin{proof}
By Corollary \ref{lem:uniqueness},  for all $U \in M_n(\mathbb{C})$ such that $U$ diagonalises $T$, the rows of $U$ are unique up to permutation and scaling by cube roots of unity. Writing this in matrix notation, if $U$ and $U'$ are two such distinct matrices that diagonalise $T$, there exists a permutation $\pi \in S_n$ and a diagonal matrix $D$ with cube roots of unity along the diagonal entries, such that $U' = DP_{\pi}U$ where $P_{\pi}$ is the permutation matrix corresponding to $\pi$. 
\par
Now, $||U'||_F = ||DP_{\pi}U||_F$. If $u_1,...,u_n$ are the rows of $U$, then the rows of $U'$, i.e. $u'_i = \omega_iu_{\pi(i)}$ where $\omega_i$ are the cube roots of unity. Using the definition of $||.||_F$, we get that $||U'||_F^2 = \sum_{i \in [n]} ||u'_i||^2 = \sum_{i \in [n]} ||\omega_iu_{\pi(i)}||^2 = \sum_{i \in [n]} ||u_i||^2 = ||U||_F^2$. Similarly, 
$$||(U')^{-1}||_F = ||(DP_{\pi}U)^{-1}||_F = ||U^{-1}(P_{\pi})^TD^{-1}||_F.$$ Since $(P_{\pi})^T$ is also a permutation matrix, multiplication by it on the right permutes the columns of $U^{-1}$. Also, inverse of cube roots of unity are cube roots of unity as well. Hence, if $v'_1,...,v'_n$ are the columns of $U^{-1}$, and $v_1,...,v_n$ are the columns of $U$, this gives us that $v'_i = \omega'_iv_{\pi^{-1}(i)}$ where $\omega_i'$ is a cube root of unity. This gives us that $||(U')^{-1}||^2_F = \sum_{i=1}^n ||v'_i||^2 = \sum_{i \in [n]} ||\omega'_iv_{\pi^{-1}(i)}||^2 = \sum_{i \in [n]} ||v_i||^2 = ||U^{-1}||_F^2$ . This finally gives us that $\kappa_F(U') = ||U'||_F^2 + ||(U')^{-1}||_F^2 = ||U||_F^2 + ||U^{-1}||_F^2 = \kappa_F(U)$.
\end{proof}
\subsection{Simultaneous Diagonalisation Algorithm for Symmetric Tensors in Finite Precision}\label{sec:properror}
\label{sec:algo}

\begin{algorithm} 
\label{algo:Jennrich}
\SetAlgoLined
\nonl Let $C_{\text{gap}},C_{\eta} > 0$ and $c_F > 1$ be some absolute constants we will fix in (\ref{eq:cgapcf}). \\ 
\nonl \textbf{Input:} An order-3 symmetric diagonalisable tensor $T \in (\mathbb{C}^n)^{\otimes 3}$, an estimate $B$ for the condition number of the tensor and  an accuracy parameter {$\varepsilon < 1$}.  \\
\nonl \textbf{Output:} A solution to the $\varepsilon$-approximation problem for the decomposition of~$T$. \\
\nonl Set $k_{\text{gap}} := \frac{1}{C_{\text{gap}}n^6B^3}$  and $k_{F} := c_Fn^5B^3$.  \\
\nonl Pick $(a_1,...,a_n,b_1,...,b_n) \in G_{\eta}$ uniformly at random where $\eta := \frac{1}{C_{\eta}n^{\frac{17}{2}}B^4}$ is the grid size. \\
\nonl Let $T_1,...,T_n$ be the slices of $T$.\\

Compute $S^{(a)} = \sum_{i=1}^n a_iT_i$ and $S^{(b)} = \sum_{i=1}^n b_iT_i$ on a floating point machine. \\
Compute $S^{(a)'} = INV(S^{(a)})$ on a floating point machine where $INV$ is the stable matrix inversion algorithm in Theorem \ref{thm:fastlinearalgebra}.\\  
\nonl Let $\delta := \frac{\varepsilon^3}{Cn^{12}B^{\frac{9}{2}}}$ where $C$ is a constant we will set in (\ref{eq:C}). \\
Compute $D = \text{MM}(S^{(a)'},S^{(b)})$ on a floating point machine where $MM$ is the stable matrix multiplication algorithm in Theorem \ref{thm:fastlinearalgebra}. \\
Let $v^{(0)}_1,...,v^{(0)}_n$ be the output of $EIG-FWD$ on the input $(D, \delta, \frac{3nB}{k_{\text{gap}}}, 2B^{\frac{3}{2}}\sqrt{nk_F})$ on a floating point machine. \\
\nonl Let $V^{(0)}$ be the matrix with $v^{(0)}_1,...,v^{(0)}_n$  as columns. \\
Compute $C = \text{INV}(V^{(0)})$ on a floating point machine where $\text{INV}$ is the matrix inversion algorithm in Theorem \ref{thm:fastlinearalgebra} and let $u_i'$ be the rows of $C$. \\
Let $\alpha'_1,...,\alpha'_n$ be the output of $TSCB(T,V^{(0)})$ where $TSCB$ is the algorithm for computing the trace of the slices after a change of basis in Algorithm \ref{algo:fastcob}.\\
Output $\{l_1,...,l_n\}$ where $l_i = (\alpha'_i)^{\frac{1}{3}} u'_i$ is computed on a floating point machine for all $i \in [n]$. Note that by $(\alpha'_i)^{\frac{1}{3}}$ we refer to any one of the cube roots of $\alpha'_i$. \\
\caption{Algorithm for Complete Decomposition of Symmetric Tensors.}
\end{algorithm}
Recall that the condition number for tensor diagonalisation $\kappa(T)$ was defined in Definition \ref{def:conditionnumber}, and the notion of $\varepsilon$-approximation  for tensor decomposition was defined in Section \ref{sec:fwddef}.
{ Our main result about Algorithm~\ref{algo:Jennrich} below already appears as Theorem~\ref{th:main} in the introduction, and it is the central result of this paper. }


Let $T_1,...,T_n$ be the slices of the tensor. In the algorithm, we pick $a,b$ uniformly and independently at random from  the finite grid $G_{\eta} = \{-1,-1+\eta,-1+2\eta,...,1-2\eta,1-\eta\}^{2n}$, then define $T^{(a)}~= \sum_{i=1}^n a_iT_i$ and $T^{(b)} = \sum_{i=1}^n b_iT_i$. 
{ In Section~\ref{sec:complete} we give 
a proof of the main theorem under some extra assumptions.
Namely, we will assume in Theorem \ref{thm:mainerrorfinal} below that we have picked points $a,b$ such that the Frobenius condition number of $T^{(a)}$ is "small" and the eigenvalue gap of $(T^{(a)})^{-1}T^{(b)}$ is "large".
We will see in Section~\ref{subsection:correctnessjennrich} 
that these assumptions are satisfied with high probability over the choice of $a$ and $b$.
This will allow us to complete the proof of the main theorem in   Section~\ref{sec:finishproofmain}.}
\begin{theorem}\label{thm:mainerrorfinal}
Let $T \in \C^n \otimes \C^n \otimes \C^n$ be a degree-$3$ diagonalisable symmetric tensor. Let $U \in \text{GL}_n(\C)$ be such that $U$ diagonalises $T$ where $\kappa(T) = \kappa_F(U) <~B$. Let $T^{(a)} = \sum_{i=1}^n a_iT_i$ and $T^{(b)}= \sum_{i=1}^n b_iT_i$ be two linear combination of the slices $T_1,...,T_n$ of $T$ such that $T^{(a)}$ is invertible, $\kappa_F(T^{(a)}) \leq k_F := C_Fn^5B^3$ and $\text{gap}((T^{(a)})^{-1}T^{(b)}) \geq k_{\text{gap}} := \frac{1}{C_{\text{gap}}n^6B^3}$. Let $\varepsilon \leq 1$ be the input accuracy parameter.  Then Algorithm \ref{algo:Jennrich} outputs an 
$\varepsilon$-approximation to the tensor decomposition problem for $T$ in 
$$O(n^3 + T_{MM}(n)\log^2 \frac{nB}{\varepsilon})$$ arithmetic operations on a floating point machine with 
$$O(\log^4(\frac{nB}{\varepsilon})\log n)$$ bits of precision, with probability at least $1- \frac{1}{n} - \frac{12}{n^2}$.

\end{theorem}
As we have seen in Section \ref{sec:completeexactarithmetic}, if each of the steps of Algorithm \ref{algo:Jennrich} are performed in exact arithmetic and if we can perform matrix diagonalisation exactly in Step 4, we will get vectors $v_1,...,v_n \in \C^n$ such that $T = \sum_{i=1}^n v_i^{\otimes 3}$ exactly. We will refer to this as the \textit{ideal output}. In the next section we show that under some suitable assumptions, Algorithm \ref{algo:Jennrich} indeed outputs linear forms which are at distance at most $\varepsilon$  from the ideal output. We will gradually build towards the proof of this theorem and will finish it in Section~\ref{sec:finishproof}. 
\par
\textbf{Organization: }For better readability of this exposition, we keep the error analysis for Steps 1,2 and 3 in the main text in order to provide a flavour of the calculations involved. The analysis for Steps 4 - 7 are somewhat technical and we defer them to Appendix \ref{app:erroranalysis}.
\subsection{Error accumulated in Steps 1,2,3}\label{sec:error123}
The slices $T_1,.., T_n$ can be computed without any additional error. Let $T^{(a)} = \sum_{i=1}^n a_iT_i$ and $T^{(b)} = \sum_{i=1}^n b_iT_i$. Also, let $S^{(a)}$ and $S^{(b)}$ be the matrices $\sum_{i=1}^n a_iT_i$ and $\sum_{i=1}^n b_iT_i$ respectively, computed on a floating point machine with precision $\textbf{u}$. Note that if infinite precision is allowed, ideally, the output at the end of Step 3 would have been $(T^{(a)})^{-1}T^{(b)}$. We show that the actual output at the end of Step 3 of the algorithm, i.e., $D := MM(S^{(a)'},S^{(b)})$ is very close to $(T^{(a)})^{-1}T^{(b)}$. More formally, we will show there exist constants $c', C_3$ such that if the algorithm is run on a floating point machine with number of bits of precision $\log(\frac{1}{\textbf{u}}) > c'\log(nB)\log(n)$, then 
\begin{equation}\label{eq:mubounds}
    ||D - (T^{(a)})^{-1}T^{(b)}|| := \varepsilon_3 \leq (nB)^{C_3 \log(n)}. \textbf{u}.
\end{equation}
Note that we will set these constants later in (\ref{eq:eps3}).
\subsubsection{Error for Step 1}
If the algorithm is run in exact arithmetic, the "ideal" outputs at the end of Step 1 are $T^{(a)}$ and $T^{(b)}$. Let us denote by $S^{(a)}$ and $S^{(b)}$ the outputs at the end of Step 1 of the algorithm. In the next lemma, we estimate the difference between the ideal outputs and the actual outputs.
Recall the Frobenius norm of a tensor, $||T||_F$, from  Definition~\ref{def:tensornorm}.
\begin{lemma} \label{lem:errstep1}
Let $T_1,...,T_n$ be the slices of  $T$ which are the inputs to Step 1. Let $T^{(a)} = \sum_{i=1}^n a_iT_i$ and $T^{(b)} = \sum_{i=1}^n b_iT_i$. Let $S^{(a)}$ and $S^{(b)}$ be the matrices $\sum_{i=1}^n a_iT_i$ and $\sum_{i=1}^n b_iT_i$ respectively computed on a floating point machine with machine precision $\textbf{u}$ where $n\textbf{u} < 1$. Then  
\begin{align*}
    ||S^{(a)} - T^{(a)} || \leq 2n \textbf{u} ||a|| ||T||_F \text{ and } ||S^{(b)} - T^{(b)} || \leq 2n \textbf{u} ||b|| ||T||_F
\end{align*}
where $a = (a_1,...,a_n)$ and $b = (b_1,...,b_n)$.
\end{lemma}
\begin{proof}
{We only need to prove the first  bound (on $||S^{(a)} - T^{(a)} || $) since the proof of the second one is the same.}
Since  $(T^{(a)})_{jk} = \sum_{i=1}^n a_i(T_i)_{j,k}$, it follows from (\ref{eq:ipbound}) that $$||(S^{(a)})_{j,k} - (T^{(a)})_{j,k}|| \leq \gamma_n ||a|| \sqrt{\sum_{i=1}^n |(T_i)_{j,k}|^2}.$$

Moreover, $\gamma_n \leq 2n\textbf{u}$ since $n\textbf{u} < 1$. Hence 
\begin{equation}\label{eq:errstep1a}
\begin{split}
    ||S^{(a)} - T^{(a)} || \leq ||S^{(a)} - T^{(a)} ||_F &= \sqrt{\sum_{j,k=1}^n ||(S^{(a)})_{j,k} - (T^{(a)})_{j,k}||^2} \\
    &\leq 2n \textbf{u} ||a|| \sqrt{\sum_{i,j,k=1}^n |(T_i)_{j,k}|^2} \\
    &= 2n \textbf{u} ||a||||T||_F.
\end{split}
\end{equation}
\end{proof}
\begin{lemma}\label{lem:normtensorbound}
Let $T$ be an order-$3$ diagonalisable tensor. Then $||T||_F \leq (\kappa(T))^{\frac{3}{2}}$.
\end{lemma}
\begin{proof}
Let $T =  \sum_{i=1}^n u_i^{\otimes 3}$ where the $u_i$'s are linearly independent. Let $U \in \text{GL}_n(\C)$ be the matrix with rows $u_1,...,u_n$. From Corollary \ref{corr:P3structural}, we get that the slices $T_i$ of $T$ can be written as $T_i = U^T D_i U$ where $D_i = \text{diag}(u_{1,i},...,u_{n,i})$. 
Therefore, 
\begin{align*}
    ||T||^2_F = \sum_{i=1}^n ||T_i||_F^2 &= \sum_{i=1}^n ||U^TD_iU||_F^2 \\
    &\leq ||U||_F^4 \sum_{i=1}^n ||D_i||_F^2 \\
    &= ||U||_F^4 (\sum_{i=1}^n (\sum_{k=1}^n |u_{k,i}|^2)) \\
    &= ||U||_F^6 \leq \kappa(T)^3.
\end{align*}
\end{proof}
Now we bound the error in terms of the input parameters. Recall that the input tensor has condition number $\kappa(T) \leq B$ and $a,b \in G_{\eta} \subset [-1,1]^n$. 
\begin{claim}
\begin{equation}\label{eq:sata}
    ||S^{(a)} - T^{(a)}|| \leq 2\textbf{u}(nB)^{\frac{3}{2}}  \text{ and } ||S^{(b)} - T^{(b)}|| \leq 2\textbf{u}(nB)^{\frac{3}{2}}.
\end{equation}
\end{claim}
\begin{proof}
As in Lemma \ref{lem:normtensorbound}, we only need to prove the first inequality. Since $||a|| \leq \sqrt{n}$,  it follows from (\ref{eq:errstep1a})  that 
\begin{equation}\label{eq:normflTa}
    ||S^{(a)} - T^{(a)}|| \leq 2\textbf{u}n^{\frac{3}{2}}||T||_F \leq 2\textbf{u}(nB)^{\frac{3}{2}}.
\end{equation}
\end{proof}
\subsubsection{Error for Step 2}
The error in Step 2 stems from two facts. Firstly, the input has some error due to Step 1 and secondly, in finite arithmetic the inverse of a matrix cannot be computed exactly. 
In this section, we give a bound on the error from Step~2 taking into account these two sources of error. This will be the template for the analysis of the other steps of the algorithm. The following is the main result of this subsection.
\begin{theorem}\label{thm:epsilon2}
Let $k_F = c_Fn^5B^3$ as defined in Algorithm \ref{algo:Jennrich} and assume that $\kappa_F(T^{(a)}) < k_F$. Assume also that the algorithm is run on a floating point machine with the number of bits of precision $\log(\frac{1}{\textbf{u}}) > c\log^4{nB}$ where $c = 2\max\{\log(4\sqrt{c_F}),4\}$. Then the output at the end of Step 2 has error
\begin{equation}
    \varepsilon_2 \leq (nB)^{C_2\log n}.\textbf{u}
\end{equation}
for some constant $C_2$.
\end{theorem}
In this direction,  we first prove the following theorem: if $A$ is a matrix with bounded $\kappa_F$ and $A'$ is another matrix which is close to $A$, then $(A')^{-1}$ is also close to $A^{-1}$.
\begin{lemma}\label{lem:nearconditionnumber}
Let $A \in M_n(\mathbb{C})$ be such that $\kappa_F(A) \leq K < \infty$ . Define $A' \in M_n(\mathbb{C})$ as $A' = A + \Delta$ where $||\Delta|| \leq M$ and $M\sqrt{K} < 1$. Then $A'$ is invertible and
\begin{align*}
    ||(A')^{-1} - A^{-1}|| \leq \frac{MK}{1 - M\sqrt{K}}
\end{align*}
\end{lemma}
\begin{proof}
We first use the fact that for any matrix $B \in M_n(\C)$, if $||B|| < 1$, $I+B$ is invertible and 
\begin{equation}\label{eq:i+b}
    (I+ B)^{-1} = \sum_{i=0}^{\infty} (-1)^i B^i.
\end{equation}
Since $A'  = A(I + A^{-1}\Delta)$ where $||\Delta|| \leq M$ and $||A^{-1}|| \leq ||A^{-1}||_F \leq \sqrt{\kappa_F(A)} \leq \sqrt{K}$, we have that $||A^{-1}\Delta|| \leq ||A^{-1}||||\Delta|| \leq M\sqrt{K} < 1$. This shows that $A'$ is invertible, hence $(A')^{-1} = (I+ A^{-1}\Delta)^{-1}A^{-1}$. Now, we can use (\ref{eq:i+b}) for $B = A^{-1}\Delta$ and apply the  triangle inequality to get that
\begin{align*}
    ||(A')^{-1} - A^{-1}|| &= ||(I+ A^{-1}\Delta)^{-1}A^{-1} - A^{-1}|| \\
    &\leq ||A^{-1}||||(I+ A^{-1}\Delta)^{-1} - I|| \\
    &\leq ||A^{-1}|| ||\Big(\sum_{i=1}^{\infty} ||A^{-1}\Delta||^i  \Big).
\end{align*}
Hence, we can finally conclude that
\begin{align*}
    ||(A')^{-1} - A^{-1}|| \leq  \frac{MK}{1 - M\sqrt{K}}.
\end{align*}
\end{proof}
If the algorithm was run in exact arithmetic, ideally the output at the end of step 2 would be $(T^{(a)})^{-1}$. In the next lemma we  show that the output $S^{(a)'}$ at the end of step 2 is close to $(T^{(a)})^{-1}$.
\begin{lemma}\label{lem:errstep2}
Let $S^{(a)}$ be the input to Step 2 of Algorithm \ref{algo:Jennrich} run on a floating point machine with machine precision $\textbf{u}$. Let $||S^{(a)} - T^{(a)}|| \leq \varepsilon_1$ be the error from Step 1. We also assume that $\kappa_F(T^{(a)}) \leq~k_F$. Let $S^{(a)'}$ be the output of Step 2. Then
\begin{equation}\label{eq:satainv}
    ||S^{(a)'} - T^{(a)})^{-1}|| \leq \varepsilon_{21} + \gamma' =: \varepsilon_2
\end{equation}
where $\varepsilon_{21} := \frac{\varepsilon_1 k_F}{1 - \varepsilon_1\sqrt{k_F}}$ and $\gamma' := n^{c_{\eta} + \log 10} \cdot \textbf{u} \cdot (\kappa_F(S^{(a)}))^{c_{\text{INV}}\log n + 1}$ is the error for fast and stable matrix inversion as in Definition \ref{def:INValg} and Theorem~\ref{thm:fastlinearalgebra}.
\end{lemma}
\begin{proof}
First, we show that $(S^{(a)})^{-1}$ is close to $(T^{(a)})^{-1}$ as well. Applying Lemma \ref{lem:nearconditionnumber} to $T^{(a)}$, we see that
\begin{align*}
    ||(S^{(a)})^{-1} - (T^{(a)})^{-1}|| \leq \frac{\varepsilon_1 k_F}{1 - \varepsilon_1\sqrt{k_F}}.
\end{align*}
Since $S^{(a)'} = INV(S^{(a)})$, by Theorem \ref{thm:fastlinearalgebra} we also have 
\begin{equation}\label{eq:sa'sainv}
    ||S^{(a)'} - (S^{(a)})^{-1}|| \leq~\gamma'.
\end{equation}
Combining these two equation with the triangle inequality, we get the desired result.
\end{proof}
The next part of the analysis is aimed at giving a bound for $\varepsilon_2$ in terms of the input parameters.
First, applying Theorem \ref{lem:errstep2} to $T^{(a)}$ along with the bounds of $\varepsilon_1$ from (\ref{eq:normflTa}), we have 
\begin{equation}\label{eq:normflTainv}
    \varepsilon_{21} \leq  \frac{2\textbf{u}(nB)^{\frac{3}{2}}k_F}{1 - 2\textbf{u}(nB)^{\frac{3}{2}}\sqrt{k_F}} .
\end{equation}
Now we try to give a bound on $\gamma'$. For this, we first  show that $\kappa_F(S^{(a)})$ is bounded by some constant time $\kappa_F(T^{(a)})$.
\begin{claim}\label{claim:kappasa}
Let $k_F := C_Fn^5B^3$ as defined in Algorithm \ref{algo:Jennrich}. Let the bits of precision of the floating point machine be $\log(\frac{1}{\textbf{u}}) > c\log(nB)$ where $c = 2\max\{\log(4\sqrt{c_F}),4\}$. Then $\kappa_F(S^{(a)}) \leq 8k_F$.
\end{claim}
\begin{proof}
Using the fact that $k_F = C_Fn^5B^3$, we will show that $4n^{\frac{3}{2}}B^{\frac{3}{2}}\sqrt{k_F} = 4\sqrt{c_F}n^4B^3 \leq \frac{1}{u}.$ Since $\log(\frac{1}{\textbf{u}}) > c\log(nB) > \log(4\sqrt{c_F}) + 4 \log(n) + 3\log(B)$, we have
\begin{equation}\label{eq:firstone}
    4n^{\frac{3}{2}}B^{\frac{3}{2}}\sqrt{k_F} = 4\sqrt{c_F}n^4B^3 \leq \frac{1}{u}.
\end{equation}
Applying that to (\ref{eq:normflTainv}) shows that 
\begin{equation}\label{eq:satainvred}
     \varepsilon_{21} \leq  4u(nB)^{\frac{3}{2}}k_F.
\end{equation}
By (\ref{eq:normflTa}) and (\ref{eq:firstone}) along with the hypothesis that $\kappa_F(T^{(a)})) \leq k_F$,
\begin{align*}
    \kappa_F(S^{(a)}) &= ||S^{(a)}||_F^2 + || (S^{(a)})^{-1} ||_F^2 \\
    &\leq ( 2\textbf{u}(nB)^{\frac{3}{2}} + \sqrt{k_F} )^2 + (4u(nB)^{\frac{3}{2}}k_F + \sqrt{k_F})^2.
\end{align*}
Now,  using again (\ref{eq:firstone}) we  have 
$\textbf{u}(nB)^{\frac{3}{2}}\sqrt{k_F} \leq \frac{1}{4}$. Since $c_F > 1$ as mentioned at the beginning of the algorithm, this implies $k_F > 1$. But  $\textbf{u} < \frac{1}{(nB)^c} < \frac{1}{n^4B^4}$, 
so 
$2\textbf{u}(nB)^{\frac{3}{2}} < 1$. Thus, we can finally conclude that $\kappa_F(S^{(a)}) \leq 8k_F$. 
\end{proof}
The above claim implies that
\begin{equation}
    \gamma' \leq n^{c_{\eta} + \log 10} \cdot \textbf{u} \cdot (8k_F)^{c_{\text{INV}}\log n + 1}.
\end{equation}
Using the fact that $k_F := c_Fn^5B^3$ we can then conclude that
\begin{equation}
\begin{split}\label{eq:gamma'}
        \gamma' &\leq 8c_F\Big(n^{c_{\eta} + \log(10) + 8\log(8c_F) +5(c_{\text{INV}}\log(n) + 1)}B^{3(c_{\text{INV}}\log(n) + 1)}\Big).\textbf{u} \\
        &\leq (nB)^{C\log n}.\textbf{u}.
\end{split}
\end{equation}
where $C$ is some suitable constant.

\begin{proof}[Proof of Theorem \ref{thm:epsilon2}]
From (\ref{eq:gamma'}), we know that there exists a constant $C$ such that $\gamma' \leq (nB)^{C\log n}.\textbf{u}$. Similarly from (\ref{eq:satainvred}), we get that $\varepsilon_{21} \leq 4u(nB)^{\frac{3}{2}}k_F = 4\sqrt{c_F}n^4B^3$. Combining this, we can conclude that $\varepsilon_2 = \varepsilon_{21} + \gamma' \leq (nB)^{C_2\log n}.\textbf{u}$ for some suitable chosen constant $C_2$.
\end{proof}
\subsubsection{Error for Step 3}
\begin{lemma}
Let $S^{(a)'}$ and $S^{(b)}$ be the input to Step 3 of Algorithm \ref{algo:Jennrich} run on a floating point machine with machine precision $\textbf{u}$. Let $||S^{(a)'} - (T^{(a)})^{-1}|| \leq \varepsilon_2$ be the error from Step 2 and $||S^{(b)} - T^{(b)}|| \leq \varepsilon_1$ be the error from Step 1. We assume that $\kappa_F(T^{(a)}) \leq k_F$ and $||T^{(b)}|| \leq k_b$. Let $D := MM(S^{(a)'},S^{(b)})$ be the output of Step 3 where $MM$ is the fast and stable matrix multiplication algorithm as in Definition \ref{def:MULTalg} and Theorem \ref{thm:fastlinearalgebra}. Then
\begin{align*}
    ||D - (T^{(a)})^{-1}T^{(b)}|| \leq (\varepsilon_2 + \sqrt{k_F})\varepsilon_1 + k_{b}\varepsilon_2 + \gamma_3 =: \varepsilon_3
\end{align*}
where $\gamma_3 := n^{c_{\eta}}\cdot \textbf{u} \cdot ||S^{(a)'}|| ||S^{(b)}||$ is the error for $MM$ on inputs $S^{(a)'}$ and $S^{(b)}$.
\end{lemma}
\begin{proof}
We use the second inequality of Lemma~\ref{lem:errstep1}. Multiplying both sides 
by $||S^{(a)'}||$ and using the fact that $||(T^{(a)})^{-1}|| \leq \sqrt{k_F}$, we have 
\begin{align*}
    ||S^{(a)'}S^{(b)} - S^{(a)'}T^{(b)} || \leq ||S^{(a)'}||\varepsilon_1 \leq (\varepsilon_2 + \sqrt{k_F})\varepsilon_1.
\end{align*}
Multiplying both sides of (\ref{eq:satainv}) by $||T^{(b)}||$ gives us that
\begin{align*}
    ||S^{(a)'}T^{(b)} - (T^{(a)})^{-1}T^{(b)} || \leq ||T^{(b)}||\varepsilon_2 \leq k_b \varepsilon_2.
\end{align*}
Since $D := \text{MM}(S^{(a)'},S^{(b)})$, using Theorem \ref{thm:fastlinearalgebra}, we already have  $||D - S^{(a)'}S^{(b)}|| \leq~\gamma_3$.
Combining all of these together using the triangle inequality, we conclude that
\begin{align*}
        ||D - (T^{(a)})^{-1}T^{(b)}|| \leq  k_{b}\varepsilon_2 + (\varepsilon_2 + \sqrt{k_F})\varepsilon_1 + \gamma_3 =: \varepsilon_3.
\end{align*}
\end{proof}
Now we want to bound the error $\varepsilon_3$ from Steps 1-3 of the Algorithm. 
\newline
\underline{\textbf{Bounding } $\gamma_1 = k_b \varepsilon_2$:}
\par
First, we show that $k_b \leq \sqrt{nB^3}$. Since $T$ is a diagonalisable tensor, there exist linearly independent vectors $u_i \in \C^n$ such that $T = \sum_{i=1}^n u_i^{\otimes 3}$. Let $U \in \text{GL}_n(\C)$ be the matrix with rows $u_1,...,u_n$. Using Corollary \ref{corr:P3structural}, we get that $T^{(b)} = \sum_{i=1}^n b_iT_i = U^TD^{(b)}U$ where $D^{(b)} = \text{diag}(\langle b,u_1\rangle,...,\langle b,u_n\rangle)$. Then $||T^{(b)}|| \leq ||U||^2 ||D^{(b)}||$. Now $$||D^{(b)}|| \leq ||D^{(b)}||_F \leq \sqrt{\sum_{k \in [n]}|\langle b,u_i \rangle|^2}.$$ A similar proof applies to $||T^{(a)}||$ as well. The following lemma helps us give a bound for $||D^{(a)}||$ and $||D^{(b)}||$.
\begin{lemma}\label{lem:normnum}
Let $U = (u_{ij}) \in M_n(\mathbb{\C})$ be such that $\kappa_F(U) \leq B$. Then, given $\textbf{a} \in [-1,1]^n$, $\sum_{k \in [n]}|\langle a,u_k \rangle|^2 \leq nB$.
\end{lemma}
\begin{proof}
By the Cauchy-Schwarz inequality, 
\begin{align*}
     \sum_{k \in [n]}|\langle a,u_k \rangle|^2 &\leq \sum_{k \in [n]} ||a||^2||u_k||^2 
     = ||a||^2||U||_F^2.
\end{align*}
Since $\textbf{a} \in [-1,1]^n $, we know that $||a||^2 \leq n$. Hence $ \sum_{k \in [n]}|\langle a,u_k \rangle|^2 \leq nB$.
\end{proof}
This finally gives us that 
\begin{equation}\label{eq:normtatb}
    ||T^{(a)}||,||T^{(b)}||  \leq \sqrt{nB^3}.
\end{equation}
Using Theorem \ref{thm:epsilon2} and (\ref{eq:normtatb}) and setting the number of bits of precision to be 
\begin{equation}\label{eq:u1}
    \log(\frac{1}{\textbf{u}}) > c\log(nB) \text{ where } c > 2\max\{\log(4\sqrt{c_F}),4\}
\end{equation}
we get that
\begin{equation}\label{eq:gamma1}
    \gamma_1 := k_b\varepsilon_2 \leq \sqrt{nB^3}(nB)^{C_2 \log n}. \textbf{u} < (nB)^{C_{31} \log n}. \textbf{u}.
\end{equation}
for some appropriate constant $C_{31}$.
\par
\underline{\textbf{Bounding } $\gamma_2 = (\varepsilon_2 + \sqrt{k_F})\varepsilon_1 $}:
\par
First we claim that $\varepsilon_2 \leq \sqrt{k_F}$.
Here we assume that the number of bits of precision to be 
\begin{equation}\label{eq:u2}
    \log(\frac{1}{\textbf{u}}) > c'\log(nB)\log(n) \text{ where } c > \max\{c,2C_2\}
\end{equation}
Note that $c,C_2$ are the constants from Theorem \ref{thm:epsilon2}. 

By (\ref{eq:u2}), we have $\log(\frac{1}{\textbf{u}}) > C_2\log(nB)\log(n)$. This gives us that $\varepsilon_2 \leq (nB)^{C_2\log(n)}.\textbf{u} < 1$. But since $c_F > 1$ as mentioned in the beginning of Algorithm \ref{algo:Jennrich}, $\sqrt{k_F} > 1$. Hence, 
\begin{equation}\label{eq:epsilon2kf}
    \varepsilon_2 \leq \sqrt{k_F}.
\end{equation}
By (\ref{eq:sata}), 
we have $\varepsilon_1 \leq 2\textbf{u}(nB)^{\frac{3}{2}}$. This shows that
\begin{equation}\label{eq:gamma2}
    \gamma_2 = (\varepsilon_2 + \sqrt{k_F})\varepsilon_1 < 2\sqrt{k_F}\varepsilon_1 < 4\sqrt{k_F}(nB)^{\frac{3}{2}}.\textbf{u} < (nB)^{C_{32}\log(n)}. \textbf{u}
\end{equation}
where $C_{32}$ is some appropriate constant. 
\par
\underline{\textbf{Bounding } $\gamma_3$:}
From Definition \ref{def:MULTalg} and Theorem \ref{thm:fastlinearalgebra}, we get that
\begin{align*}
       ||MM(S^{(a)'},S^{(b)}) - S^{(a)'}S^{(b)} || \leq \gamma_3 := n^{c_{\eta}}\cdot \textbf{u} \cdot ||S^{(a)'}|| ||S^{(b)}||.
\end{align*}
Since $||(T^{(a)})^{-1}|| \leq \sqrt{\kappa_F(T^{(a)})} \leq  \sqrt{k_F}$, by Lemma \ref{lem:errstep2} and (\ref{eq:epsilon2kf}), we have $||S^{(a)'}|| \leq \varepsilon_2 + \sqrt{k_F} < 2\sqrt{k_F}$. 
\par
Using (\ref{eq:normtatb}), we already have that $||T^{(b)}|| \leq \sqrt{nB^3}$.
Now using (\ref{eq:sata}) and the fact that $n\textbf{u}< \frac{1}{2}$, we get that $||S^{(b)}|| \leq 2\textbf{u}(nB)^{\frac{3}{2}} + \sqrt{nB^3} < 2\sqrt{nB^3}$. This shows that 
\begin{equation}\label{eq:gamma3}
    \gamma_3 \leq n^{c_{\eta}}\cdot \textbf{u} \cdot 4\sqrt{k_FnB^3} \leq (nB)^{C_{33}\log(n)}. \textbf{u}
\end{equation}
where $C_{33}$ is some appropriate constant.
Finally, we combine (\ref{eq:gamma1}), (\ref{eq:gamma2}) and (\ref{eq:gamma3}) for the error bounds along with (\ref{eq:u1}) and (\ref{eq:u2}) for the bounds on the number of bits of precision. This gives us that if the number of bits of precision satisifes
\begin{equation}\label{eq:steps1-3u}
    \log(\frac{1}{\textbf{u}}) > c'\log(nB)\log(n),
\end{equation}
then the error at the end of the third step of the algorithm is upper bounded as follows:
\begin{equation}\label{eq:eps3}
    ||MM(S^{(a)'},S^{(b)}) - (T^{(a)})^{-1}T^{(b)}|| =: \varepsilon_3 \leq \gamma_1 + \gamma_2 + \gamma_3 < (nB)^{C_3\log(n)}.\textbf{u},
\end{equation}
where $C_3 = 3\max\{C_{31},C_{32},C_{33}\}$.

\subsection{Proof of Theorem \ref{thm:mainerrorfinal}}\label{sec:finishproof}
\textbf{Running time of the algorithm:} We first analyse each step of the algorithm and deduce its complexity:
\begin{itemize}
    \item Computing $S^{(a)}$ and $S^{(b)}$ in Step 1 requires at most $2n^3$ many arithmetic operations.
    \item From Theorem \ref{thm:fastlinearalgebra}, since $S^{(a)} \in \C^{n \times n}$, applying $INV$ 
    to compute $S^{(a)'}$ in Step 2 requires $O(T_{MM}(n))$ many arithmetic operations.
    \item From Theorem \ref{thm:fastlinearalgebra}, since $S^{(a)'}, S^{(b)} \in \C^{n \times n}$, computing $D = MM(S^{(a)'}, S^{(b)})$ in Step 3 requires $T_{MM}(n)$ arithmetic operations.
    \item In Step 4, we compute $EIG-FWD$ on $(D, \delta, \frac{3nB}{k_{\text{gap}}}, 2B^{\frac{3}{2}}\sqrt{nk_F})$ where $\delta := \frac{\varepsilon^3}{Cn^{\frac{33}{2}}B^{\frac{13}{2}}}$, $k_{\text{gap}} := \frac{1}{C_{\text{gap}}n^6B^3}$  and $k_{F} := c_Fn^5B^3$. 
    The constants~$C$, $C_{\text{gap}}$, $c_F$ will be set in (\ref{eq:C}) and (\ref{eq:cgapcf}) respectively.  By Theorem \ref{thm:errstep4}, 
    Step 4 requires $O(T_{MM}(n)\log^2((\frac{nB}{\varepsilon})^{c_{\text{comp}}})) = O(T_{MM}(n)\log^2(\frac{nB}{\varepsilon}))$ many arithmetic operations where $c_{\text{comp}}$ is some appropriate constant. 
    \item Since $V^{(0)} \in M_n(\C)$, by Theorem \ref{thm:fastlinearalgebra}  computing $C = INV(V^{(0)})$ in Step 5 requires $O(T_{MM}(n))$ many arithmetic operations.
    \item By Theorem \ref{thm:fastcob},  computing $\alpha'_1,...,\alpha'_n$ requires $O(n^3)$ many arithmetic operations. 
    \item Finally, computing $l_1,...,l_n$ in Step 7 requires $O(n^2)$ many arithmetic operations.
\end{itemize}
So the total number of arithmetic operations required for this algorithm is $O(n^3 + T_{MM}(n)\log^2(\frac{nB}{\varepsilon}))$.
\par
\textbf{Precision:} { From (\ref{eq:steps1-3u}), the number of bits of precision required by Steps~1-3 of the Algorithm is  $\log(\frac{1}{\textbf{u}}) > c'\log(nB)\log(n)$ for some constant~$c'$. From Theorem \ref{thm:alg2boundfinalerror}, we also have that the number of bits of precision required by Steps 4-7 is $\log(\frac{1}{\textbf{u}}) = O(\log^4(\frac{nB}{\varepsilon})\log(n))$. Combining these we get the required bound on the number of bits of precision for the algorithm.}

\textbf{Correctness Analysis:} The error analysis in Section \ref{sec:error123} gives us that at the end of Step 3, on input $(T,a,b)$, the algorithm outputs a matrix $D$ which is close to $ (T^{(a)})^{-1}T^{(b)}$. More formally, by (\ref{eq:mubounds}), Step 3 of the Algorithm outputs a matrix $D$ such that $||D - (T^{(a)})^{-1}T^{(b)}|| \leq (nB)^{C_3 \log n}\textbf{u}$ for some constant $C_3$. 
\par
We now use the error analysis for Steps 4-7 that we have postponed to Appendix \ref{app:erroranalysis}. Applying Theorem \ref{thm:alg2boundfinalerror} to $D$, we conclude that if $D$ is given as input to Step~4, then the algorithm indeed returns vectors $l_1,...,l_n$ such that $||\omega_i u_i - l_i|| < \varepsilon$ up to some permutation. 

{
\begin{remark}
In the above complexity analysis we have neglected the cost of sampling 
a point uniformly at random from the grid $G_{\eta}$. Assuming that
we have access only to a source of random bits, this cost can be analyzed as follows. 

Given the value chosen for $\eta$ in the algorithm, we need to sample $2n$ coordinates from a set of size $N=O(n^{17/2}B^4).$
Moreover, we can round $N$ to the next higher power of 2 since 
working with a finer grid does not decrease the probability of success
of our algorithm. Therefore we can just sample independently $\log N$
random bits to obtain one coordinate. The overall cost for the $2n$ coordinates is therefore $O(n \log(nB))$. This is indeed negligible compared
to the complexity bound in Theorem~\ref{thm:mainerrorfinal}.
\end{remark}}
\section{Probability Analysis of Condition Numbers and Gap}\label{subsection:correctnessjennrich}

{ The central theme of this section is to  deduce anti-concentration inequalities about certain  families of polynomials arising in the analysis of Algorithm~\ref{algo:Jennrich}. 
Compared to~\cite{BCMV13},} an interesting novelty of these inequalities is that the underlying distribution for the random variables is discrete and that they are applicable to polynomials from $\R^n$ to $\C$. In Section \ref{sec:normbounds}, we first study some polynomial norms and then prove these results.
\par
Let $T \in (\C^n)^{\otimes 3}$ be a 
diagonalisable tensor given as input to Algorithm~\ref{algo:Jennrich} with $\kappa(T) < B$, and let $T_1,...,T_n$ be the slices of $T$. In the algorithm, we pick $a_1,...,a_n,b_1,...,b_n$ uniformly and independently at random from a finite discrete grid $G_{\eta} \subset [-1,1]^{2n}$ and define $T^{(a)}~= \sum_{i=1}^n a_iT_i$, $T^{(b)} = \sum_{i=1}^n b_iT_i$. In this section we show that  $T^{(a)}$ is invertible, $\text{gap}((T^{(a)})^{-1}T^{(b)}) \geq k_{\text{gap}} := \frac{1}{c_{\text{gap}}n^6B^3} $ and $\kappa_F(T^{(a)}) \leq k_F := c_Fn^5B^3$ with high probability. This is the main result of Section \ref{sec:probanalysisproof}. We also justify the choice of $k_{\text{gap}}$ and $k_F$ and choose appropriate values for $c_{\text{gap}}$ and $c_F$ in Section \ref{sec:finishproofmain}.
As a consequence of this, a central theorem arising out of this section is Theorem \ref{th:main} which concludes the probability analysis of Algorithm \ref{algo:Jennrich}. We state it here again for completeness. 
\begin{theorem}
Given a diagonalisable tensor $T$, a desired accuracy parameter $\varepsilon$ and some estimate $B \geq \kappa(T)$, Algorithm \ref{algo:Jennrich} outputs an $\varepsilon$-approximate solution to the tensor decomposition problem for $T$ in 
$$O(T_{MM}(n)\log^2 \frac{nB}{\varepsilon})$$ arithmetic operations on a floating point machine with 
$$O(\log^4(\frac{nB}{\varepsilon})\log n)$$ bits of precision, with probability at least $(1-\frac{1}{n} - \frac{12}{n^2})(1- \frac{1}{\sqrt{2n}} - \frac{1}{n})$.
\end{theorem}

\subsection{Some definitions and bounds on norms of polynomials}\label{sec:normbounds}

We define the norm of a polynomial following 
Forbes and Shpilka~\cite{forbes2017pspace}.
Recall from Section~\ref{sec:overview} that their goal was to construct so-called "robust hitting sets".
\begin{definition}(Norm of a complex-valued polynomial)
For an $n$-variate polynomial $f(x) \in \mathbb{C}[\textbf{x}]$, we denote 
\begin{align*}
    ||f||_2 := (\int_{[-1,1]^n} |f(\textbf{x})|^2d \mu(x))^{\frac{1}{2}}
\end{align*}
where $\mu(x)$ is the uniform probability measure on $[-1,1]^n$. We also denote 
\begin{align*}
    ||f||_{\infty} = \max_{v \in [-1,1]^n} |f(v)|.
\end{align*}
\end{definition}
\begin{lemma}\label{lem:normlowerbound1}
Let $U = (u_{ij}) \in \text{GL}_n(\C)$ such that $\kappa_{F}(U) \leq B$. Then, for all $k \in [n]$, $\sum_{i \in [n]} |u_{ik}|^2 \geq \frac{1}{B}$.
\end{lemma}
\begin{proof}
Since $\kappa_{F}(U) \leq B$, we already have  $||U^{-1}||_{F} \leq \sqrt{B}$. Also, we know that $||U^{-1}|| \leq ||U^{-1}||_{F}$. Hence, $||U^{-1}|| \leq \sqrt{B}$. By definition of the matrix norm,  $||U^{-1}x|| \leq \sqrt{B}||x||$ for all $x \in \C^n$. We define $x = Uy$ and this shows that $||Uy|| \geq \varepsilon||y||$ where $\varepsilon = \frac{1}{\sqrt{B}}$. Let $u_k$ be the $k$-th column of $U$ . Then $||u_k||^2 \geq \varepsilon^2$. Hence, $\sum_{i \in [n]} |u_{ik}|^2 \geq \varepsilon^2 = \frac{1}{B}$.
\end{proof}
The following theorem states that if the $l_2$ norm of a polynomial is not too small, then on inputs picked uniformly and independently at random from $[-1,1)^n$, the value of the polynomial is not too small with high probability. We use the following presentation of the theorem from \cite{forbes2017pspace}.  
\begin{theorem}[Carbery-Wright]\label{thm:CW}
There exists an absolute constant $C_{CW}$ such that if $f: \mathbb{R}^n \xrightarrow[]{} \mathbb{R}$ is a polynomial of degree at most $d$, then for $\alpha > 0$, it holds that
\begin{align*}
    \text{Pr}_{\textbf{v} \in_U [-1,1)^n}[|f(\textbf{v})| \geq \alpha] \geq 1- C_{CW}d\Big(\frac{\alpha}{||f||_2}\Big)^{\frac{1}{d}}.
\end{align*}
\end{theorem}
\begin{theorem}[Carbery-Wright for complex-valued polynomials]\label{thm:CWcomplex}
There exists an absolute constant $C_{CW}$ such that if $f: \mathbb{R}^n \xrightarrow[]{} \mathbb{C}$ is a polynomial of degree at most $d$, then for $\alpha > 0$, it holds that
\begin{align*}
    \text{Pr}_{\textbf{v} \in_U [-1,1)^n}[|f(\textbf{v})| \geq \alpha] \geq 1- 2C_{CW}d\Big(\frac{\alpha}{||f||_2}\Big)^{\frac{1}{d}}.
\end{align*}
\end{theorem}
\begin{proof}
Since $f: \R^n \xrightarrow{} \C$, we can write $f = \mathfrak{R}(f) + \iota \mathfrak{I}(f)$ where $\mathfrak{R}(f),\mathfrak{I}(f) \in \R^n \xrightarrow{} \R$ are real polynomials of degree $\leq d$. Then using Theorem \ref{thm:CW}, we get that 
\begin{align*}
     \text{Pr}_{\textbf{v} \in_U [-1,1)^n}[|\mathfrak{R}(f)(\textbf{v})| \geq \alpha] &\geq 1- C_{CW}d\Big(\frac{\alpha}{||\mathfrak{R}(f)||_2}\Big)^{\frac{1}{d}} \\
      \text{Pr}_{\textbf{v} \in_U [-1,1)^n}[|\mathfrak{I}(f)(\textbf{v})| \geq \alpha] &\geq 1- C_{CW}d\Big(\frac{\alpha}{||\mathfrak{I}(f)||_2}\Big)^{\frac{1}{d}}.
\end{align*}
This gives us that
\begin{align*}
    &\text{Pr}_{\textbf{v} \in_U [-1,1)^n}[|f(v)|^2 \leq \alpha^2||f||_2^2] \\
    &= \text{Pr}_{\textbf{v} \in_U [-1,1)^n}[|\mathfrak{R}(f)(\textbf{v})|^2 + |\mathfrak{I}(f)(\textbf{v})|^2 \leq \alpha^2\Big(||\mathfrak{R}(f)||^2_2 + |\mathfrak{I}(f)||^2_2\Big)] \\
    &\leq \text{Pr}_{\textbf{v} \in_U [-1,1)^n}[|\mathfrak{R}(f)(\textbf{v})|^2 \leq \alpha^2||\mathfrak{R}(f)||^2_2 \bigcup |\mathfrak{I}(f)(\textbf{v})|^2 \leq \alpha^2||\mathfrak{I}(f)||^2_2] \\
    &\leq \text{Pr}_{\textbf{v} \in_U [-1,1)^n}[|\mathfrak{R}(f)(\textbf{v})|^2 \leq \alpha^2|  |\mathfrak{R}(f)||^2_2] + \text{Pr}_{\textbf{v} \in_U [-1,1)^n}[|\mathfrak{I}(f)(\textbf{v})|^2 \leq \alpha^2||\mathfrak{I}(f)||^2_2] \\
    &\leq 2C_{CW}d(\alpha)^{\frac{1}{d}}.
\end{align*}
As a result,
\begin{align*}
    \text{Pr}_{\textbf{v} \in_U [-1,1)^n}[|f(v)|^2 \geq \alpha^2] \geq 1- 2C_{CW}d\Big(\frac{\alpha}{||f||_2}\Big)^{\frac{1}{d}}.
\end{align*}
\end{proof}
The next two theorems are directed towards applying the Carbery-Wright Theorem to a special polynomial which we will require later in Theorem~\ref{thm:normprob}. More specifically, let $U = (u_{ij})_{i,j \in [n]}$ be a matrix with bounded $\kappa_F$. Consider the linear form  $P^{k}(\textbf{x}) = \sum_{i \in [n]} p_{i}x_i$ where $p_{i} = u_{ik}$. We will apply the Carbery-Wright theorem to $P^{k}$.
\begin{theorem}\label{thm:CWapplied1}
Let $U = (u_{ij}) \in \text{GL}_n(\C)$ be such that $\kappa_{F}(U) \leq B$. Let $P^{k}(\textbf{x}) = \sum_{i \in [n]} p_{i}x_i$ where $p_{i} = u_{ik}$. Then 
\begin{align*}
     \text{Pr}_{\textbf{v} \in_U [-1,1)^n}[|f(\textbf{v})| \geq \frac{\alpha}{\sqrt{3B}}] \geq 1- 2C_{CW}\alpha.
\end{align*}
\end{theorem}
\begin{proof}
Applying Theorem \ref{thm:CWcomplex} to $P^k$ for $d = 1$, we get that
\begin{equation}\label{eq:2}
    \text{Pr}_{\textbf{v} \in_U [-1,1)^n}[|f(\textbf{v})| \geq \alpha||P^{k}||_2] \geq 1- 2C_{CW}\alpha.
\end{equation}
Now we claim that $||P^{k}||^2_2 \geq  \frac{1}{B}$. Expanding the $l_2$-norm of $P^{k}$ shows that
\begin{align*}
    ||P^{k}||_2^2 = \int_{[-1,1]^{n}} |P^{k}(\textbf{x})|^2d\mu(\textbf{x})
    \end{align*}
where $\mu(\textbf{x})$ is the uniform probability distribution on $[-1,1]^{n}$. Let $p^{(r)}_k$ and $p^{(i)}_k$ be the real and imaginary parts respectively of $p_k$. We have: 
\begin{align*}
    ||P^{k}||_2^2 &= \int_{[-1,1]^{n}} |\sum_{i \in [n]} p_{i}x_i|^2 d\mu(\textbf{x}) \\
    &=  \int_{[-1,1]^{n}} |\Big(\sum_{k \in [n]} p_k^{(r)}x_k\Big) + \iota\Big(\sum_{k \in [n]} p_k^{(i)}x_k\Big)|^2 d\mu(\textbf{x}) \\
    &= \int_{[-1,1]^{n}} \Big(\sum_{k \in [n]} p_k^{(r)}x_k\Big)^2 d\mu(\textbf{x}) + \int_{[-1,1]^{n}} \Big(\sum_{k \in [n]} p_k^{(i)}x_k\Big)^2 d\mu(\textbf{x}) \\
    &=  \sum_{k,l \in [n]} p^{(r)}_{k}p^{(r)}_{l} (\int_{[-1,1]^{n}}x_k x_l d\mu(\textbf{x})) + \sum_{k,l \in [n]} p^{(i)}_{k}p^{(i)}_{l} (\int_{[-1,1]^{n}}x_k x_l d\mu(\textbf{x})).
\end{align*}

When $i \neq k$, $\int_{[-1,1]^{n}}x_ix_j d\mu(\textbf{x}) = (\frac{1}{2}\int_{-1}^1x_idx_i)^2 = 0$. Therefore, 
\begin{align*}
    ||P^{k}||_2^2 = \sum_{k \in [n]} \frac{(p^{(r)}_{k})^2 + (p^{(i)}_{k})^2}{2} (\int_{-1}^1x_i^2dx_i) 
\end{align*}
Using the fact that $\frac{1}{2}\int_{-1}^1x_i^2dx_i = \frac{1}{3}$, we have 
   $||P^{k}||_2^2 = \frac{1}{3}\sum_{i \in [n]} |p_{i}|^2$.
Lemma \ref{lem:normlowerbound1} already implies that 
\begin{align*}
    ||P^{k}||_2^2 \geq \frac{1}{3B}.
\end{align*}
Using this in (\ref{eq:2}), we conclude that
\begin{align*}
     \text{Pr}_{\textbf{v} \in_U [-1,1)^n}[|f(\textbf{v})| \geq \frac{\alpha}{\sqrt{3B}}] \geq 1- 2C_{CW}\alpha.
\end{align*}
\end{proof}
The next two theorems are directed towards applying the Carbery-Wright theorem to another special polynomial which we will require later in Theorem \ref{thm:gapbound}. More specifically, let $U = (u_{ij})_{i,j \in [n]}$ be a matrix with bounded $\kappa_F$. Let  $P^{kl}(\textbf{x},\textbf{y}) = \sum_{i,j \in [n]} p^{kl}_{ij}x_iy_j$ be the quadratic polynomial defined for all $k,l \in [n]$ by its coefficients $p^{kl}_{ij} = u_{ik}u_{jl} - u_{il}u_{jk}$. We will apply the Carbery-Wright theorem to~$P^{kl}$. First we give a lower bound for the $l_2$ norm of the polynomial.
\begin{lemma}\label{lem:normlowerbound}
Let $U = (u_{ij}) \in \text{GL}_n(\C)$ be such that $\kappa_{F}(U) \leq B$. Then, for all $k,l \in [n]$, $\sum_{i,j \in [n]} |u_{ik}u_{jl} - u_{il}u_{jk}|^2 \geq \frac{2}{B^2}$.
\end{lemma}
\begin{proof}
We construct a submatrix $U_2 \in M_{n,2}(\C)$ with  the $k$-th and $l$-th columns of $U$. Let $k=1$ and $l=2$ without loss of generality. Since $\kappa_F(U) \leq B$, following the proof of Lemma \ref{lem:normlowerbound1}, for all $y \in \mathbb{C}^n$, $||Uy|| \geq \varepsilon||y||$ where $\varepsilon = \frac{1}{\sqrt{B}}$. Then for all $y \in \mathbb{C}^2$, we have $||U_2y|| \geq \varepsilon||y||$. This implies that
\begin{align*}
    ||U_2y||^2 &\geq \varepsilon^2||y||^2 \\
    y^*U_2^*U_2y &\geq \varepsilon^2 y^*y.
\end{align*}
The minimum singular value 
$\sigma_{\text{min}}$ of  $U_2$ is defined as $\sigma^2_{\text{min}} = \min_{y \in \C^n, y \neq 0} \frac{y^*U_2^*U_2y}{y^*y}$. Therefore, 
    $\sigma^2_{\text{min}}(U_2) \geq \varepsilon^2.$
Since $U_2^*U_2$ is a Hermitian matrix, $\sigma^2_{\text{min}}(U_2) = \lambda_{\text{min}}(U_2^*U_2)$ where $\lambda^2_{\text{min}}$ refers to the smallest eigenvalue. This gives us that
\begin{align*}
    \lambda_{\text{min}}(U_2^*U_2) \geq \varepsilon^2.
\end{align*}
Let $a = (a_1,...,a_n)$ and $b = (b_1,...,b_n)$ be the columns of $U_2$. Then $U_2^*U_2 = \begin{pmatrix}
||a||^2 & a^*b\\
b^*a & ||b||^2
\end{pmatrix}$. Also, $\text{det}(U_2^*U_2) \geq \lambda^2_{\text{min}}(U_2^*U_2)$, i.e.,
\begin{align*}
    ||a||^2||b||^2 - |a^*b|^2 \geq \lambda^2_{\text{min}}(U_2^*U_2) \geq \varepsilon^4.
\end{align*}
Now from the complex form of Lagrange's identity, we know that 
\begin{align*}
    ||a||^2||b||^2 - |a^*b|^2 = \frac{1}{2}\sum_{i,j =1}^n |a_ib_j - a_jb_i|^2.
\end{align*}
As a result, 
    $\sum_{i,j =1}^n |a_ib_j - a_jb_i|^2  \geq 2\varepsilon^4.$
Choosing $\varepsilon = \frac{1}{\sqrt{B}}$, we finally conclude that for all $k,l \in [n]$, $\sum_{i,j \in [n]} |u_{ik}u_{jl} - u_{il}u_{jk}|^2 \geq \frac{2}{B^2}$.
\end{proof}

\begin{theorem}\label{thm:CWapplied2}
Let $U = (u_{ij}) \in \text{GL}_n(\C)$ be such that $\kappa_{F}(U) \leq B$. Let $P^{kl}(\textbf{x},\textbf{y}) = \sum_{i,j \in [n]} p_{ij}x_iy_j$ where $p_{ij} = u_{ik}u_{jl} - u_{il}u_{jk}$. Then 
\begin{align*}
     \text{Pr}_{\textbf{v} \in_U [-1,1)^n}[|f(\textbf{v})| \geq \frac{\sqrt{2}\alpha}{3B}] \geq 1- 4C_{CW}\alpha^{\frac{1}{2}}.
\end{align*}
\end{theorem}
\begin{proof}
Applying Theorem \ref{thm:CWcomplex} to $P^{kl}$ with $d = 2$ shows that
\begin{equation}\label{eq:1}
    \text{Pr}_{\textbf{v} \in_U [-1,1)^n}[|f(\textbf{v})| \geq \alpha||P^{kl}||_2] \geq 1- 4C_{CW}\alpha^{\frac{1}{2}}.
\end{equation}
Now we claim that $||P^{kl}||_2 \geq  \frac{\sqrt{2}}{3B}$. 
Recall that
\begin{align*}
    ||P^{kl}||_2^2 = \int_{[-1,1]^{2n}} |P^{kl}(\textbf{x},\textbf{y})|^2d\mu(\textbf{x},\textbf{y})
    \end{align*}
where $\mu(\textbf{x},\textbf{y})$ is the uniform probability distribution on $[-1,1]^{2n}$.
Let us define $p^{(r)}_{ij}$ and $p^{(i)}_{ij}$ as the real and imaginary parts respectively of $p_{ij}$. 
We can estimate $|P^{kl}||_2^2$ as follows:
\begin{align*}
   & \int_{[-1,1]^{2n}} |\sum_{i,j \in [n]} p_{ij}x_iy_j|^2 d\mu(\textbf{x},\textbf{y}) 
    = \int_{[-1,1]^{2n}} |\Big(\sum_{i,j \in [n]} p^{(r)}_{ij}x_iy_j\Big) + \iota \Big(\sum_{i,j \in [n]} p^{(i)}_{ij}x_iy_j\Big)|^2 d\mu(\textbf{x},\textbf{y}) \\
   & =  \Big(\sum_{i,j,k,l \in [n]} p^{(r)}_{ij}p^{(r)}_{kl} (\int_{[-1,1]^{2n}}x_iy_jx_ky_l d\mu(\textbf{x}\textbf{y}))\Big) + \Big(\sum_{i,j,k,l \in [n]} p^{(i)}_{ij}p^{(i)}_{kl} (\int_{[-1,1]^{2n}}x_iy_jx_ky_l d\mu(\textbf{x}\textbf{y})\Big)\\
   &  =  \sum_{i,j,k,l \in [n]} (p^{(r)}_{ij}p^{(r)}_{kl} + p^{(i)}_{ij}p^{(i)}_{kl}) (\int_{[-1,1]^{n}}x_ix_kd\mu(\textbf{x})) (\int_{[-1,1]^{n}} y_jy_l d\mu(\textbf{y})).
    \end{align*}
When $i \neq k$, $\int_{[-1,1]^{n}}x_ix_kd\mu(\textbf{x}) = (\int_{[-1,1]}x_id\mu(x_i))^2 = 0$. Similarly $\int_{[-1,1]^{n}}y_jy_l d\mu(\textbf{y}) = 0$ for $j \neq l$.  This gives us that 
\begin{align*}
    ||P^{kl}||_2^2 = \sum_{i,j \in [n]} \Big((p^{(r)}_{ij})^2 + (p^{(i)}_{ij})^2\Big) (\int_{[-1,1]^n}x_i^2 d\mu(\textbf{x})) (\int_{[-1,1]^n} y_j^2 d\mu(\textbf{y}))
\end{align*}
Since $\int_{[-1,1]^{n}}x_i^2d\mu(\textbf{x}) = \frac{1}{2}\int_{-1}^1 x^2_idx_i = \int_{[-1,1]^{n}} y_j^2 d\mu(\textbf{y}) = \frac{1}{2}\int_{-1}^1 y_j^2 dy_j = \frac{1}{3}$, we get that
\begin{align*}
    ||P^{kl}||_2^2 = \frac{1}{9}\sum_{i,j \in [n]} |p_{ij}|^2.
\end{align*}
Now, from Lemma \ref{lem:normlowerbound}, it follows that
    $||P^{kl}||_2^2 \geq \frac{2}{9B^2}.$
Using this in (\ref{eq:1}), we can conclude that
\begin{align*}
     \text{Pr}_{\textbf{v} \in_U [-1,1)^n}[|f(\textbf{v})| \geq \frac{\sqrt{2}\alpha}{3B}] \geq (1- 4C_{CW}\alpha^{\frac{1}{2}}).
\end{align*}
\end{proof}
Our next goal is to show a similar probabilistic result for both families of polynomials (linear and quadratic), but replacing the previous continuous distribution over $[-1,1)^n$ by a distribution where the inputs are chosen uniformly and independently at random from a discrete grid. To formalise this distribution, we describe another equivalent random process of picking an element at random from $[-1,1)^n$ and rounding it to the nearest point on the grid.
\begin{definition}[Rounding function]\cite{forbes2017pspace}\label{def:roundingfunction}
Given $\eta \in (0,1)$ such that $\frac{1}{\eta}$ is an integer, for any point $(\textbf{a},\textbf{b}) \in [-1,1]^{2n}$, we define $g_{\eta}(\textbf{a},\textbf{b})$ to be the point $(\textbf{a}',\textbf{b}')$ such that the $i$-th element, $(\textbf{a}',\textbf{b}')_i = m_i\eta$, where $ m_i\eta \leq (\textbf{a},\textbf{b})_i < (m_i+1)\eta$.
\par
We also define $G_{\eta} = \{-1,-1+\eta,-1+2\eta,...,1-2\eta,1-\eta\}^{2n}$. Note here that for any point $(a,b) \in [-1,1)^{2n}$, $g_{\eta}(a,b) \in G_{\eta}$. Also, note that the process of picking $(a,b)$ uniformly and independently at random from $[-1,1)^{2n}$ and then using the rounding function $g_{\eta}$ on $(a,b)$ is equivalent to the process of picking an element uniformly and independently at random from $G_{\eta}$.
\end{definition}
\begin{theorem}[Multivariate Markov's Theorem]\label{thm:markov}
Let $f : \mathbb{R}^n \xrightarrow{} \mathbb{R}$ be a homogeneous polynomial of degree $r$, that for every $\textbf{v} \in [-1,1]^n$ satisfies $|f(\textbf{v})| \leq 1$. Then, for every $||\textbf{v}|| \leq 1$, it holds that $||\nabla(f)(\textbf{v})|| \leq 2r^2$.
\end{theorem}
\begin{theorem}\label{thm:norminf}
Let $f: \mathbb{R}^{2n} \xrightarrow{} \C$ be a homogeneous polynomial of degree at most $d$. Let $\eta > 0$ be such that $\frac{1}{\eta}$ is an integer. Let $\textbf{a},\textbf{b} \in [-1,1)^{2n}$ and $(\textbf{a}',\textbf{b}') = g_{\eta}(\textbf{a},\textbf{b})$. Then $|f(\textbf{a},\textbf{b}) - f(\textbf{a}',\textbf{b}')| \leq 4\eta\sqrt{n}||f||_{\infty}d^2$.
\end{theorem}
\begin{proof}
We write $f = \mathfrak{R}(f) + \iota \mathfrak{I}(f)$ where $\mathfrak{R}(f),  \mathfrak{I}(f): \R^n \xrightarrow[]{} \R$. By the mean value theorem, there exists a point $(\textbf{a}_0,\textbf{b}_0)$ on the line segment connecting $(\textbf{a},\textbf{b})$ and $(\textbf{a}',\textbf{b}')$, such that $|\mathfrak{R}(f)(\textbf{a},\textbf{b}) - \mathfrak{R}(f)(\textbf{a}',\textbf{b}')| = ||(\textbf{a},\textbf{b}) - (\textbf{a}',\textbf{b}')||\cdot |(\mathfrak{R}(f))'(\textbf{a}_0,\textbf{b}_0)|$ where $(\mathfrak{R}(f))'(\textbf{a}_0,\textbf{b}_0)$ is the derivative of $\mathfrak{R}(f)$ in the direction $(\textbf{a},\textbf{b}) - (\textbf{a}',\textbf{b}')$ evaluated at $\textbf{a}_0,\textbf{b}_0$. From Theorem \ref{thm:markov}, it follows that $|(\mathfrak{R}(f))'(\textbf{a}_0,\textbf{b}_0)| \leq 2 ||\mathfrak{R}(f)||_{\infty}d^2$. Similarly, we also get that $|(\mathfrak{I}(f))'(\textbf{a}_0,\textbf{b}_0)| \leq 2 ||\mathfrak{I}(f)||_{\infty}d^2$. This finally gives us that
\begin{align*}
    |f(\textbf{a},\textbf{b}) - f(\textbf{a}',\textbf{b}')| &= |\Big(\mathfrak{R}(f)(\textbf{a},\textbf{b}) - \mathfrak{R}(f)(\textbf{a}',\textbf{b}')\Big) + \iota \Big(\mathfrak{I}(f)(\textbf{a},\textbf{b}) - \mathfrak{I}(f)(\textbf{a}',\textbf{b}')\Big)| \\
    &= \sqrt{\Big(\mathfrak{R}(f)(\textbf{a},\textbf{b}) - \mathfrak{R}(f)(\textbf{a}',\textbf{b}')\Big)^2 + \Big(\mathfrak{I}(f)(\textbf{a},\textbf{b}) - \mathfrak{I}(f)(\textbf{a}',\textbf{b}')\Big)^2} \\
    &\leq ||(\textbf{a},\textbf{b}) - (\textbf{a}',\textbf{b}')||\cdot \sqrt{4 ||\mathfrak{R}(f)||^2_{\infty}d^4 + 4 ||\mathfrak{I}(f)||^2_{\infty}d^4} \leq 4\eta\sqrt{n}||f||_{\infty}d^2.
\end{align*}
The last inequality follows from the fact that $||\mathfrak{R}(f)||_{\infty},||\mathfrak{I}(f)||_{\infty} \leq ||f||_{\infty}$.
\end{proof}
\begin{theorem}Let $U = (u_{ij}) \in \text{GL}_n(\C)$ be such that $\kappa_{F}(U) \leq B$. Let $P^{kl}(\textbf{x},\textbf{y}) = \sum_{i,j \in [n]} p_{ij}x_iy_j$ where $p_{ij} = u_{ik}u_{jl} - u_{il}u_{jk}$. Let $C_{CW}$ be the absolute constant guaranteed by Theorem \ref{thm:CW}. Then
\begin{align*}
    \text{Pr}_{(\textbf{a},\textbf{b}) \in_U G_{\eta}}[|P^{kl}(\textbf{a},\textbf{b})| \geq \frac{\sqrt{2}\alpha}{3B} - 16\eta n^{\frac{3}{2}}B] \geq 1- 4C_{CW}\alpha^{\frac{1}{2}}.
\end{align*}
\end{theorem}
\begin{proof}
Using Theorem \ref{thm:norminf} for $f = P^{kl}$ where $d = 2$, we already have that $|P^{kl}(\textbf{a},\textbf{b}) - P^{kl}(\textbf{a}',\textbf{b}')| \leq 16\eta\sqrt{n}||P^{kl}||_{\infty}$. Since $(\textbf{a}',\textbf{b}')$ is selected uniformly at random from $[-1,1]^{2n}$, by Theorem \ref{thm:CWapplied2} we have
$|P^{kl}(\textbf{a}',\textbf{b}')| \geq \frac{\sqrt{2}\alpha}{3B}$ with probability at least $ (1- 4C_{CW}\alpha^{\frac{1}{2}})$. This gives us that
\begin{equation}\label{eq:discreteCW}
    |P^{kl}(\textbf{a}',\textbf{b}')| \geq \frac{\sqrt{2}\alpha}{3B} - 16\eta\sqrt{n}||P^{kl}||_{\infty}
\end{equation}
Now we claim that $||P^{kl}||_{\infty} \leq Bn$. 
Indeed,
\begin{align*}
    ||P^{kl}||_{\infty} &= \max_{v \in [-1,1]^{2n}} |P^{kl}(v)| \\
    &\leq \sum_{i,j \in [n]} |u_{ik}u_{jl} - u_{jk}u_{il}| \\
    &\leq \sum_{i,j \in [n]} (|u_{ik}u_{jl}| + |u_{jk}u_{il}|) \\
    &\leq \sum_{i,j \in [n]} (\frac{|u_{ik}|^2 + |u_{jl}|^2}{2} + \frac{u_{jk}^2 + u_{il}^2}{2}) \\
    &\leq n||U||^2_{F} \leq Bn.
\end{align*}
Putting this in (\ref{eq:discreteCW}), we can conclude that
\begin{align*}
    \text{Pr}_{(\textbf{a},\textbf{b}) \in_U G_{\eta}}[|P^{kl}(\textbf{a},\textbf{b})| \geq \frac{\sqrt{2}\alpha}{3B} - 16\eta Bn^{\frac{3}{2}}] \geq 1- 4C_{CW}\alpha^{\frac{1}{2}}.
\end{align*}
\end{proof}
\begin{corollary}\label{corr:gapnumer}
Let $U = (u_{ij}) \in \text{GL}_n(\C)$ be such that $\kappa_{F}(U) \leq B$. Let $P^{kl}(\textbf{x},\textbf{y}) = \sum_{i,j \in [n]} p_{ij}x_iy_j$ where $p_{ij} = u_{ik}u_{jl} - u_{il}u_{jk}$. Let $C_{CW}$ be the absolute constant guaranteed by Theorem \ref{thm:CW}. Then
\begin{align*}
    \text{Pr}_{(\textbf{a},\textbf{b}) \in_U G_{\eta}}[|P^{kl}(\textbf{a},\textbf{b})| \geq k] \geq (1- 4C_{CW}\Big(\frac{3B(k + 16\eta Bn^{\frac{3}{2}})}{\sqrt{2}}\Big)^{\frac{1}{2}}).
\end{align*}
\end{corollary}
In the next theorem, we give a similar result for the polynomial in Theorem \ref{thm:CWapplied1}. We will require this later in Theorem \ref{thm:normprob}. First we give a lower bound for the $l_2$ norm of the polynomial.
\begin{theorem}\label{thm:normdenom}
Let $U = (u_{ij}) \in \text{GL}_n(\C)$ be such that $\kappa_{F}(U) \leq B$. Let $P^{k}(\textbf{x}) = \sum_{i \in [n]} p_{i}x_i$ where $p_{i} = u_{ik}$. Let $C_{CW}$ be the absolute constant guaranteed in Theorem \ref{thm:CW}. Then
\begin{align*}
    \text{Pr}_{(\textbf{a},\textbf{b}) \in_U G_{\eta}}[|P^{k}(\textbf{a})| \geq \frac{\alpha}{\sqrt{3B}} - \eta\sqrt{nB}] \geq 1- 2C_{CW}\alpha.
\end{align*}
\end{theorem}
\begin{proof}
Using the fact that $P^{k}$ is a linear polynomial and using the Cauchy-Schwarz inequality, we get that
\begin{align*}
    &||P^{k}(\textbf{a}) - P^{k}(\textbf{a}')||^2 \\
    &= ||\sum_{i \in [n]} p_{i}(a_i-a'_i)||^2 \\
    &\leq ||\textbf{a} - \textbf{a}'||^2||\sum_{i=1}^n |p_i|^2|| \\
    &\leq ||\textbf{a} - \textbf{a}'||^2||U||_F^2 \\
    &\leq ||\textbf{a} - \textbf{a}'||^2\kappa_F(U) < n\eta^2B
\end{align*}
This gives us that  $||P^{k}(\textbf{a}) - P^{k}(\textbf{a}')|| \leq \eta\sqrt{nB}$. Since $\textbf{a}$ is selected uniformly at random from $[-1,1]^{n}$, using Theorem \ref{thm:CWapplied1}, we have that 
$|P^{k}(\textbf{a})| \geq \frac{\alpha}{\sqrt{3B}}$ with probability at least $ (1- 2C_{CW}\alpha)$. This gives us that
\begin{equation}
    |P^{k}(\textbf{a}')| \geq \frac{\alpha}{\sqrt{3B}} - \eta\sqrt{nB}
\end{equation}
with probability at least $ 1- 2C_{CW}\alpha$.
\end{proof}
{ Finally, we give a lemma that will be needed in Section~\ref{sec:probanalysisproof}}.
\begin{lemma}\label{lem:gapdenom}
Let $U = (u_{ij}) \in M_n(\C)$ be such that $\kappa_F(U) \leq B$. Then, given $\textbf{a} \in [-1,1]^n$, for all $k,l \in [n]$, $| (\sum_{i \in [n]}a_iu_{ki})(\sum_{j \in [n]}a_ju_{lj})| \leq \frac{nB}{2}$.
\end{lemma}
\begin{proof}
\begin{align*}
 &| (\sum_{i \in [n]}a_iu_{ki})(\sum_{j \in [n]}a_ju_{lj})|   \\
 & =  | (\sum_{i \in [n]}a_iu_{ki})||(\sum_{j \in [n]}a_ju_{lj})| \\
 &\leq ||\textbf{a}||^2||u_k||||u_l||
\end{align*}
where $u_k$ and $u_l$ are the $k$-th and $l$-th rows of $U$ respectively. Now we get that
\begin{align*}
    ||u_k||||u_l|| \leq \frac{||u_k||^2 + ||u_l||^2}{2} \leq \frac{(\kappa_F(U))}{2} \leq \frac{B}{2}.
\end{align*}
Since $\textbf{a} \in [-1,1]^n$, then $||\textbf{a}|| \leq \sqrt{n}$. Combining these inequalities yields the desired result. 
\end{proof}

\subsection{Towards a proof of Theorem \ref{th:main}}\label{sec:probanalysisproof}
Let $T_1,...,T_n$ be the slices of the tensor $T$ given as input to Algorithm \ref{algo:Jennrich}. Let $k_{\text{gap}}, k_F$ be the parameters as set in Algorithm \ref{algo:Jennrich}. Let $a_1,...,a_n,b_1,...,b_n$ be picked uniformly and independently at random from a finite grid $G_{\eta} \subset~[-1,1]^{2n}$ (as defined in Definition \ref{def:roundingfunction}). Let $T^{(a)} = \sum_{i=1}^n a_iT_i$ and $T^{(b)} = \sum_{i=1}^n b_iT_i$. Recall from (\ref{eq:gap}) that for a matrix $A$, $\text{gap}(A)$ is defined as the minimum distance 
between its eigenvalues. In this subsection, as claimed in Section \ref{sec:properror}, we show that $T^{(a)}$ is invertible, $\text{gap}((T^{(a)})^{-1}T^{(b)}) \geq k_{\text{gap}}$ and $\kappa_F(T^{(a)}) \leq k_F$ with high probability.  
\begin{theorem}\label{thm:gapbound}
Let $T \in (\C^{n})^{\otimes 3}$ be a diagonalisable order-$3$ symmetric tensor such that $\kappa(T) \leq B$. We denote by $T_1,...,T_n$ the slices of $T$. Let $(a_1,...,a_n,b_1,...,b_n) \in \mathbb{R}^{2n}$ be picked from $G_{\eta}$ uniformly at random and set $T^{(a)} := \sum_{i=1}^n a_iT_i$, $T^{(b)} := \sum_{i=1}^n b_iT_i$. If $T^{(a)}$ is invertible, let $T^{(a)'} = (T^{(a)})^{-1}$. Then for any $k_{\text{gap}} > 0$, we have that
\begin{align*}
    \text{Pr}_{(\textbf{a},\textbf{b}) \in_U G_\eta}[T^{(a)} \text{ is invertible and } \text{gap}(T^{(a)'}T^{(b)}) \geq k_{\text{gap}}] \geq 1 -  \Big(4n^2C_{CW}(\frac{3B\alpha_{\text{gap}}}{\sqrt{2}})^{\frac{1}{2}} + \frac{n\eta}{2}\Big)
\end{align*}
where $\alpha_{\text{gap}} = \frac{nBk_{\text{gap}}}{2} + 16\eta Bn^{\frac{3}{2}}$.
\end{theorem}
\begin{proof}
Let $U$ be the matrix with rows $u_1,...,u_n$ such that $T = \sum_{i=1}^n u_i^{\otimes 3}$ and $\kappa(T) = \kappa_F(U) \leq B$. If $a_1,...,a_n$ are picked independently and uniformly at random from a finite set $S$, 
from (\ref{eq:tainvertible}), we get that $T^{(a)}$ is invertible with probability at least $1 - \frac{n}{|S|}$. We use this for $S = \{-1, -1+ \eta,...,1-2\eta, 1-\eta\} \subset [-1,1]$.
Since $|S| = \frac{2}{\eta}$, if $a$ is picked uniformly and independently at random from $G_{\eta}$, $\langle a,u_k \rangle = 0$  with probability at most $\frac{\eta}{2}$. { Recall the definition of $D^{(a)}$ in Theorem~\ref{thm:P3structural}.}
It follows from the union bound that $\text{det}(D^{(a)}) \neq 0$ with probability at least $1- \frac{n\eta}{2}$.
\par
If $T^{(a)}$ is invertible, let $\lambda_1,...,\lambda_n$ be the eigenvalues of $T^{(a)'}T^{(b)}$. Then by Theorem \ref{thm:P3structural} (more precisely the fact that $T^{(a)} = U^TD^{(a)}U$), we get that $\lambda_k = \frac{\langle b,u_k \rangle}{\langle a,u_k \rangle}$ where $u_k$ are the rows of $U$ and $\langle a,u_k \rangle \neq 0$. Hence
\begin{align*}
    gap(T^{(a)'}T^{(b)}) &= \min_{k \neq l \in [n]} \Big|\frac{\langle b,u_k \rangle}{\langle a,u_k \rangle} - \frac{\langle b,u_l \rangle}{\langle a,u_l \rangle}\Big| \\
    &= \min_{k \neq l \in [n]} \Big|\frac{\langle b,u_k \rangle\langle a,u_l \rangle - \langle b,u_l \rangle\langle a,u_k \rangle}{\langle a,u_k \rangle\langle a,u_l \rangle}\Big|
\end{align*}
By Corollary \ref{corr:gapnumer}, 
if $\textbf{a}$ is picked from $G_{\eta}$ uniformly at random, then $|\langle b,u_k \rangle\langle a,u_l \rangle - \langle b,u_l \rangle\langle a,u_k \rangle| < t$ with probability at most $4C_{CW}(\frac{3B}{\sqrt{2}} (t + 16\eta Bn^{\frac{3}{2}}))^{\frac{1}{2}}$. Combining these results with the union bound, we get that 
\begin{align*}
    &\text{Pr}_{a,b \in G_{\eta}}[\exists k,l \in [n] |\langle b,u_k \rangle\langle a,u_l \rangle - \langle b,u_l \rangle\langle a,u_k \rangle| < t \text{ } \cup T^{(a)} \text{ } \text{ is not invertible}] \\ 
    &\leq 4n^2C_{CW}(\frac{3B}{\sqrt{2}} (t + 16\eta Bn^{\frac{3}{2}}))^{\frac{1}{2}} + \frac{n\eta}{2}.
\end{align*}
This gives us that
\begin{align*}
    &\text{Pr}_{a,b \in G_{\eta}}[T^{(a)} \text{ is invertible and }\text{for all } k,l \in [n] |\langle b,u_k \rangle\langle a,u_l \rangle - \langle b,u_l \rangle\langle a,u_k \rangle| > t ] \\
    &\geq 1 -  (4n^2C_{CW}(\frac{3B}{\sqrt{2}} (t + 16\eta Bn^{\frac{3}{2}}))^{\frac{1}{2}} + \frac{n\eta}{2}).
\end{align*}
Now if $|\langle b,u_k \rangle\langle a,u_l \rangle - \langle b,u_l \rangle\langle a,u_k \rangle| > t$,  we have that 
\begin{align*}
    \text{gap}(T^{(a)'}T^{(b)}) > t\min_{k \neq l \in [n]} \frac{1}{|\langle a,u_k \rangle\langle a,u_l \rangle|}
\end{align*}
By Lemma \ref{lem:gapdenom},  since $\kappa_F(U) \leq B$ we have  $|\langle a,u_k \rangle\langle a,u_l \rangle| \leq \frac{nB}{2}$ for all $\textbf{a} \in G_{\eta} \subseteq [-1,1]^n$. This implies that 
$\text{gap}(T^{(a)'}T^{(b)}) > \frac{2t}{nB}$. Finally, setting $t = \frac{nBk_{\text{gap}}}{2}$, we get the desired conclusion.
\end{proof}
\begin{theorem}\label{thm:normprob}
Let $T \in \C^{n \times n \times n}$ be a diagonalisable degree-$3$ symmetric tensor such that $\kappa(T) \leq B$, where $T_1,...,T_n$ are the slices of $T$. Let $a \in~[-1,1]^{2n}$ be picked from $G_{\eta}$ uniformly at random and set $T^{(a)} := \sum_{i=1}^n a_iT_i$. If $T^{(a)}$ is invertible, let $T^{(a)'} = (T^{(a)})^{-1}$. Then for all $k_F > nB^3$, we have that
\begin{align*}
    \text{Pr}_{\textbf{a} \in_U G_\eta}[T^{(a)} \text{ is invertible and } \kappa_F(T^{(a)}) \leq k_{F}] \geq  1- (2nC_{CW}\alpha_F + \frac{n\eta}{2})
\end{align*}
where $\alpha_{F} = \sqrt{3B}(\sqrt{\frac{nB^2}{k_F - nB^3}}+ \eta \sqrt{nB})$.
\end{theorem}
\begin{proof}
Let $U$ be the matrix with rows $u_1,...,u_n$ such that $T = \sum_{i=1}^n u_i^{\otimes 3}$ and $\kappa(T) = \kappa_F(U) \leq B$. Since $a$ is picked uniformly and independently from $G_{\eta}$, following the proof of Theorem \ref{thm:gapbound}, 
$T^{(a)}$ is invertible with probability at least $(1- \frac{n\eta}{2})$.
If $T^{(a)}$ is invertible, using Theorem \ref{thm:P3structural}, and more precisely the fact that $T^{(a)} = U^TD^{(a)}U$, we have:
\begin{align*}
    ||T^{(a)'}||_F &\leq ||U^{-1}||^2_F||(D^{(a)})^{-1}||_F \\
    &\leq  \kappa_F(U)||(D^{(a)})^{-1}||_F \\
    &\leq B||(D^{(a)})^{-1}||_F.
\end{align*}
Now,  $||(D^{(a)})^{-1}||_F^2 = \sum_{i=1}^n \frac{1}{|\langle a, u_i \rangle|^2}$. By Theorem \ref{thm:normdenom},  if $\textbf{a}$ is picked from $G_{\eta}$ uniformly at random, 
then $|\langle a, u_i \rangle| \geq k$ with probability at least $1 - 2C_{CW}(\sqrt{3B}(k+ \eta \sqrt{nB}))$. This gives us that
\begin{align*}
    &\text{Pr}_{a \in G_{\eta}}[\exists m \in [n] |\langle a,u_m \rangle| \leq k \text{ } \cup  \text{ }  T^{(a)}\text{ is not invertible}] \\
    &\leq \sum_{m=1}^n\text{Pr}_{a \in G_{\eta}}[|\langle a,u_m \rangle| \leq k] + \text{Pr}_{a \in G_{\eta}}[T^{(a)} \text{ } \text{ is not invertible}] \\
    &\leq 2nC_{CW}(\sqrt{3B}(k+ \eta \sqrt{nB})) + \frac{n\eta}{2}.
\end{align*}
As a result, 
\begin{align*}
    &\text{Pr}_{a \in G_{\eta}}[\text{for all } m \in [n] |\langle a,u_m \rangle| \geq k \text{ and }  T^{(a)} \text{ is invertible}] \\
    &\geq 1- (2nC_{CW}(\sqrt{3B}(k+ \eta \sqrt{nB})) + \frac{n\eta}{2}).
\end{align*}
By Lemma \ref{lem:normnum},  $||D^{(a)}||^2 \leq nB$. This further implies that if $|\langle a,u_m \rangle| \geq k$ for all $m$, then  $||(D^{(a)})^{-1}||_F^2 + ||D^{(a)}||_F^2 \leq \frac{n}{k^2} + nB$, which in turn implies that $\kappa_F(T^{(a)}) = ||T^{(a)'}||_F^2 + ||T^{(a)}||^2_F \leq \frac{nB^2}{k^2} + nB^3 $. Setting $k = \sqrt{\frac{nB^2}{k_F - nB^3}}$ gives the desired conclusion.
\end{proof}

\begin{theorem}\label{thm:mainprob}
Let $T \in \mathbb{C}^{n \times n \times n}$ be a diagonalisable degree-$3$ symmetric tensor such that $\kappa(T) \leq B$. Let $T_1,...,T_n$ be the slices of $T$ and given $a,b$ picked uniformly and independently at random from $G_{\eta}$, set $T^{(a)} :=~\sum_{i=1}^n a_iT_i$ and $T^{(b)} := \sum_{i=1}^n b_iT_i$. If $T^{(a)}$ is invertible, let $T^{(a)'} = (T^{(a)})^{-1}$. We assume that $l_1,...,l_n$ is the output returned by Algorithm \ref{algo:Jennrich} on input $T$, $B$ and an accuracy parameter $\varepsilon$. Let $k_{\text{gap}}$ and $k_F$ be as defined in Theorem \ref{thm:mainerrorfinal}. Then there exist cube roots of unity $\omega_i$ such that $||\omega_iu_i - l_i|| < \varepsilon$, $T^{(a)}$ is invertible, $\text{gap}(T^{(a)'}T^{(b)}) \geq k_{\text{gap}}$ and $\kappa_F(T^{(a)}) \leq k_F$ with probability at least $$\Big(1- \frac{1}{n} - \frac{12}{n^2}\Big)\Big(1 - \Big(nC_{CW}\alpha_F + 4n^2C_{CW}(\frac{3B\alpha_{\text{gap}}}{\sqrt{2}})^{\frac{1}{2}} + n\eta\Big)\Big)$$
where $\alpha_{\text{gap}} = \frac{nBk_{\text{gap}}}{2} + 16\eta Bn^{\frac{3}{2}}$ and $\alpha_{F} = \sqrt{3B}(\sqrt{\frac{nB^2}{k_F - nB^3}}+ \eta \sqrt{nB})$.
\end{theorem}
\begin{proof}
Let $E_1$ be the event that there exist cube roots of unity $\omega_i$ with $||\omega_iu_i - l_i|| < \varepsilon$. Let $E_2$ be the event that $\text{gap}(T^{(a)'}T^{(b)}) \geq k_{\text{gap}}$. We define $E_3$ to be the event that $\kappa_F(T^{(a)}) \leq k_F$ and $E_4$ to be the event that $T^{(a)}$ is invertible. We want to bound
\begin{align*}
    & \text{Pr}_{\textbf{a},\textbf{b} \in G_{\eta}}[E_1,E_2,E_3,E_4] \\
    &= \text{Pr}[E_1|E_2,E_3,E_4]\text{Pr}_{\textbf{a},\textbf{b} \in G_{\eta}}[E_2,E_3,E_4].
\end{align*}
Note here the probability in the first line and the first factor in the second line is also with respect to the internal choice of randomness in the diagonalisation algorithm (Algorithm \ref{algo:eigfwd}). We refrain from mentioning it at every step in order to make the equations more readable. 
\par
Using Theorem \ref{thm:mainerrorfinal}, we get that $\text{Pr}[E_1|E_2,E_3,E_4] \geq 1- \frac{1}{n} - \frac{12}{n^2}$. Using Theorem \ref{thm:gapbound}, we also have  $\text{Pr}_{\textbf{a},\textbf{b} \in G_{\eta}}[E_2,E_4] \geq 1 -  (4n^2C_{CW}(\frac{3B\alpha_{\text{gap}}}{\sqrt{2}})^{\frac{1}{2}} + \frac{n\eta}{2})$ where $\alpha_{\text{gap}} = \frac{nBk_{\text{gap}}}{2} + 16\eta Bn^{\frac{3}{2}}$. From Theorem \ref{thm:normprob}, we already know that $\text{Pr}_{\textbf{a},\textbf{b} \in G_{\eta}}[E_3,E_4] \geq 1- (n(C_{CW}(\alpha_F)) + \frac{n\eta}{2})$
where $\alpha_{F} = \sqrt{3B}(\sqrt{\frac{nB^2}{k_F - nB^3}}+ \eta \sqrt{nB})$. Combining these using the union bound shows that
\begin{align*}
    \text{Pr}_{\textbf{a},\textbf{b} \in G_{\eta}}[E_2,E_3,E_4] &\geq 1 - (n(C_{CW}(\alpha_F)) + \frac{n\eta}{2}) + (4n^2C_{CW}(\frac{3B\alpha_{\text{gap}}}{\sqrt{2}} )^{\frac{1}{2}} + \frac{n\eta}{2}) \\
    &= 1 - (n(C_{CW}(\alpha_F)) + (4n^2C_{CW}(\frac{3B\alpha_{\text{gap}}}{\sqrt{2}})^{\frac{1}{2}}) + n\eta).
\end{align*}
Multiplying this by $1- \frac{1}{n} - \frac{12}{n^2}$
gives the desired result.
\end{proof}
\subsection{Finishing the proof of Theorem \ref{th:main}}\label{sec:finishproofmain}
Let $T$ be the 
diagonalisable symmetric tensor    given as input and let $U \in \text{GL}_n(\C)$ 
be such that $U$ diagonalises $T$. Let $B$ be an estimate for $\kappa(T) = \kappa_F(U)$. Let $a,b$ be picked uniformly and independently at random from $G_{\eta}$ and define $T^{(a)} = \sum_{i=1}^n a_iT_i$, $T^{(b)} = \sum_{i=1}^n b_iT_i$ to be two linear combination of the slices $T_1,...,T_n$ of $T$.
Let $E_1$ be the event  that Algorithm \ref{algo:Jennrich} outputs an $\varepsilon$-approximate solution to the tensor decomposition problem, $E_2$ be the event  that $\text{gap}(T^{(a)'}T^{(b)}) \geq k_{\text{gap}}$, $E_3$ be the event  that $\kappa_F(T^{(a)}) \leq k_F$ and $E_4$ be the event  that $T^{(a)}$ is invertible. By Theorem \ref{thm:mainprob}, 
\begin{align*}
    \text{Pr}_{\textbf{a},\textbf{b} \in G_{\eta}}[E_1,E_2,E_3,E_4] \geq \Big(1- \frac{1}{n} - \frac{12}{n^2}\Big)\Big(1 - nC_{CW}\alpha_F + 4n^2C_{CW}(\frac{3B\alpha_{\text{gap}}}{\sqrt{2}})^{\frac{1}{2}} + n\eta\Big).
\end{align*}
{ As promised in Algorithm~\ref{algo:Jennrich}, we  define at last the constants $C_{\text{gap}}$ and $c_F$. Namely, we set}
\begin{equation}\label{eq:cgapcf}
    C_{\text{gap}} := \frac{1}{48\sqrt{2}C_{CW}^2} \text{ and } c_F =  96C^2_{CW} + 1.
\end{equation}
Since in Algorithm \ref{algo:Jennrich}, we set $k_{\text{gap}} = \frac{1}{48\sqrt{2}C_{CW}^2n^6B^3}$ and $\eta = \frac{1}{C_{\eta}n^{\frac{15}{2}}B^4}$, we have for large enough $n$,
\begin{align*}
    \alpha_{\text{gap}} &= \frac{nBk_{\text{gap}}}{2} + 16\eta Bn^{\frac{3}{2}} \\
    &= \frac{1}{96\sqrt{2}C_{CW}^2n^5B^2} + \frac{1}{C_{\eta}n^{\frac{15}{2}}B^4} \\
    &\leq \frac{1}{48\sqrt{2}C_{CW}^2n^5B^4}.
\end{align*}
This gives us that $(\frac{3B\alpha_{\text{gap}}}{\sqrt{2}})^{\frac{1}{2}} \leq \frac{1}{4C_{CW}\sqrt{2 n^5 B}}$, hence
\begin{equation}\label{eq:prb1}
    4n^2C_{CW}(\frac{3B\alpha_{\text{gap}}}{\sqrt{2}})^{\frac{1}{2}} \leq \frac{1}{\sqrt{2nB}} \leq \frac{1}{\sqrt{2n}}.
\end{equation}
The last inequality follows from the fact that $B > 1$.
We also set $k_F = (96C^2_{CW} + 1)n^5B^3$. Since $nB^3 < n^5B^3$, we have
\begin{align*}
    \alpha_F &= \sqrt{3B}(\sqrt{\frac{nB^2}{k_F - nB^3}} + \eta \sqrt{nB}) \\
    &=\sqrt{3B}(\sqrt{\frac{1}{96C^2_{CW}n^4B}} +  \frac{1}{C_{\eta}n^8B^{\frac{7}{2}}}) \\
    &\leq \frac{1}{8C_{CW}n^2} + \frac{\sqrt{3}}{C_{\eta}n^7B^3} \\
    &\leq \frac{1}{4C_{CW}n^2}
\end{align*}
This gives us that $2n_{C_{CW}}\alpha_F \leq \frac{1}{2n}$. Also, $\eta n = \frac{1}{C_{\eta}n^{\frac{13}{2}}B^4} \leq \frac{1}{2n}$. Combining these with (\ref{eq:prb1}) finally shows that 
\begin{align*}
    \text{Pr}_{\textbf{a},\textbf{b} \in G_{\eta}}[E_1] &\geq \text{Pr}_{\textbf{a},\textbf{b} \in G_{\eta}}[E_1,E_2,E_3,E_4] \\ &\geq \Big(1- \frac{1}{n} - \frac{12}{n^2}\Big)\Big(1 - 2nC_{CW}\alpha_F + 4n^2C_{CW}(\frac{3B\alpha_{\text{gap}}}{\sqrt{2}})^{\frac{1}{2}} + n\eta\Big) \\ 
    &\geq \Big(1- \frac{1}{n} - \frac{12}{n^2}\Big)\Big(1 - \frac{1}{\sqrt{2n}} - \frac{1}{n}\Big).
\end{align*}
\section*{Acknowledgement}
Subhayan Saha acknowledges the support by the European Union (ERC consolidator, eLinoR, no 101085607).  Most of this work was done when Subhayan Saha was with Univ Lyon, EnsL, UCBL, CNRS, LIP. 
\newcommand{\etalchar}[1]{$^{#1}$}
{\small 
\appendix
\section{Appendix: Technical lemmas for Section~\ref{sec:diag} and Section~\ref{sec:complete}}
\begin{lemma}\label{lem:boundsforu1}
Let $\textbf{u}$ be such that $\log(\frac{1}{u}) > \log^4(\frac{n}{\delta})\log n$ where $\delta < \frac{1}{2}$. Then for $n \geq 4$, $n^2\textbf{u} < \delta$.
\end{lemma}
\begin{proof}
From the hypothesis  $\delta < \frac{1}{2}$ it follows that 
$\log(\frac{1}{\delta}) > 1$ and  
$\log(\frac{1}{\delta}) \leq \log^4(\frac{1}{\delta})$. Also, from $n \geq 4$ it follows that 
$\log^4(n)-2 \geq 0$ and 
$\log(\frac{1}{\delta}) \leq \log^4(\frac{1}{\delta}) + \log^4(n)-2$. Now we use the fact that $\log^4(\frac{1}{\delta}) + \log^4(n) \leq \log^4(\frac{n}{\delta})$. This implies that
\begin{align*}
    \log(\frac{1}{\delta}) \leq (\log^4(\frac{n}{\delta}) - 2)\log n.
\end{align*}
Thereefore $\log(\frac{n^2}{\delta}) < \log^4(\frac{n}{\delta})\log n   \leq \log(\frac{1}{\textbf{u}})$, and 
$n^2\textbf{u} \leq \delta$.
\end{proof}
 
\begin{lemma}\label{lem:boundsforu}
For constants $c_{n_1}, c_{n_2}, c_{B_1}, c_{B_2}, c_{\varepsilon}, C > 0$, if $$\log(\frac{1}{\textbf{u}}) > c\log(\frac{nB}{\varepsilon})\log n,$$ where $c= 2\max\{c_{n_1}+c_{n_2}, c_{B_1}+c_{B_2}, c_{\varepsilon}, \log C\}$ then 
$$ C\frac{n^{c_{n_1}\log n + c_{n_2}}B^{c_{B_1}\log n + c_{B_2}}}{\varepsilon^{c_{\varepsilon}}} \leq \frac{1}{\textbf{u}}.$$
\end{lemma}
\begin{proof}
Let $c = 2\max\{c_{n_1}+c_{n_2}, c_{B_1}+c_{B_2}, c_{\varepsilon}, \log C\}$. Then 
\begin{align*}
    \log(\frac{1}{u}) &= c\log(\frac{nB}{\varepsilon})\log n \\
    &= c\log(nB) \log n + \log(\frac{1}{\varepsilon}) \log n \\
    &\geq (c_{n_1}\log n + c_{n_2})\log(n) + (c_{B_1}\log n + c_{B_2}) \log(B) + c_{\varepsilon} \log(\frac{1}{\varepsilon}) + \log(C) \\
    &\geq \log\Big(C\frac{n^{c_{n_1}\log n + c_{n_2}}B^{c_{B_1}\log n + c_{B_2}}}{\varepsilon^{c_{\varepsilon}}}\Big).
\end{align*}
\end{proof}

\section{Error accumulated in Steps 4,5,6,7 of Algorithm~\ref{algo:Jennrich}}\label{app:erroranalysis}
Let $T$ be the input diagonalisable tensor.
It can be decomposed as $T = \sum_{i=1}^n u_i^{\otimes 3}$ where the $u_i$'s are linearly independent vectors. Let $T^{(a)},T^{(b)}$ be two linear combinations of the slices $T_1,...,T_n$. 
\par
In exact arithmetic, we assume that $T^{(a)}$ is invertible and $(T^{(a)})^{-1}T^{(b)}$ have distinct eigenvalues. Recall $k_F,k_{\text{gap}}$ and $B$ as defined in Algorithm~\ref{algo:Jennrich}. In finite arithmetic, we additionally assume that  $\kappa_F(T^{(a)}) < k_F$ where $k_F > 1$ and $\text{gap}((T^{(a)})^{-1}T^{(b)}) > k_{\text{gap}}$. We have shown in Section \ref{subsection:correctnessjennrich} that the assumption on $k_{\text{gap}}$ and $k_F$ are satisfied with high probability. 
\subsection{Step 4}\label{sec:step4}
In this section, we estimate the error at the end of Step 4 of Algorithm \ref{algo:Jennrich}.
\textbf{Correspondence to exact arithmetic:}  If the algorithm is run in exact arithmetic, the "ideal input" to Step 4 is $(T^{(a)})^{-1}T^{(b)}$ and if we assume that we can diagonalise an input matrix exactly, then as seen in Algorithm \ref{algo:completeexact}, the "ideal output" at the end of Step 4 are the true normalized eigenvectors $p_1,...,p_n$ of $(T^{(a)})^{-1}T^{(b)}$. 
\par
\textbf{Finite arithmetic:} In finite arithmetic, we assume that Step 4 of Algorithm \ref{algo:Jennrich} takes in a matrix $D$ which is not too far from the ideal input $(T^{(a)})^{-1}T^{(b)}$. We first compute the estimates for $K_{\text{eig}} > \kappa_{\text{eig}}(D)$ and $K_{\text{norm}} > ||D||$ and then use the algorithm $EIG-FWD$ from Theorem~\ref{thm:eigfwd} on $(D,\delta,K_{\text{eig}},K_{\text{norm}})$ for some $\delta$ that we will fix later to return the approximate normalized eigenvectors of $D$. Let $v^{(0)}_1,...,v_n^{(0)}$ be the output of Step 4 of this algorithm. In this section, we show that the output $v_i^{(0)}$ of this step is close to the ideal output $p_i$.
Note that we will refer to $\delta$ in the subsequent steps as the accuracy parameter for Step 4. 
The following is the main result of this subsection:
\begin{theorem}\label{thm:errstep4}
Let $T$ be the diagonalisable tensor given as input to Algorithm \ref{algo:Jennrich} and let $U \in \text{GL}_n(\C)$ be the matrix that diagonalises $T$ where 
$\kappa(T) = \kappa_F(U) \leq B$.
Let $T^{(a)}, T^{(b)}$ be two linear combination of the slices $T_1,...,T_n$ of $T$ such that $\kappa_F(T^{(a)}) < k_F$ where $k_F > 1$ and $\text{gap}((T^{(a)})^{-1}T^{(b)}) > k_{\text{gap}}$. Let $p_1,...,p_n$ be the true normalized eigenvectors of $(T^{(a)})^{-1}T^{(b)}$. Let $D$ be the output of Step 3 of Algorithm \ref{algo:Jennrich} where $||D - (T^{(a)})^{-1}T^{(b)}|| \leq \varepsilon_3 <~\frac{k_{\text{gap}}}{4B}$. Let $v_1^{(0)},...,v_n^{(0)}$ be the output at the end of Step 4, 
run with accuracy parameter $\delta < \frac{k_{\text{gap}}}{24nB}$. Let $V^{(0)}$ be the matrix with columns $v_i^{(0)}$. Then with probability at least $1- \frac{1}{n} - \frac{12}{n^2}$,
\begin{itemize}
    \item $||p_i - v^{(0)}_i|| \leq (\frac{3nB\varepsilon_3}{k_{\text{gap}}} + \delta)$ for all $i \in [n]$ up to permutation and multiplication by phases.
    \item $ k_V  := \kappa_F(V^{(0)}) \leq 800n^6B^2$.
\end{itemize}
 Step 4 makes $$O(T_{MM}(n)\log^2(\frac{nBk_F}{\delta k_{\text{gap}}}))$$ arithmetic operations on a floating point machine with $$O(\log^4(\frac{nBk_F}{\delta k_{\text{gap}}})\log(n))$$ bits of precision.
\end{theorem}
Toward the proof of this theorem, we first give a lemma showing how the output of the algorithm EIG-FWD changes when there is some error on the input, and how the input parameters for the algorithm need to be modified accordingly to accommodate that error. 
\begin{lemma}\label{lem:diaginmainalglemma}
Let $A \in M_n(\C)$ be a diagonalisable matrix such that $||A|| \leq~M$ where $M > 1$, $\kappa^F_V(A) < B$ and $\text{gap}(A) \geq~k_{\text{gap}}$ for some $k_{\text{gap}} < 1$. Let $w_1,...,w_n$ be the normalized eigenvectors of $A$. Let $D \in M_n(\C)$ be such that $||D - A|| \leq \varepsilon_3 < \frac{k_{\text{gap}}}{4B}$.  Let $v^{(0)}_1,...,v^{(0)}_n$ be the output of $\text{EIG-FWD}$ on $(D, \delta, \frac{3nB}{k_{\text{gap}}}, 2M)$ where the accuracy parameter $\delta \in (0, \frac{k_{\text{gap}}}{24nB})$ and let $V^{(0)}$ be the matrix with columns $v_i^{(0)}$. Then 
\begin{itemize}
    \item $||w_i - v^{(0)}_i|| < \frac{3nB\varepsilon_3}{k_{\text{gap}}} + \delta$ for all $i \in [n]$ up to permutation and multiplication by phases.
    \item $\kappa_F(V^{(0)}) \leq 800n^6B^2.$
\end{itemize}
with probability at least $1- \frac{1}{n} - \frac{12}{n^2}$. This algorithm makes $$O(T_{MM}(n)\log^2(\frac{nBM}{\delta k_{\text{gap}}}))$$ arithmetic operations on a floating point machine with $O(\log^4(\frac{nBM}{\delta k_{\text{gap}}})\log(n))$ bits of precision.
\end{lemma}
\begin{proof}
We first want to show that $\text{EIG-FWD}$ can be run on $(D, \delta, \frac{3nB}{k_{\text{gap}}}, 2M)$. 
\par
Since $\text{gap}(A) \geq k_{\text{gap}} > 0$, $A$ has distinct eigenvalues. Since $$\kappa_{\text{eig}}(A) = \frac{\kappa_V(A)}{\text{gap}(A)} \leq \frac{\kappa^F_V(A)}{2\text{gap}(A)} <  \frac{B}{2k_{\text{gap}}},$$ we get $\varepsilon_3  < \frac{k_{\text{gap}}}{4B} < \frac{1}{8\kappa_{\text{eig}}(A)}$. By Lemma \ref{lem:bgvkslemma}, $D$ is diagonalisable and has distinct eigenvalues. 
\par
Now, we want to show that $\kappa_{\text{eig}}(D) < \frac{3nB}{k_{\text{gap}}}$.
Since $\varepsilon_3 < \frac{1}{8\kappa_{\text{eig}}(A)}$, we can apply Lemma \ref{lem:AA'relation} to $(A, A') = (A,D)$. This gives us that 
\begin{equation}\label{eq:kappavD}
    \kappa_V(D) \leq \frac{\kappa^F_V(D)}{2} \leq  \frac{3n\kappa^F_V(A)}{2} \leq \frac{3nB}{2}. 
\end{equation}
\par
Also, we know that $\text{gap}(A) \geq k_{\text{gap}}$. Let $\lambda_1,...,\lambda_n$ be the eigenvalues of $A$. Since the columns $w_1,...,w_n$ of $W$ are the true normalized eigenvectors of $A$, $(w_1,\lambda_1),...,(w_n,\lambda_n)$ are the eigenpairs of $A$. Let $(v'_1,\lambda'_1),...,(v'_n,\lambda'_n)$ be the eigenpairs of $D$. Using Corollary \ref{corr:prop1.1new} for ($A$, $A') = (A,D)$ and using the fact that $\varepsilon_3 < \frac{1}{8\kappa_{\text{eig}}(A)}$, we get that there exist phases $\rho_i$ such that
\begin{equation}\label{eq:v_iv_i'}
    ||w_i - \rho_iv'_i|| \leq \frac{3nB\varepsilon_3}{k_{\text{gap}}} \text{ and } |\lambda_i - \lambda'_i| \leq \kappa_V(A)\varepsilon_3 < \frac{B\varepsilon_3}{2}.
\end{equation}
Therefore, we have that for all $i \neq j \in [n]$,
\begin{align*}
     k_{\text{gap}} &\leq |\lambda_i - \lambda_j| \\
    &\leq |\lambda'_i - \lambda'_j| + |\lambda_i - \lambda'_i| + |\lambda_j - \lambda'_j| \\
    &\leq |\lambda'_i - \lambda'_j| + B\varepsilon_3
\end{align*}
This gives us that 
\begin{equation}\label{eq:gapD}
    \text{gap}(D) = \min_{i\neq j \in [n]}|\lambda'_i- \lambda'_j| \geq k_{\text{gap}} - B\varepsilon_3.
\end{equation}
Since $\varepsilon_3  < \frac{k_{\text{gap}}}{4B}$, we get that $k_{\text{gap}} - B\varepsilon_3 \geq \frac{k_{\text{gap}}}{2}$.
Putting this back in (\ref{eq:gapD}) and using (\ref{eq:kappavD}) shows that
\begin{equation}\label{eq:keigbound}
    \kappa_{\text{eig}}(D) \leq \frac{3nB}{2(k_{\text{gap}} - B\varepsilon_3)} \leq \frac{3nB}{k_{\text{gap}}}.
\end{equation}
Now we want to give an upper bound for the norm of $D$. By the triangle inequality,  $||D||_F \leq M + \varepsilon_3.$ Since $\varepsilon_3 < \frac{k_{\text{gap}}}{4B} < 1$ and $M > 1$, we have
\begin{equation}\label{eq:normDbound}
   ||D||_F \leq 2M.
\end{equation}
Since $\delta < \frac{k_{\text{gap}}}{24nB}$, using (\ref{eq:keigbound}), $\delta < \frac{1}{8 \times \frac{3nB}{k_{\text{gap}}}} \leq  \frac{1}{8\kappa_{\text{eig}}(D)}$.
Now, we can apply Theorem \ref{thm:eigfwd} to $D$ with the respective bounds for $\kappa_{\text{eig}}(D)$ and $||D||_F$ from~(\ref{eq:keigbound}) and (\ref{eq:normDbound}). We deduce that there exist phases $\rho'_i$ such that
\begin{align*}
    ||\rho'_iv'_i - v^{(0)}_i|| \leq \delta.
\end{align*}
This gives us that $||\rho_iv'_i - \rho_i(\rho'_i)^{-1}v^{(0)}_i|| \leq \delta$.
Combining this with (\ref{eq:v_iv_i'}), we see that $||\alpha_i w_i - v^{(0)}_i|| \leq \frac{3nB\varepsilon_3}{k_{\text{gap}}} + \delta$ where $\alpha_i = \rho'_i(\rho_i)^{-1}$. 
\par
Using Theorem \ref{thm:eigfwd}, 
(\ref{eq:keigbound}) and (\ref{eq:normDbound}) for bounding $K_{\text{eig}}$ and $K_{\text{norm}}$ respectively, we can conclude that that the number of bits of precision required by the floating point machine for this step is $O(\log^4(\frac{nBM}{\delta k_{\text{gap}}})\log(n))$.
\par
We now prove the second part of the lemma. From (\ref{eq:kappavD}) we have that $\kappa^F_V(D) \leq 3nB $. Putting this back in Theorem \ref{thm:eigfwd}, we can conclude that  $$\kappa_F(V^{(0)}) \leq \frac{9n}{4}+ 81n^4(9n^2B^2) \leq 800n^6B^2.$$
\end{proof}
Now we apply this lemma to the current setting to get error bounds for Step 4 of Algorithm \ref{algo:Jennrich}. 
\begin{proof}[Proof of Theorem \ref{thm:errstep4}]
From (\ref{eq:normtatb}) we get that $||T^{(b)}||_F \leq \sqrt{nB^3}$. 
Moreover, 
$||(T^{(a)})^{-1}||_F \leq \sqrt{\kappa_F(T^{(a)})} \leq \sqrt{k_F}$. 
This gives us that $ ||(T^{(a)})^{-1}T^{(b)}||_F \leq \sqrt{k_F}\sqrt{nB^3} = B^{\frac{3}{2}}\sqrt{nk_F}$. Since $k_F > 1$, we also have 
$B^{\frac{3}{2}}\sqrt{nk_F} > 1$.
\par
From the proof of Theorem \ref{thm:eigenvalues}, we know that the columns of $U^{-1}$ form the eigenvectors of $(T^{(a)})^{-1}T^{(b)}$ as well. Hence, $\kappa^F_V((T^{(a)})^{-1}T^{(b)}) \leq \kappa_F(U^{-1}) = \kappa(T) < B$. Now, we can use Lemma \ref{lem:diaginmainalglemma} with $A = (T^{(a)})^{-1}T^{(b)}$. Using the lemma for $M = B^{\frac{3}{2}}\sqrt{nk_F}$ and $\delta < \frac{k_{\text{gap}}}{24nB}$, if $p_1,...,p_n$ are the normalized eigenvectors of $(T^{(a)})^{-1}T^{(b)}$ 
we have $$||p_i - v_i^{(0)}|| \leq \delta'$$  where $\delta'= \frac{3nB\varepsilon_3}{k_{\text{gap}}} + \delta$ up to permutation and multiplication by phases.
\par
For the second part, applying Lemma \ref{lem:diaginmainalglemma} (ii) for $V = V^{(0)}$
shows that $k_V := \kappa_F(V^{(0)}) \leq 800n^6B^2.$
\end{proof}
 
\subsection{Estimating error at the end of Step 5}
In this section we estimate the error at the end of Step 5 of Algorithm~\ref{algo:Jennrich}. 

\textbf{Correspondence to exact arithmetic:} If the input tensor $T$ can be diagonalised as $T = \sum_{i=1}^n u_i^{\otimes 3}$, then, as seen in Algorithm \ref{algo:completeexact},  step 5 will ideally get as input the normalized eigenvectors $p_1,...,p_n$ of $(T^{(a)})^{-1}T^{(b)}$. Let $P$ be the matrix with columns $p_1,...,p_n$. The output of this step ideally is $P^{-1}$.
\par
\textbf{Finite arithmetic:} Since this algorithm is being run in finite arithmetic, we  assume  at this stage, 
that the columns $v_i^{(0)}$ of the input matrix $V^{(0)}$ are "close" to $p_i$. We cannot invert the input matrix exactly in finite arithmetic. Instead, we use the stable matrix inversion algorithm from Theorem \ref{thm:fastlinearalgebra}. In this section, we show that the output $C$ of this step is close to the ideal output $P^{-1}$.
The following is the main theorem of this section:
\begin{theorem}\label{thm:errstep5}
Let $T$ be an order-$3$ diagonalisable symmetric tensor given as input to Algorithm \ref{algo:Jennrich}. Let $U \in \text{GL}_n(\C)$ be such that $U$ diagonalises $T$ where $\kappa(T) = \kappa_F(U) < B$ and let $u_1,...,u_n$ be the rows of $U$. Let $T^{(a)} = \sum_{i=1}^n a_iT_i$ and $T^{(b)} = \sum_{i=1}^n b_iT_i$ be two linear combination of the slices $T_1,...,T_n$ of $T$ such that $\kappa_F(T^{(a)}) < k_F$ and $\text{gap}((T^{(a)})^{-1}T^{(b)}) \geq k_{\text{gap}}$.
Let $p_1,...,p_n$ be the true normalized eigenvectors of $(T^{(a)})^{-1}T^{(b)}$ and let $P$ be the matrix with columns $p_i$. Let $v_1,...,v_n$ be the rows of $P^{-1}$. Let $v^{(0)}_1,...,v^{(0)}_n$ be the output of Step 4 of Algorithm \ref{algo:Jennrich} run with accuracy parameter $\delta < \min\{\frac{k_{\text{gap}}}{24nB},\frac{1}{120n^{\frac{7}{2}}B}\}$ where $||p_i - v^{(0)}_i|| \leq \delta $.  Let $V^{(0)}$ be the matrix with columns $v^{(0)}_i$. Let $C$ be the output of Step 5 with rows $u'_1,...,u'_n$. Then,
\begin{itemize}
    \item $||v_i - u'_i|| := \varepsilon_5 \leq \tau_{INV} + \frac{4\delta k_V\sqrt{n}}{3}.$
    \item $||u'_i|| \leq \tau_{INV} + \sqrt{k_V}$.
\end{itemize}
where $k_V := \kappa_F(V^{(0)})$ and $||C-(V^{(0)})^{-1}|| \leq~\tau_{INV} $  is the error of matrix inversion from Theorem~\ref{thm:fastlinearalgebra}.
\end{theorem}

To prove this, we first give a couple of intermediate lemmas. The first lemma shows  how an upper bound on the norm of the columns of a matrix can be translated to upper bound the norm of a matrix. 
\begin{lemma}\label{lem:coltomatdiff}
Let $A$ be an $n \times n$ matrix where  $(a_i)_{i \in [n]}$ are the columns of $A$ and $||a_i|| \leq k$ for all $i \in [n]$. Then $||A|| \leq k\sqrt{n}$.
\end{lemma}
\begin{proof}
We use the fact that $||A||^2 \leq ||A||_F^2 = \sum_{i \in [n]} ||a_i||^2 \leq nk^2$. 
\end{proof}

Let $T$ be the input 
tensor and let $U$ be the matrix that diagonalizes it. In Step 5 of the algorithm, we want to compute the inverse of~$V$ with columns $v_1,...,v_n$ as output in Step 4. We first show that $V^{-1}$ is indeed close to $U$, albeit with scaled columns.  
\begin{corollary}\label{corr:inverses}
Let $T$ be an order-$3$ diagonalisable symmetric tensor. Let $U \in \text{GL}_n(\C)$ be such that $U$ diagonalises $T$. Let $V^{(0)}$ be the matrix with $\kappa_F(V) < k_V$ and columns $v^{(0)}_1,...,v^{(0)}_n$ output at the end of Step 4 of Algorithm \ref{algo:Jennrich}. Define $k_V := \kappa_F(V^{(0)})$. Let us assume that $||p_i - v^{(0)}_i|| \leq \delta < \frac{1}{4\sqrt{nk_V}}$ for all $i \in [n]$. Let $P$ be the matrix with $p_1,...,p_n$ as columns. Then $$||P^{-1} - (V^{(0)})^{-1}|| \leq \frac{\delta\sqrt{k_Vn}}{\frac{1}{\sqrt{k_V}} - \delta\sqrt{n}}.$$ 
\end{corollary}
\begin{proof}
By Lemma \ref{lem:coltomatdiff}, 
we also have  $||P - V^{(0)}|| \leq \delta \sqrt{n} =: M$. So $M\sqrt{k_V} < \frac{1}{4}$. Now, applying Lemma \ref{lem:nearconditionnumber} to $A' = P$, $A = V^{(0)}$, $K = k_V$ and $M = \delta\sqrt{n}$, we obtain
\begin{align*}
    ||P^{-1} - (V^{(0)})^{-1}|| \leq \frac{\delta k_V\sqrt{n}}{1 - \delta\sqrt{nk_V}}  = \frac{\delta\sqrt{k_Vn}}{\frac{1}{\sqrt{k_V}} - \delta\sqrt{n}}.
\end{align*}
\end{proof}

\begin{proof}[Proof of Theorem \ref{thm:errstep5}]
Since $\delta < \frac{1}{120n^{\frac{7}{2}}B}$,  from Theorem \ref{thm:errstep4} (ii) we have  $k_V < 800n^6B^2$. Combining these, we get that $\delta < \frac{1}{4\sqrt{nk_V}}$. Now, we can use Corollary \ref{corr:inverses} to deduce that $$||P^{-1} - (V^{(0)})^{-1}|| \leq \frac{\delta\sqrt{k_Vn}}{\frac{1}{\sqrt{k_V}} - \delta\sqrt{n}} < \frac{4\delta k_V\sqrt{n}}{3}.$$ Combining this with the hypothesis that $||C - (V^{(0)})^{-1}|| \leq~\tau_{INV}$,
by the triangle we also have $$||P^{-1} - C||\leq \tau_{INV} + \frac{4\delta k_V\sqrt{n}}{3}.$$ Since $||A|| = ||A^T||$, we have
\begin{align*}
   ||v_i - u'_i|| = ||P^{-T}e_i - C^Te_i|| \leq ||P^{-1} - C|| \leq \tau_{INV} + \frac{4\delta k_V\sqrt{n}}{3}.
\end{align*}
\par
For the second part, we get that $||u'_i|| \leq ||C|| \leq ||C - (V^{(0)})^{-1}|| + ||(V^{(0)})^{-1}|| \leq \tau_{INV} + \sqrt{k_V}$. 
\end{proof}

\subsection{Error accumulated in Step 6}
In this section we  analyse the error at the end of Step 6 of Algorithm~\ref{algo:Jennrich}.
\par 
\textbf{Correspondence to exact arithmetic:} Following Algorithm \ref{algo:completeexact} in exact arithmetic, this step ideally takes as input a matrix $P$ whose columns are the true eigenvectors of $(T^{(a)})^{-1}T^{(b)}$. Let $S = (P \otimes P \otimes P).T$ and let $S_1,...,S_n$ be the slices of $S$. The goal is to compute $Tr(S_i)$ for all $i \in [n]$. This would output exact scaling factors $\alpha_i \in \C$ as in Step 6 of Algorithm \ref{algo:completeexact}. For a more formal proof of this fact, refer to Theorem~\ref{thm:complexteexactmain}.
\par
\textbf{Finite arithmetic:} Here, we assume that at the end of Step 4, we have vectors $v^{(0)}_i$ close to the normalized eigenvectors of $(T^{(a)})^{-1}T^{(b)}$. Then in the next lemma, we show that  Step 6 outputs scalars $\alpha'_1,...,\alpha'_n$ close to the scaling factors in the ideal situation. 
The following is the main theorem of this section:
\begin{theorem}\label{thm:errstep6}
Let $T$ be an order-$3$ diagonalisable symmetric tensor  given as input to Algorithm \ref{algo:Jennrich}. Let $U \in \text{GL}_n(\C)$ be such that $U$ diagonalises $T$ where $\kappa(T) = \kappa_F(U) < B$. Let $T^{(a)}, T^{(b)}$ be two linear combination of the slices $T_1,...,T_n$ of $T$ such that $\kappa_F(T^{(a)}) < k_F$ where $k_F > 1$ and $\text{gap}((T^{(a)})^{-1}T^{(b)}) > k_{\text{gap}}$. Let $p_1,...,p_n$ be the true normalized eigenvectors of $(T^{(a)})^{-1}T^{(b)}$. Define $S = (P \otimes P \otimes P).T$ and let $S_1,...,S_n$ be the slices of $T$. Let $\alpha_i = Tr(S_i)$ for all $i \in [n]$.
\par
We will assume that at the end of Step 4, the algorithm outputs $v^{(0)}_1,...,v^{(0)}_n \in \C^n$ with the property that $$||p_i - v^{(0)}_i|| \leq \delta < 1,$$ up to permutation and multiplication by phases.

Let $\Tilde{S} = (V^{(0)} \otimes V^{(0)} \otimes V^{(0)}).T$ and let $\Tilde{S}_1,...,\Tilde{S}_n$ be the slices of $\Tilde{S}$. Define $\Tilde{\alpha_i} := \text{Tr}(\Tilde{S}_i)$. Let $\alpha'_1,...,\alpha'_n$ be the output of $\text{CB}(T,V^{(0)})$ where $\text{CB}$ is the change of basis algorithm 
from Theorem \ref{thm:fastcob}. Then for all $i \in [n]$, 
$$|\alpha_i - \alpha'_i| \leq 7\delta k_V n^3 B^{\frac{3}{2}} + \gamma_{CB}$$
where $|\Tilde{\alpha_i} - \alpha_i'| \leq \gamma_{CB}$ and $\gamma_{CB}$ is an upper bound for the error in the $CB$ algorithm on inputs $T,V^{(0)}$.
\end{theorem}

The proof goes through the following idea: If two vectors are "close", then the traces of the slices of their tensor powers are close as well.
\begin{proof}
Since $||p_i - v^{(0)}_i|| \leq \delta$, $||P - V^{(0)}|| \leq \delta \sqrt{n}$ by Lemma \ref{lem:coltomatdiff}.  Multiplying on the left by $||U||$ and using the fact that $||A|| = ||A^T||$ for all matrices $A$, this gives us that
\begin{equation}\label{eq:diffatbt}
\begin{split}
    ||(UP)^T e_i - (UV^{(0)})^T e_i|| \leq ||(UP)^T - (UV^{(0)})^T || &= ||UP - UV^{(0)}|| \\
    &\leq ||U||||P - V^{(0)}|| \\
    &\leq \delta \sqrt{n} ||U||_F < \delta \sqrt{nB}.    
\end{split}
\end{equation}
Let $a_i = (UP)^Te_i$ and $b_i = (UV^{(0)})^Te_i$. From (\ref{eq:diffatbt}), we have  $||b_i|| \leq \delta \sqrt{nB} + ||UP||$. Since, by hypothesis, the columns of $P$ are normalized, we have $||UP|| \leq ||U||_F||P||_F \leq \sqrt{nB}$.  Since $\delta < 1$, we have  for all $i \in [n]$
\begin{equation}\label{eq:normanormb}
    ||a_i|| \leq \sqrt{nB} \text{ and } ||b_i|| \leq 2\sqrt{nB}.
\end{equation}
Since $T = \sum_{t=1}^n (U^Te_t)^{\otimes 3}$, we have $\Tilde{S} = \sum_{t=1}^n ((UV^{(0)})^Te_t)^{\otimes 3}$. Hence, $$\Tilde{\alpha_i} = Tr(\Tilde{S}_i) = \sum_{j=1}^n (\Tilde{S}_i)_{j,j} = \sum_{t,j=1}^n (((UV^{(0)})^Te_t)^{\otimes 3})_{i,j,j} = \sum_{j,t=1}^n\Big(b_t^{\otimes 3}\Big)_{i,j,j}.$$
Now, using the fact that $\alpha_i = \sum_{j,t=1}^n\Big(((UP)^Te_t)^{\otimes 3}\Big)_{i,j,j} =  \sum_{j,t=1}^n\Big(a_t^{\otimes 3}\Big)_{i,j,j}$, we have 
\begin{align*}
    |\alpha_i - \Tilde{\alpha_i}| &= \Big|\sum_{j,t=1}^n\Big(a_t^{\otimes 3}\Big)_{i,j,j} - \sum_{j,t=1}^n\Big(b_t^{\otimes 3}\Big)_{i,j,j}\Big|    \\
    &= \Big|\sum_{j,t=1}^n (a_t)_i(a_t)^2_j - \sum_{j,t=1}^n(b_t)_i(b_t)^2_j\Big| \\
    &\leq \sum_{j,t = 1}^n \Big(| (a_t)_i(a_t)^2_j - (b_t)_i(a_t)^2_j| + | (b_t)_i(a_t)^2_j - (b_t)_i(b_t)^2_j | \Big) \\
    &\leq \sum_{t=1}^n \Big(|(a_t)_i - (b_t)_i| ||a_t|| +  ||b_t||( ||a_t||+ ||b_t||)\sum_{j=1}^n |(a_t)_j - (b_t)_j|\Big).
\end{align*}
By the Cauchy-Schwarz inequality,  it follows that 
\begin{equation}\label{eq:alpha_itildealphai}
    |\alpha_i - \Tilde{\alpha_i}| \leq  \sum_{t=1}^n \Big(|(a_t)_i - (b_t)_i| ||a_t|| +  ||b_t||( ||a_t||+ ||b_t||)\sqrt{n}||a_t - b_t||\Big).
\end{equation}
Putting this back in (\ref{eq:diffatbt}) and (\ref{eq:normanormb}), we have 
\begin{equation}\label{eq:kialphai}
    |\alpha_i - \Tilde{\alpha_i}| \leq \sum_{t=1}^n\Big( \delta \sqrt{nB}\sqrt{nB} + 2\sqrt{nB}(3\sqrt{nB})\sqrt{n}\delta \sqrt{nB}\Big) \leq 7\delta n^3 B^{\frac{3}{2}}.
\end{equation}
Now, applying the triangle inequality to this equation along with the hypothesis  $|\Tilde{\alpha_i} - \alpha'_i| \leq \gamma_{CB}$ shows that
\begin{align*}
    |\alpha_i - \alpha'_i| \leq 7\delta k_V n^3B^{\frac{3}{2}} + \gamma_{CB}.
\end{align*}
\end{proof}

\subsection{Estimating the error at the end of Step 7}
In this section, we assume that the input to Step 4 of Algorithm \ref{algo:Jennrich} has error~$\varepsilon_3$ and $\delta$ is the desired accuracy parameter for the EIG-FWD algorithm in Step~4 of the algorithm. Let us assume that the input tensor $T$ can be written as $T = \sum_{i=1}^n (u_i)^{\otimes 3}$. Let $u'_1,...,u'_n$ be the output of Step 5 and $\alpha'_1,...,\alpha'_n$ be the output of Step 6. If the algorithm was run in exact arithmetic, the output at the end of Step 7 would be $(\alpha'_i)^{\frac{1}{3}}u'_i$. In turn, this vector  would be equal to $u_i$ up to multiplication by phases. Here, we want to show that even in finite arithmetic, the output $l_i$ is close to $u_i$ up to multiplication by phases.
\begin{theorem}\label{thm:alg2bound}
Let $T$ be an order-$3$ diagonalisable symmetric tensor 
given as input to Algorithm \ref{algo:Jennrich}. Let $U \in \text{GL}_n(\C)$ be such that $U$ diagonalises $T$ and $\kappa(T) = \kappa_F(U) < B$ and let $u_1,...,u_n$ be the rows of $U$. Let $T^{(a)}$,$T^{(b)}$ be two linear combinations of the slices of $T$ such that $\kappa_F(T^{(a)}) < k_F$ and $\text{gap}((T^{(a)})^{-1}T^{(b)}) \geq k_{\text{gap}}$ where $k_{\text{gap}} < 1$. Let $D \in M_n(\C)$ where $||D - (T^{(a)})^{-1}T^{(b)}|| \leq \varepsilon_3 < \frac{k_{\text{gap}}}{4B}$ be an input to Step~4 of the algorithm and let $\delta$ be the accuracy parameter for Step~4 such that $\frac{3nB\varepsilon_3}{k_{\text{gap}}} < \delta <~\min\{ \frac{k_{\text{gap}}}{48nB}, \frac{1}{240n^{\frac{7}{2}}B}\}$. Let $v^{(0)}_1,..., v^{(0)}_n$ be the output of Step 4 and let $V^{(0)}$ be the matrix with columns $v_i^{(0)}$. Let $l_i$ be the output of Step 7 of the Algorithm. Then, with probability at least $1 - \frac{1}{n} - \frac{12}{n^2}$, there exist cube roots of unity $\omega_i$ such that
\begin{align*}
    &||\omega_iu_i - l_i|| \\
    \leq &3(14\delta n^3 B^{\frac{3}{2}} + \gamma_{CB})^{\frac{1}{3}} (\tau_{\text{INV}} + \sqrt{k_V}) + \sqrt{nB}(2\tau_{\text{INV}} + 5\textbf{u}n^{\frac{3}{2}}\sqrt{k_V} + \frac{8\delta k_V\sqrt{n}}{3}).
\end{align*}
 where $\kappa_F(V^{(0)}) = k_V < 800n^6B^2$, $||C- (V^{(0)})^{-1}|| \leq \tau_{INV}$ is the error of matrix inversion as mentioned in Definition \ref{def:INValg} and $\gamma_{CB}$ is the error on change of basis as mentioned in Theorem \ref{thm:fastcob}. 
\par
The algorithm runs on a floating point machine with $$ \log(\frac{1}{\textbf{u}}) = O(\log^4(\frac{nBk_F}{\delta k_{\text{gap}}})\log(n))$$ bits of precision.
\end{theorem}

\begin{lemma}\label{lem:normalphai}
Let $T$ be an order-$3$ diagonalisable symmetric tensor and let $U \in \text{GL}_n(\C)$ be such that $U$ diagonalises $T$. Let $T^{(a)}, T^{(b)}$ be two linear combination of the slices $T_1,...,T_n$ of $T$ such that $T^{(a)}$ is invertible and the eigenvalues of $(T^{(a)})^{-1}T^{(b)}$ are distinct. Let $p_1,...,p_n$ be the true normalized eigenvectors of $(T^{(a)})^{-1}T^{(b)}$ and let $P$ be the matrix with columns $p_1,...,p_n$. Define $S = (P \otimes P \otimes P).T$ and let $S_1,...,S_n$  be the slices of $T$. Let $\alpha_i = Tr(S_i)$ for all $i \in [n]$. Then $|(\alpha_i)^{\frac{1}{3}}| = |(UP)^Te_i|$.
\end{lemma}
\begin{proof}
From (\ref{eq:eigenvectormaineq}),  there exist scalars $k_1,...,k_n \in \C$ such that $P = U^{-1}D$ where $D = \text{diag}(k_1,...,k_n)$. Hence $(UP)^Te_i = D^Te_i = k_ie_i$. Hence from~(\ref{eq:alphaiki3}), we have 
\begin{align*}
    |(UP)^Te_i| = |k_i| = |(\alpha_i)^{\frac{1}{3}}|.
\end{align*}
\end{proof}
\begin{proof}[Proof of Theorem \ref{thm:alg2bound}]
By Theorem \ref{thm:errstep4},  if a matrix $D$ and an accuracy parameter $\delta$ are given as input to Step 4 of Algorithm \ref{algo:Jennrich} satisfying the conditions mentioned in the hypothesis, then Step 4 outputs $v^{(0)}_1,..., v^{(0)}_n$ such that $||p_i - v_i^{(0)}|| \leq \frac{3nB\varepsilon_3}{k_{\text{gap}}} + \delta < 2\delta$. Then we can use Lemma \ref{lem:normalphai} and (\ref{eq:normanormb}) to deduce that 
\begin{equation}\label{eq:normkistep7}
        |(\alpha_i)^{\frac{1}{3}}| \leq \sqrt{nB}
\end{equation}
Now if $u_i'$ is the output of Step 5 of Algorithm \ref{algo:Jennrich}, using Theorem \ref{thm:errstep5} for $2\delta$, we have $||v_i - u_i'|| \leq \tau_{INV} + \frac{8\delta k_V\sqrt{n}}{3}$. Multiplying  both sides by $|(\alpha_i)^{\frac{1}{3}}|$ and using (\ref{eq:normkistep7}), we obtain
\begin{equation}\label{eq:uiui'}
\begin{split}
    ||\alpha_i^{\frac{1}{3}}v_i - \alpha_i^{\frac{1}{3}}u'_i|| &\leq (\tau_{INV} + \frac{8\delta k_V\sqrt{n}}{3})|(\alpha_i)^{\frac{1}{3}}| \\
    &\leq (\tau_{INV} + \frac{8\delta k_V\sqrt{n}}{3})\sqrt{nB}.
\end{split}
\end{equation}
Let $\alpha'_1,...,\alpha'_n$ be the output of $\text{CB}(T,V^{(0)})_{i,i,i}$. Using Theorem \ref{thm:errstep6} for $2\delta$, we have  $$|\alpha_i - \alpha'_i| \leq 14\delta n^3 B^{\frac{3}{2}} + \gamma_{CB}.$$ Hence there exists a cube root of $\alpha'_i$, denoted by $(\alpha'_i)^{\frac{1}{3}}$, 
such that 
\begin{equation}\label{eq:kialpha'i1}
    ||(\alpha_i)^{\frac{1}{3}} - (\alpha'_i)^{\frac{1}{3}}|| \leq (14\delta n^3 B^{\frac{3}{2}} + \gamma_{CB})^{\frac{1}{3}}.
\end{equation}
Let us denote by $\text{fl}((\alpha'_i)^{\frac{1}{3}})$ the output when a cube root of $\alpha'_i$ is computed on a floating point machine with precision $\textbf{u}$. Then from (\ref{eq:floatingpointarithmetic}), there exists cube roots of unity $\omega'_i$ such that
\begin{equation}\label{eq:kialpha'i2}
    || \omega'_i(\alpha'_i)^{\frac{1}{3}} - \text{fl}((\alpha'_i)^{\frac{1}{3}})|| \leq |\Delta||(\alpha'_i)^{\frac{1}{3}}|
\end{equation}
where $|\Delta| \leq \textbf{u}$. By (\ref{eq:normkistep7}), (\ref{eq:kialpha'i1}) and the triangle inequality,
\begin{equation}\label{eq:normalpha'i}
    ||(\alpha'_i)^{\frac{1}{3}}|| \leq \Big((14\delta n^3 B^{\frac{3}{2}} + \gamma_{CB})^{\frac{1}{3}} + \sqrt{nB}\Big).
\end{equation}
Let us denote  $p_{\text{err}} :=  (14\delta n^3 B^{\frac{3}{2}} + \gamma_{CB})^{\frac{1}{3}} $. Then, putting (\ref{eq:normalpha'i}) back in (\ref{eq:kialpha'i2}) and combining that with (\ref{eq:kialpha'i1}) using the triangle inequality,  there exist cube roots of unity $\omega''_i$ such that
\begin{equation}\label{eq:kiflalpha'i}
\begin{split}
    ||(\alpha_i)^{\frac{1}{3}} - \omega''_i\text{fl}((\alpha'_i)^{\frac{1}{3}})|| &\leq p_{\text{err}} + \textbf{u}(p_{\text{err}} + \sqrt{nB}) \\
    &\leq 2p_{\text{err}} + \textbf{u}.\sqrt{nB} .
\end{split}
\end{equation}
The last inequality relies on the facts that $\textbf{u} < 1$, $\delta  < \frac{1}{2}$ and $k_V > 1$.
Using these facts along with (\ref{eq:kiflalpha'i}) and (\ref{eq:normkistep7}), we also obtain
\begin{equation}\label{eq:normflalpha'_i}
    ||\text{fl}((\alpha'_i)^{\frac{1}{3}})|| \leq 2(p_{\text{err}} + \sqrt{nB}).
\end{equation}
Moreover, we know from Theorem \ref{thm:errstep5} that 
\begin{equation}\label{eq:u'_i2}
    ||u'_i|| \leq \tau_{INV} + \sqrt{k_V}.
\end{equation}
Using this along with (\ref{eq:kiflalpha'i}),
we have 
\begin{align*}
    || \alpha_i^{\frac{1}{3}}u'_i - \omega''_i\text{fl}((\alpha'_i)^{\frac{1}{3}})u'_i|| &\leq 2(p_{\text{err}} + \textbf{u}.\sqrt{Bk_V})||u'_i|| \\
    &\leq 
(2p_{\text{err}} + \textbf{u}.\sqrt{nB})(\tau_{INV} + \sqrt{k_V})
\end{align*}
Combining this with (\ref{eq:uiui'}), by the triangle inequality
\begin{equation}\label{eq:joiningeq}
\begin{split}
    &||\omega_iu_i - \text{fl}((\alpha'_i)^{\frac{1}{3}})u'_i|| \\
    &\leq (2p_{\text{err}} + \textbf{u}.\sqrt{nB})(\tau_{INV} + \sqrt{k_V}) + \sqrt{nB}(\tau_{INV} + 
    \frac{8\delta k_V\sqrt{n}}{3})
\end{split}
\end{equation}
where $\omega_i = (\omega''_i)^{-1}$.
Let $l_i$ be the output of $\text{fl}((\alpha'_i)^{\frac{1}{3}}) u'_i$ computed on a floating point machine with precision $\textbf{u}$. By~(\ref{eq:ipbound}), 
\begin{align*}
    ||l_i - \text{fl}((\alpha'_i)^{\frac{1}{3}})u'_i|| \leq 2n^{\frac{3}{2}}\textbf{u}|\text{fl}((\alpha'_i)^{\frac{1}{3}})| ||u'_i||.
\end{align*}
Using (\ref{eq:normflalpha'_i}) and (\ref{eq:u'_i2}), we also have 
\begin{equation}
\begin{split}
     ||l_i - \text{fl}((\alpha'_i)^{\frac{1}{3}}) u'_i|| &\leq |\text{fl}((\alpha'_i)^{\frac{1}{3}})|||u'_i||.\textbf{u} \\
     &\leq 2 n^{\frac{3}{2}}\textbf{u} \cdot \Big(2p_{\text{err}} + 2 \sqrt{nB}\Big)\cdot\Big(\tau_{INV} +
     \sqrt{k_V}\Big).
\end{split}
\end{equation}
Combining this with (\ref{eq:joiningeq}),   we can finally conclude from the  triangle inequality that
$|\omega_iu_i - l_i||$ is upper bounded by
\begin{align*}
    &\leq (2p_{\text{err}} + \textbf{u}. \sqrt{nB})(\tau_{\text{INV}} + \sqrt{k_V}) + \sqrt{nB}(\tau_{\text{INV}} + \frac{8\delta k_V\sqrt{n}}{3}) \\
    &+ 2 n^{\frac{3}{2}}\textbf{u} \cdot \Big(2p_{\text{err}} + 2 \sqrt{nB}\Big)\cdot\Big(\tau_{INV} + \sqrt{k_V}\Big) \\
    &\leq 3p_{\text{err}}(\tau_{\text{INV}} + \sqrt{k_V}) + \sqrt{nB}(2\tau_{\text{INV}} + 5\textbf{u}n^{\frac{3}{2}}\sqrt{k_V} + \frac{8\delta k_V\sqrt{n}}{3}).
\end{align*}
The final inequality relies on the fact that $\textbf{u} \leq \frac{1}{5n^{\frac{3}{2}}}$.
\end{proof}

\subsection{Combining the errors}\label{sec:finish4-7proof}
If $B$ is the input estimate for the condition number of the given tensor and~$\varepsilon$ is the required accuracy parameter, recall that we have set the parameters in the following way:
\begin{equation}\label{eq:parameters}
\begin{split}
        k_{\text{gap}} := \frac{1}{C_{\text{gap}}n^6B^3} , k_{F} := c_Fn^5B^3  \text{ and } \delta := \frac{\varepsilon^3}{Cn^{17}B^{\frac{13}{2}}} 
\end{split}
\end{equation}
\text{ where } $C$ \text{ is a constant we will set in (\ref{eq:C})}.
In the following theorem, we show that if the input to Step 4 of the algorithm is not too far from the input in the exact arithmetic setting, the Algorithm indeed outputs an $\varepsilon$-approximate solution to the tensor decomposition problem. 
\begin{theorem}\label{thm:alg2boundfinalerror}
Let $T = \sum_{i=1}^n u_i^{\otimes 3}$ be the input to Algorithm \ref{algo:Jennrich}, 
where the $u_i$'s are linearly independent and $\kappa(T) \leq B$. Let $\varepsilon$ be the required accuracy parameter given as input.
\par
Let $T^{(a)}$, $T^{(b)}$ be two linear combinations of the slices of $T$ such that $T^{(a)}$ is invertible, $\kappa_F(T^{(a)}) < k_F$ and $\text{gap}((T^{(a)})^{-1}T^{(b)}) \geq k_{\text{gap}}$. Let $D \in~M_n(\C)$ where $||D - (T^{(a)})^{-1}T^{(b)}|| \leq \varepsilon_3 < (nB)^{C_3 \log n}.\textbf{u}$ for some appropriate constant $C_3$ be an input to Step~4 of the algorithm, and let $\delta$ be the accuracy parameter for Step~4 ($\delta, k_{\text{gap}}$ and $k_F$ are set as in (\ref{eq:parameters})). Let $l_i$ be the output of Step 7 of the algorithm. Then, with probability at least $1 - \frac{1}{n} - \frac{12}{n^2}$, there exist cube roots of unity $\omega_i$ such that
\begin{align*}
    ||\omega_iu_i - l_i|| \leq \varepsilon
\end{align*} 
up to permutation. The algorithm runs on a floating point machine with $$ \log(\frac{1}{\textbf{u}}) = O(\log^4(\frac{nB}{\varepsilon})\log(n))$$ bits of precision.
\end{theorem}
\begin{proof}
By hypothesis, the number of bits of precision is 
\begin{equation}\label{eq:ubounds}
    \log(\frac{1}{\textbf{u}}) > c \log^4(\frac{nB}{\varepsilon})\log(n)
\end{equation}
where we assume $c$ is a large enough constant.
We will apply Theorem \ref{thm:alg2bound} with appropriate bounds. First, we will bound different quantities which appear in the statement of the theorem. We will require these bounds later on in the proof.
\begin{enumerate}
    \item  $\mathbf{k_V}$: By Theorem \ref{thm:errstep4} (ii),  $k_V := \kappa_F(V^{(0)}) < 800n^6B^2$.  
    \item $\mathbf{\varepsilon_3}$: By (\ref{eq:mubounds}) and the fact that $k_{\text{gap}} := \frac{1}{c_{\text{gap}}n^6B^3}$, 
    we already have that
    $\frac{3nB\varepsilon_3}{k_{\text{gap}}} \leq 3c_{\text{gap}}(nB)^{7+C_3 \log n}\cdot \textbf{u}$. Let us now apply Lemma \ref{lem:boundsforu} to $\textbf{u}$. We see that for every constant $C > 0$, there exists a large enough constant~$c$ as mentioned in (\ref{eq:ubounds}) such that
    \begin{equation}\label{eq:epsilon3}
        \frac{3nB\varepsilon_3}{k_{\text{gap}}} < \frac{\varepsilon^3}{Cn^{17}B^{\frac{13}{2}}} =: \delta  
    \end{equation}
    We will later set $C$ in (\ref{eq:C}).
    Since $\varepsilon < 1$, this already gives us that $\delta \leq 1 \leq \sqrt{3k_V}$. By (\ref{eq:epsilon3}), 
    \begin{equation}\label{eq:epsilon32}
        \varepsilon_3 \leq \frac{1}{4c_{\text{gap}}n^6B^4} = \frac{k_{\text{gap}}}{4B} 
    \end{equation}    \item $\mathbf{\tau_{INV}}$: Applying Theorem \ref{thm:fastlinearalgebra} to $C = \text{INV}(V^{(0)})$, if $||C - (V^{(0)})^{-1}|| := \tau_{INV}$ then 
    $\tau_{INV} \leq n^{c_{\text{INV}} \log 10}\cdot \textbf{u} \cdot (k_V)^{8 \log n} ||V^{(0)}||$. 
    Taking into account the inequality $k_V  \leq 800n^6B^2$ from Part (1), we get that $\tau_{INV} \leq (nB)^{c_{\tau} \log n}\cdot~\textbf{u}$ where $c_{\tau}$ is some appropriate constant.
Applying Lemma \ref{lem:boundsforu} to $\textbf{u}$,  we see that
    \begin{equation}\label{eq:2tauinv}
        2\tau_{INV} \leq \frac{\varepsilon}{6\sqrt{nB}} < 1
    \end{equation}
    since $c$ is large enough according to (\ref{eq:ubounds}).
    \item $\mathbf{\gamma_{CB}}$: Now we want to bound the error for the application of the change of basis algorithm in Step 6. Let $\alpha_i = ((V^{(0)} \otimes V^{(0)} \otimes V^{(0)}).T)_{i,i,i}$ and let $\alpha'_1,...,\alpha'_n$  be the output of $CB(T,V^{(0)})$. We want to bound $\gamma_{CB} := |\alpha_i - \alpha'_i|$. We have  $||V^{(0)}||_F \leq \sqrt{\kappa_F(V^{(0)})} = \sqrt{k_V} \leq 20n^3B$. Also, from Lemma \ref{lem:normtensorbound}, we have  $||T||_F \leq B^{\frac{3}{2}}$. By Theorem \ref{thm:fastcob}, we already have $\gamma_{CB} \leq  14n^{\frac{3}{2}}\cdot \textbf{u} \cdot (k_VB)^{\frac{3}{2}}$. By part (1), we get that $k_V \leq 800n^6B^2$ and hence, $(k_VB)^{\frac{3}{2}} \leq 30^3n^9B^\frac{9}{2}$. Applying Lemma \ref{lem:boundsforu} to $\textbf{u}$,  
    \begin{equation}\label{eq:gammacb}
        \gamma_{CB} \leq  \frac{\varepsilon^3}{2 \times (12)^3 k_V^{\frac{3}{2}}}
    \end{equation}
since $c$ is large enough 
according to (\ref{eq:ubounds}). 
\end{enumerate}
Next, we show that Theorem \ref{thm:alg2bound} can be applied to this situation. The following are the necessary conditions for the theorem to be applied: 
\begin{itemize}
    \item $\varepsilon_3 < \frac{k_{\text{gap}}}{4B}$ : This is shown to be satisfied in (\ref{eq:epsilon32}).
    \item $\frac{3nB\varepsilon_3}{k_{\text{gap}}} < \delta$ : This is shown to be satisfied in (\ref{eq:epsilon3}).
    \item $\delta < \min\{ \frac{k_{\text{gap}}}{48nB}, \frac{1}{240n^{\frac{7}{2}}B}\}$ : Since $k_{\text{gap}} \leq \frac{1}{c_{\text{gap}}n^6B^3}$, we get that for large enough $n$, $\delta <  \frac{1}{48c_{\text{gap}}n^7B^4} = \min\{\frac{k_{\text{gap}}}{48nB}, \frac{1}{240n^{\frac{7}{2}}B}\}$.
\end{itemize}
 This 
 shows that the conditions of Theorem~\ref{thm:alg2bound} are indeed satisfied.
Now we want to show that the error in Theorem \ref{thm:alg2bound} is bounded by $\varepsilon$. More formally, let us define 
\begin{align*}
    E_1 := 3\Big(e_{11} + \gamma_{CB})^{\frac{1}{3}}(\tau_{INV} + \sqrt{k_V})
\end{align*}
where $e_{11} := 14\delta n^3 B^{\frac{3}{2}}$ and
\begin{align*}
    E_2 := \sqrt{nB}(2\tau_{INV} + 5\textbf{u}.\sqrt{k_Vn^3} + e_{21})
\end{align*}
where $e_{21} := \frac{8\delta k_V\sqrt{n}}{3}$. We want to show that $E_1 + E_2 < \varepsilon$.
\par
We first show that $E_1 \leq \frac{\varepsilon}{2}$. 
From (\ref{eq:2tauinv}), we get that $2\tau_{INV} < 1$ whereas $k_V > 1$ 
implies $\tau_{INV} < \sqrt{k_V}$.
Set
\begin{equation}\label{eq:C}
    C = 28 \times (12)^3 \times (30)^3.
\end{equation}
Setting  $\delta = \frac{\varepsilon^3}{C n^{12}B^{\frac{9}{2}}}$ gives us that $e_{11} \leq \frac{\varepsilon^3}{2 \times 12^3 (k_V)^{\frac{3}{2}}}$. From (\ref{eq:gammacb}) we have 
$\gamma_{CB} \leq \frac{\varepsilon^3}{2 \times 12^3 (k_V)^{\frac{3}{2}}}$.

This gives us that $(e_{11}+\gamma_{CB})^{\frac{1}{3}} \leq \frac{\varepsilon}{12\sqrt{k_V}}$ and 
$3(e_{11}+\gamma_{CB})^{\frac{1}{3}}(2\sqrt{k_V}) <~\frac{\varepsilon}{2}$. Hence, finally $E_1 \leq \frac{\varepsilon}{2}$. 
We also want to show that $E_2 \leq \frac{\varepsilon}{2}$.  Since $k_V \leq 800n^6B^2$, applying Lemma \ref{lem:boundsforu1} to $u$ as set in (\ref{eq:ubounds}),  there must exist a large enough $c$ such that $\frac{30n^2\sqrt{k_VB}}{\varepsilon} \leq \frac{900n^5B^{\frac{3}{2}}}{\varepsilon} \leq \frac{1}{\textbf{u}}$. 
Thus we have
\begin{equation}\label{eq:E21}
    5\textbf{u}.\sqrt{k_V n^3} \leq \frac{\varepsilon}{6\sqrt{nB}}.
\end{equation}
Now, we claim that $e_{21} < \frac{\varepsilon}{6\sqrt{nB}}$. Using $\delta =  \frac{\varepsilon^3}{C n^{12}B^{\frac{9}{2}}}$, we get that $6e_{21}\sqrt{nB} \leq \frac{16 \times 800 \times \varepsilon^3}{Cn^5 B^4}. 
$ Since $\frac{16 \times 800}{C} < 1$, $\varepsilon < 1$ and $n,B > 1$, we finally obtain
\begin{equation}\label{eq:E22}
    e_{21} < \frac{\varepsilon}{6\sqrt{nB}}.
\end{equation}
Also, by (\ref{eq:2tauinv}) we already have 
\begin{equation}
    2\tau_{INV} \leq \frac{\varepsilon}{6\sqrt{nB}}.
\end{equation}
Thus, combining this with (\ref{eq:E21}) and (\ref{eq:E22}), finally gives us that $E_2 \leq \frac{\varepsilon}{2}$.
\par
We conclude that there exist cube roots of unity $\omega_i$ such that $$||\omega_iu_i - l_i|| \leq E_1 + E_2 < \varepsilon.$$
This proves that Algorithm \ref{algo:Jennrich} indeed returns an $\varepsilon$-approximate decomposition of $T$.
\end{proof}

\end{document}